\documentclass[12pt,twoside,beltcrest]{ociamthesis} 

\usepackage{amssymb}
\usepackage[rightcaption]{sidecap}
\usepackage{graphicx}
\usepackage{subcaption}
\usepackage{amsthm,amsmath,amssymb,mathrsfs}
\usepackage{latexsym}
\usepackage[english]{babel}
\usepackage{verbatim}
\usepackage{color}
\usepackage{hyperref}
\usepackage{cancel}
\usepackage{cite}
\usepackage{dsfont}
\usepackage{comment}

\let\origdoublepage\cleardoublepage\newcommand
{\clearemptydoublepage}{
  \clearpage
  {\pagestyle{empty}\origdoublepage}
}

\newcommand{\bea}{\begin{eqnarray}}
\newcommand{\eea}{\end{eqnarray}}
\newcommand{\beq}{\begin{equation}}
\newcommand{\eeq}{\end{equation}}
\newcommand{\nn}{\nonumber}
\newcommand{\bal}{\begin{align}}
\newcommand{\eal}{\end{align}}
\newcommand{\bit}{\begin{itemize}}
\newcommand{\eit}{\end{itemize}}

\newcommand{\s}{\sigma}

\newcommand{\meas}{{\mu(\nu)}}
\newcommand{\rar}{\rightarrow}

\newcommand{\abs}[1]{\vert #1\vert}
\newcommand{\avg}[1]{\left\langle #1 \right\rangle}
\newcommand{\expect}[2]{\left\langle #1 \right\rangle_{#2}}

\newcommand{\half}{{\frac{1}{2}}}

\newcommand{\quarter}{{\frac{1}{4}}}
\newcommand{\third}{{\frac{1}{3}}}

\newcommand{\threehalves}{{\frac{3}{2}}}

\newcommand{\dha}{{d_H}}
\newcommand{\ds}{{d_s}}

\newcommand{\La}{\Lambda}

\newcommand{\C}[1]{{\mathcal{#1}}}

\newcommand{\R}[1]{{\mathrm{#1}}}

\newcommand{\BB}[1]{{\mathbb{#1}}}

\newcommand{\tr}{\mathscr{T}}
\newcommand{\trs}{\mathscr{F}}

\newcommand{\leaveout}[1]{}

\newcommand{\CGW}{{\mathrm{GW}}}

\newtheorem{theorem}{Theorem}
\newtheorem{lemma}[theorem]{Lemma}

\newtheorem{assume}[theorem]{Assumption}

\title{Spectral dimension in graph models of causal quantum gravity}     
\author{Georgios Giasemidis}             
\college{Kellogg College}                    

\degree{Doctor of Philosophy}     
\degreedate{Trinity 2013}             

\begin{document}

\baselineskip=21pt plus1pt

\setcounter{secnumdepth}{3}
\setcounter{tocdepth}{3}

\maketitle                                          
\clearemptydoublepage
\vspace*{\fill}
\begin{quote}

{\emph{The effort to understand the Universe is one of the very few things that lifts human
life a little above the level of farce, and gives it some of the grace of tragedy.
}}
\begin{flushright}
--- Steven Weinberg (1976) \cite{Weinberg:1993tf}
\end{flushright}

\end{quote}
\vspace*{\fill}
\clearemptydoublepage
\begin{abstractseparate}
The phenomenon of scale dependent spectral dimension has attracted special interest in the quantum gravity community over the last eight years. It was first observed in computer simulations of the causal dynamical triangulation (CDT) approach to quantum gravity and refers to the reduction of the spectral dimension from 4 at classical scales to 2 at short distances. 
Thereafter several authors confirmed a similar result from different approaches to quantum gravity. 

Despite the contribution from different approaches, no analytical model was proposed to explain the numerical results as the continuum limit of CDT. In this thesis we introduce graph ensembles as toy models of CDT and show that both the continuum limit and a scale dependent spectral dimension can be defined rigorously. First we focus on a simple graph ensemble, the random comb. It does not have any dynamics from the gravity point of view, but serves as an instructive toy model to introduce the characteristic scale of the graph, study the continuum limit and define the scale dependent spectral dimension. 

Having defined the continuum limit, we study the reduction of the spectral dimension on more realistic toy models, the multigraph ensembles, which serve as a radial approximation of CDT. We focus on the (recurrent) multigraph approximation of the two-dimensional CDT whose ensemble measure is analytically controlled. The latter comes from the critical Galton-Watson process conditioned on non-extinction. Next we turn our attention to transient multigraph ensembles, corresponding to higher-dimensional CDT. Firstly we study their fractal properties and secondly calculate the scale dependent spectral dimension and compare it to computer simulations. We comment further on the relation between Ho\v rava-Lifshitz gravity, asymptotic safety, multifractional spacetimes and CDT-like models.
\end{abstractseparate}

\clearemptydoublepage
\begin{acknowledgements}
Above all, I would like to thank my D.Phil. supervisor, Dr. John F. Wheater, for introducing me to this interesting area of research and for his constant guidance and assistance during the projects. This thesis would not have been completed without his critical insights. I am also thankful to him for being patient with the difficulties I faced in this field and for giving me the freedom to work on other projects too.

It is my pleasure to thank my two collaborators and friends Stefan Zohren and Miguel Tierz, with whom I had an enjoyable and fertile  collaboration. They both guided and advised me at different stages of the D.Phil. and have left a huge impact on my learning and research. They successfully passed on their enthusiasm about research and were always a source of motivation for me. Max Atkin also deserves a thank you for many clarifications and collaboration during my first year.

My sincere thanks also go to a very good friend and physicist, Nikos Kaplis, for the endless and fruitful discussions about physics and other topics as well. 
I also enjoyed the friendly environment of my office which I consecutively shared with Andrei Constantine, Maxime Gabella, Cyril Matti, Niels Martens, Neil Robinson and Curt von Keyserlingk, who I also thank.

Furthermore, I feel indebted to my family, for their mental and practical support from the very beginning of my academic years. It is my parents' financial generosity that fully supported me during my first year when no other source of funding was available. Without their financial backing I would have never taken the decision to start this program. 

Regarding funding, I would like to acknowledge the scholarships by the A. S. Onassis Public Benefit Foundation grant F-ZG 097/ 2010-2011 and A. G. Leventis Foundation for covering my main expenses and university/college fees respectively from 2010 to 2013. 

Foremost, at least a big thank you to Vicky, she knows why.

\vspace{2cm}

To Charalampos, Georgia, \\
Konstantina, Popi, Nikos, Aggelos,  \\
and Vicky.
\end{acknowledgements}
\clearemptydoublepage

\begin{romanpages}                       
\tableofcontents                               
\listoffigures                                      
\end{romanpages}                          

\begin{chapter}{Introduction}

One of the most challenging and ongoing projects in theoretical physics is the problem of reconciling general relativity with quantum mechanics, known as quantum gravity. 

On one hand, our understanding about gravity is underlined in general relativity, which describes gravitational phenomena at large scales, from solar-sized systems to (clusters of) galaxies. General relativity provides us with a powerful mathematical framework to further describe how the universe evolved and galaxies were formed with great experimental accuracy. Within this framework gravity is the space-time itself.  

On the other hand, quantum mechanics explains how particles interact at sub-atomic level.  The quantum theory not only gave us a new mathematical framework, but also a different conceptual interpretation of nature at small scales. Quantum mechanics as formulated in Schr\"odinger's equation is non-relativistic, since time is treated as an external parameter. Its relativistic extension, quantum field theory (QFT), met with extensive success in many fields of physics, most importantly in particle and condensed matter physics. 

Quantum gravity therefore aims to understand how gravity ``behaves" at tiny scales, where it was soon clear that general relativity is not well-behaved under conventional QFT methods \cite{'tHooft:1974bx}.  It is not only a difficult mathematical problem but a conceptual one as well, because physicists have to understand the ``nature'' of space-time at the smallest scales. Since then, several approaches to quantum gravity have been proposed, each starting from different first principles of theoretical physics (for an excellent historical review on the subject consult \cite{Rovelli:2000aw, Rovelli:2004qg}). The difference usually resides on the degrees of freedom to be quantised. For example, in string theory, the fundamental degrees of freedom are string-like configurations, instead of the metric field, which are quantised using the postulations of QFT.  
Others, e.g. loop quantum gravity, consider the metric field of general relativity on equal footing with quantum mechanics and it is the metric field which is quantised. Other non-perturbative formulations include the lattice regularisation (e.g. dynamical triangulation), the asymptotic safety conjecture and many more. 
However, none of these approaches is conclusive simply because quantum gravity effects are believed to become measurable close to the Planck scale, an energy scale of order $10^{15}$ times bigger than current scales probed at the LHC
\footnote{This statement might be placed in contrast to analogue models, where quantum gravitational effects can be imitated by condensed matter systems \cite{Visser:2007du}.}.
%

Despite the diverse proposals, string theory and loop quantum gravity have gained considerable popularity and dominated the research debate on the subject. For example, string theory unveiled a rich mathematical structure, with lots of symmetries and dualities, which attempts to solve the problem of quantising gravity in a broader unified manner including the standard model of particle physics. Despite its beauty, the formulation comes with a high cost; string theorists have to reconsider the dimensionality of the world, which might be ten or eleven dimensional \cite{Zwiebach:2004tj}
\footnote{We choose on purpose to focus on and mention only the higher-dimensional space-time, ignoring other assumptions of the theory, like supersummetry.}. 
That is, our four-dimensional universe is embedded into a higher-dimensional geometry. The six or seven dimensions are still present in our world but are so tiny and curled up (compactified) that they have not actually been detected by any experimental measurement. The idea of extra dimensions is not new in physics and goes back to the pioneering work of Kaluza and Klein who tried to unify gravitation and electromagnetism by considering a five-dimensional space-time \cite{Kaluza:1921, Klein:1926}. Another modern idea for a highly-wrapped fifth dimension is considered in the so-called Randall-Sundrum models \cite{Randall:1999ee}.  The conclusion of the discussion so far is that although our world looks four-dimensional macroscopically, the description of quantum gravity phenomena might alter the number of dimensions microscopically. 

All the above examples propose a higher-dimensional universe. However, there are also theoretical proposals which indicate a lower-dimensional space-time at very high energies. For example, black hole physics is an area where quantum gravity effects become important and our standard notion of dimensionality might be reviewed too. 't Hooft suggested the possibility of \textit{dimensional reduction} in a black hole, where the observable degrees of freedom can be interpreted as Boolean numbers on a $2+1$-dimensional lattice \cite{'tHooft:1993gx}. A second example comes from the asymptotic safety scenario line of research, where the authors reported that the existence of a non-Gaussian UV fixed point for gravity modifies the graviton propagator at Planck scales in a way which imitates a two-dimensional theory \cite{Lauscher:2001ya}. Then they suggested that space-time undergoes a dimensional reduction from $4$ at large scales to $2$ at the Planck scale
\footnote{We review the asymptotic safety scenario and the argument for dimensional reduction in chapter \ref{physics}.}.

In order to be as rigorous as possible we need to adopt new definitions for the effective dimension of a geometry. Such definitions are commonly used in the discretised models of quantum geometry \cite{Ambjorn:1997}, and are known as the spectral and the Hausdorff dimensions (to be defined later). The former corresponds to the dimension that a random walker experiences in a diffusion process. In 2005, Ambj{\o}rn, Jurkiewicz and Loll studied diffusion on four-dimensional discretised geometries, defined according to the Causal Dynamical Triangulations (CDT) approach to quantum gravity, and they found that the spectral dimension varies from $4$ in the classical limit to $2$ at scales where quantum effects should be important \cite{Ambjorn:2005db}.  
Following this result, Launcher and Reuter computed the spectral dimension in the context of asymptotic safety and also found a reduction of the spectral dimension from $4$ to $2$ \cite{Lauscher:2005qz}. The running of the spectral dimension seems to serve as a dynamical mechanism that might regulate general relativity at short scales so that the theory possesses a non-Gaussian fixed point.

From this point onwards, many other approaches of quantum gravity reported that the spectral dimension is not fixed but varies with the scale \cite{Carlip:2009kf, Carlip:2009km}. We should underline that this change happens dynamically and no two-dimensional structure is embedded into the four-dimensional space-time. Therefore this kind of dimensional reduction is in contrast to string theory compactifications of the higher-dimensional geometry. This phenomenon is termed as the \textit{scale-dependent} spectral dimension or \textit{dynamical dimensional reduction}.

Despite the accumulation of results for dynamical dimensional reduction, the first evidence coming from numerical results on four-dimensional CDT had a significant drawback;  the outcome could be an artefact of computer simulations. Since then, little progress has been made in analytically understanding the numerical results coming from the CDT approach and showing that they remain valid when taking the continuum limit. The central goal of this thesis is twofold. First, we intend to define the continuum limit of discretised objects and show how this definition accommodates a scale dependent spectral dimension. Second we aim to understand the mechanism that is hidden behind the numerical results of higher-dimensional CDT.

Our methodology combines ideas from quantum geometry, statistical physics and graph theory and consists of new mathematical techniques in the study of quantum gravity, e.g. random walks on random graphs, branching stochastic models (Galton-Waltson processes) and probability.

This thesis is organised as follows. In chapter \ref{motivation}, we give the motivation for conducting the current research. Since it is important to understand the background, we review the CDT approach to quantum gravity and its main results, focusing on the scale dependent spectral dimension, due to its importance to the rest of the thesis. 
We devote chapter \ref{graphs} to presenting definitions and notions that will be essential in understanding the following computations. We first give an introduction to graphs and graph ensembles, spectral and Hausdorff dimensions of graphs and elaborate on the generating function techniques for extracting the spectral dimension of graphs. We then discuss the link between gravity models and graph ensembles and concentrate on a special class of graphs, the Galton-Watson trees. After these two introductory chapters, we are now equipped appropriately to present the novel computations.

Chapter \ref{combs} is based on article \cite{Atkin:2011ak}. We start our research for scale dependent spectral dimension on graph ensembles working with a simple class of trees, the combs. After reviewing basic facts about combs, we formulate the continuum limit of such objects which exhibit  a spectral dimension which varies with the scale of the diffusion. Next, we apply this formulation to three comb ensembles and show that it does qualitatively mimic the scale dependent spectral dimension phenomenon observed in CDT. The first two toy models, the simple measure and the power-law measure, exhibit one scale where the spectral dimension changes, whereas the third toy model random comb presents two characteristic scales. In the latter model the spectral dimension exhibits an intermediate plateau of constant spectral dimension, the apparent spectral dimension. We conclude this chapter by commenting on the reasons we should go beyond random combs and work with more realistic graph models. 

The content of chapter \ref{multigraphs} originates from article \cite{Giasemidis:2012rf} and proceedings \cite{Giasemidis:2012pu}. Following the conclusions of the previous chapter, we proceed to multigraph ensembles. We begin with the relationship among (two-dimensional) causal triangulations, Galton-Watson trees and (recurrent) multigraphs and present basic properties of the latter. First, we apply the continuum limit formalism, developed in the previous chapter, to a particular recurrent multigraph ensemble, which exhibits dynamical reduction of the spectral dimension from 2 at large distances to 1 at small scales. We comment on the physical interpretation of this result. Next, we continue with the transient multigraph ensembles, which are considered as ``radial" approximations of higher-dimensional CDT. Before applying the continuum formalism, we explore its properties and actually derive an important relation between the spectral, Hausdorff dimensions and the anomalous exponent of graph resistance. In the absence of analytical results of higher-dimensional CDT, we introduce a few assumptions which determine the measure of transient multigraphs. We justify the origin of those ansatz and apply the continuum limit which results in a running spectral dimension varying from 4 in the IR to 2 in the UV limit, as observed in CDT simulations. 

Chapter \ref{physics} includes material that first appeared in articles \cite{Giasemidis:2012qk} and \cite{Giasemidis:2012pv}. In essence, this chapter discusses the physical implications of our computations so far. We start by extending the continuum formalism and derive the return probability density of the four-dimensional model considered in the previous chapter. The congruence with numerical simulations assures us that both the multigraph approximation and the assumptions are correct and reflect the correct degrees of freedom which are responsible for the reduction of the spectral dimension in four-dimensional CDT. Beyond this, we also consider the multigraph approximation of three-dimensional CDT. Our results reproduce the reduction from 3 in the IR to 2 in the UV as is apparent in the Monte-Carlo simulations of three-dimensional CDT, but seemingly disagree with the applied fit to the data. However, studying the data in more detail, we show that our model fits well. We devote the last section of the chapter to present a plethora of evidence for dynamical dimensional reduction from other proposals of quantum gravity. We emphasise the similarities, potential connections, but also the (fundamental) differences amongst them. 

Finally, we summarise and  conclude in chapter \ref{conclusions}. In appendices \ref{Appendix_combs} and \ref{Appendix_multigraphs} we provide complementary material regarding chapters \ref{combs} and \ref{multigraphs} respectively. 

\end{chapter}
\clearemptydoublepage
\begin{chapter}{Motivation}
\label{motivation}

This chapter serves as a motivation for the research presented in this thesis. In order to present the results in a self-contained way  we start by briefly discussing the ``sum over all histories" approach to quantum gravity and give reasons for a lattice regularisation of the theory, the ``dynamical triangulations'' (DT) method. We explain why DT fails to serve as a four-dimensional theory of quantum gravity. We then review the \textit{Causal Dynamical Triangulation} (CDT) approach to quantum gravity and discuss the consequences of the theory, e.g. its phase diagram and the phenomenon of \textit{scale dependent spectral dimension}. The latter is at the core of this thesis and the main motivation for what follows. 

\section{The ``sum over all histories" approach to quantum gravity}

In a very instructive paper \cite{Misner:1957fq} in 1957, Charles Misner outlined the possible approaches one has to follow in order to ``attack" the quantum gravity problem. Among those approaches, he introduced and elaborated on the possibility of defining quantum gravity through the functional integral formalism which had been completed and popularised by Richard Feynman working on quantum electrodynamics. According to C. Misner  \cite{Misner:1957fq} ``\textit{the problem of formulating the Feynman quantisation of general relativity was originally suggested by Professor J. A. Wheeler, whose idea was simply to write }
\beq \label{Wheelers_suggestion}
\int \exp\left ((i/\hbar) (\text{Einstein action}) \right ) d (\text{field theories})\text{''}.
\eeq

According to the path integral formalism of quantum mechanics, one has to take into account all possible histories that a system follows between the initial and final states and weight every history with a phase which is proportional to the classical action of the theory. This short description of path integral formalism justifies Wheeler's suggestion \eqref{Wheelers_suggestion} as a possible way to quantise the gravitational field. This is the so called ``\textit{sum over all histories}" line of research which attempts to define the theory. 

From the beginning, C. Misner realised quickly that such an object is difficult to handle. Later, 't Hooft and Veltman, in their seminal paper \cite{'tHooft:1974bx}, showed that gravity coupled to matter is perturbatively non-renormalisable at one-loop level, with pure gravity being renormalisable, and a later result by Goroff and Sagnotti \cite{Goroff:1985sz}  confirmed that pure gravity is non-renormalisable too at two-loop level. To tame the divergences of quantum gravity, Stelle considered higher-derivative theories of gravity. The latter were perturbatively renormalisable but non-unitary \cite{Stelle:1977hd}
\footnote{A modern treatment which remedies both the power-counting renormalisability and the unitarity problems has been proposed by Ho\v{r}ava who introduced gravity models with anisotropic scaling between space and time \cite{Horava:2009uw}. We will review this approach in chapter \ref{physics}.}.
Perturbative non-renormalisability led Weinberg  to the conjecture that gravity might have a non-Gaussian fixed point, the so called \textit{asymptotic safety scenario} \cite{Weinberg:1979ud}. At the same time, S. Hawking was advocating the \textit{path integral approach to quantum gravity} \cite{Hawking:1978jz, Hawking:1979pi, Hawking:1993eq}, working in the Euclidean signature to avoid some of the pathologies of the path integral in the Lorentzian signature.    

\section{A lattice regularisation of the path integral - Dynamical Triangulations}
Attempts to study the gravitational path integral went on, but it was now clear that a non-perturbative way to regulate the theory was needed. This is provided by the \textit{lattice regularisation}. This research program started with the Euclidean counterpart of \eqref{Wheelers_suggestion} defined as \cite{Hawking:1978jz, Hawking:1979pi, Hawking:1993eq} \beq \label{Eu_path_int} 
\hat Z = \int _{M} \C{D} [g_{\mu \nu}] e^{-\hat S_{EH}[g_{\mu \nu}]}
\eeq
where $M$ is a four-dimensional Riemannian manifold, $[g_{\mu \nu}]$ the diffeomorphism equivalent classes of metrics, and $\hat S_{EH} =  \frac{1}{16 \pi G_N}\int _{M} d^4y \sqrt{g} \left (2\Lambda - R(y) \right )$ is the Euclidean Einstein-Hilbert action. Lattice regularisation refers to the process of discretising the space-time with $N$ \textit{triangles} (for $d=2$), or higher simplices (for $d>2$), which are the building blocks of the lattice and have lattice spacing $a$, which serves as a regulator for the theory. Intuitively speaking, to recover the continuum space-time, we have to perform the \textit{continuum limit}, which corresponds to simultaneously taking the number of building blocks to infinity, $N\to\infty$, and the lattice spacing to zero, $a \to 0$, by keeping the total volume fixed. In contrast to other lattice field theories, e.g. the Ising model and QCD, where the lattice is fixed, in the current theory of gravity, the lattice is the space-time itself and  therefore is a \textit{dynamical triangulation} 
\footnote{We adopt the terminology ``triangulation'' to refer to the higher-dimensional, $d>2$, lattices too, which are composed of tetrahedra and four-simplices.}.
Within this research program, which is called dynamical triangulation (DT), the gravitational path integral \eqref{Eu_path_int} is approximated by a sum over all triangulated geometries, $\tr$, and takes the form
\beq \label{Z_triang}
\hat Z_{DT} = \sum _{\tr \in \trs} \frac{1}{C(\tr)} e^{-\hat S_{ER}(\tr)}
\eeq
where $\trs$ is the set of all triangulations, $C(\tr)$ is the order of the automorphism group of  $\tr$ and $\hat S_{ER}(\tr) = \lambda N_4 (\tr) - \nu N_2(\tr)$ is the four-dimensional Euclidean Einstein-Regge action. Here $\lambda$ and $\nu$ are the bare cosmological constant and bare inverse Newton's constant respectively. In Einstein-Regge action $N_4$ is the total number of four-simplices in the triangulation, which corresponds to the volume term, and $N_2$ the total number of triangles, where curvature resides (for more details on Regge calculus see \cite{Regge:1961px, Ambjorn:2012jv}). 

One of the main observables in this approach to quantum gravity is the dimensionality of the geometry. For example, the Hausdorff, $\dha$, and the spectral, $\ds$, dimensions probe different characteristics of the quantum geometry and unveil its fractal properties. The former is related to the volume growth of a ball of radius $R$, $|B_R| \sim R^{\dha}$, for large $R$, whereas the latter is related to a diffusion process on the quantum geometry (for a rigorous definition see section \ref{sdsd}).  

DT in two dimensions is analytically tractable, since the curvature term contributes a constant due to the Gauss-Bonnet theorem, and it was found that the Hausdorff dimension is $\dha= 4$; a value which seems ``unnatural'' and indicates that the quantum geometry is fractal \cite{Kawai:1993cj, Ambjorn:1995dg, Ambjorn:1995rg, Ambjorn:1998fd}. The source of the problem originates from the appearance of the so called baby universes and their domination in the continuum limit \cite{Ambjorn:1998xu,  Ambjorn:1998fd}.  
As a result of the proliferation of the baby universes the Hausdorff dimension increases from 2, which one would expect ``naturally", to 4. However, the spectral dimension of two-dimensional Euclidean gravity is $\ds = 2$, which demonstrates that the quantum geometry retains some two-dimensional features~
\footnote{This is proven in the physics literature under a few physical assumptions and has also been observed in computer simulations \cite{Ambjorn:1997jf, Ambjorn:1998fd}. A rigorous mathematical proof is missing. However, it was recently proven in \cite{Gurevich:2012rp} that the uniform infinite planar triangulation is recurrent, hence $\ds \leq2$.} 
\cite{Ambjorn:1997jf, Ambjorn:1998fd}.

In dimensions $d>2$, DT was studied mainly through computer simulations due to the lack of analytical tools \cite{Ambjorn:1991pq, Agishtein:1991cv, Ambjorn:1997}. In particular, in these higher-dimensional models of Euclidean quantum gravity one has to further explore the phase diagram of the theory and search for  values of the parameters where critical behaviour can be found. For example, looking at \eqref{Z_triang} in four-dimensions, there is a critical line $\lambda _c (\nu)$ such that for values $\lambda > \lambda _c (\nu)$ the partition function is well-defined, whereas for values $\lambda < \lambda _c(\nu)$ the partition function is divergent%
\footnote{This analysis is valid only under the assumption that the number of triangulations with fixed number four-simplices, $N_4$, is exponentially bounded. This assumption was confirmed by computer simulations \cite{Ambjorn:1994tza}.}. 
Critical behaviour can be found at the boundary, i.e. as $\lambda$ approaches $\lambda _c(\nu)$ from above, where the infinite volume limit, $N_4 \to \infty$, is achieved (figure \ref{4dDT_phase_diagram}). 
\begin{SCfigure}
\includegraphics[scale=0.5]{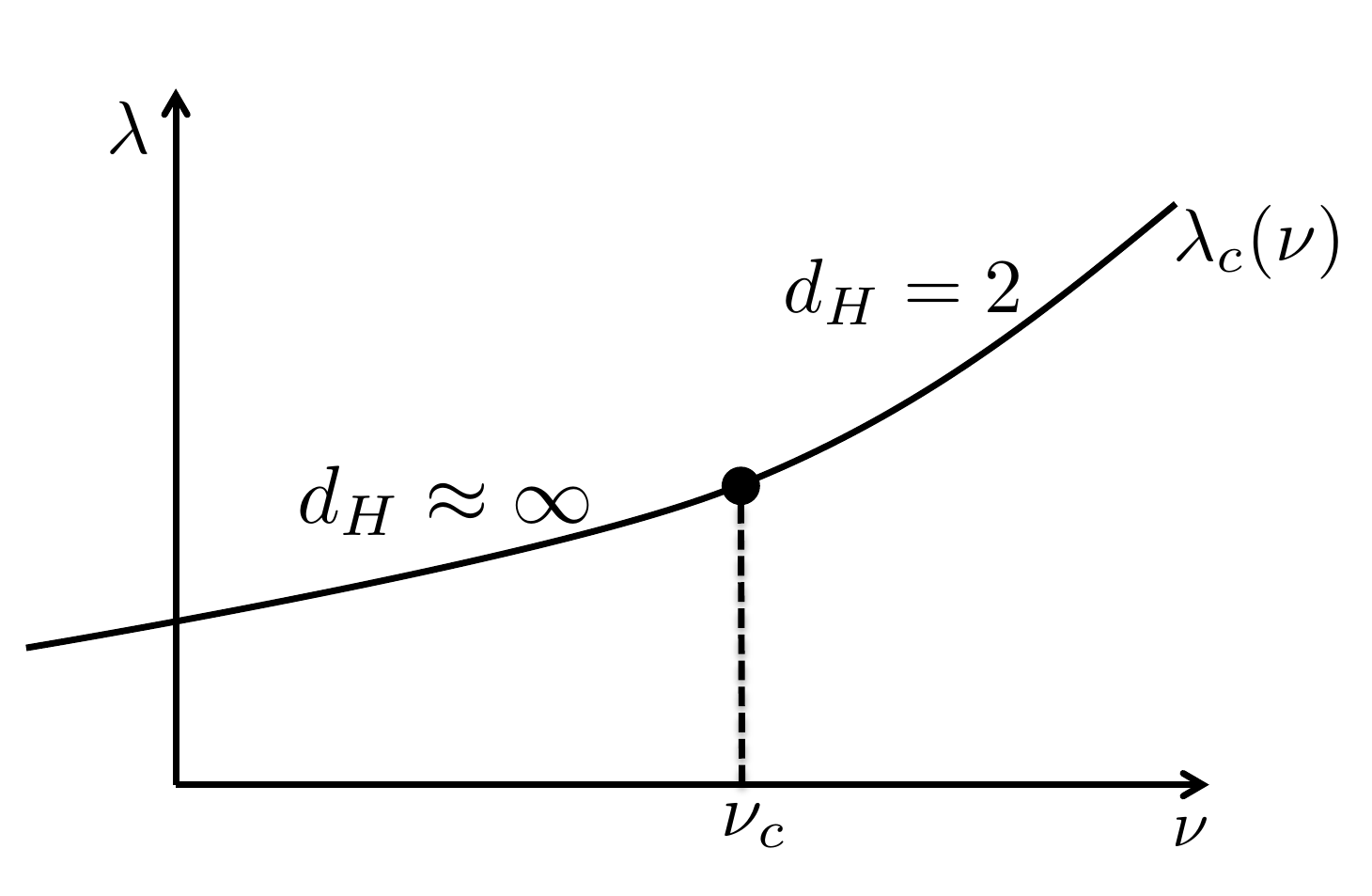}
\caption{A quantitative phase diagram of four-dimensional DT. The critical line is approached as $\lambda \searrow \lambda_c(\nu)$ from above \cite{Ambjorn:1997}.}
\label{4dDT_phase_diagram}
\end{SCfigure}
Along this line, there is a critical point $\nu_c$ which separates the two phases of the theory. For values $\nu < \nu_c$, the geometry is in the \textit{crumpled phase}, which is characterised by high connectivity and large Hausdorff dimension, i.e. $\dha \approx \infty$. For values $\nu > \nu_c$ along the critical line, the quantum geometry enters the \textit{branch polymer phase} which is dominated by branch polymer-like configurations with $\dha = 2$ \cite{Agishtein:1991cv, Ambjorn:1991pq, Ambjorn:1997}.  
Studying the phase transition at the point $\nu_c$, evidence was found  that it is of first order, a result which ruled out the possibility of defining a continuum limit of an extended geometry with finite Hausdorff dimension \cite{Bialas:1996wu, deBakker:1996zx}. 

In conclusion, even though DT started as a lattice approach to study non- perturbative aspects of non-critical bosonic string theory, it was later realised that it could also be a mechanism for regulating the four-dimensional theory and be a candidate method for quantising gravity. However, as we saw the Euclidean version of the model appears to fail to reproduce basic features of classical gravity, like the dimensionality and the emergence of a de-Sitter-like classical geometry.

\section{Causal Dynamical Triangulations}
As we already commented in the previous section, the main reasons that the DT program failed were \textit{i)} the absence of any Lorentzian feature in the theory, \textit{ii)} the existence and dominance of baby universes at the cut-off scale.  
To remedy this, J. Ambj\o rn and R. Loll,  introduced the Causal Dynamical Triangulation approach to gravity \cite{Ambjorn:1998xu}, which is a lattice definition of the gravitational path integral in the \textit{Lorentzian} signature, discretised by triangulations which have \textit{causal} structure, i.e. a \textit{global time foliation}. In particular, the basic ingredient is the formal functional integral of the metric field defined in the spirit of \eqref{Wheelers_suggestion}
\beq \label{Lor_path_int}
Z = \int _{M}\C{D} [g_{\mu \nu}] e^{i S_{EH}[g_{\mu \nu}]}.
\eeq
Following similar arguments as above, the triangulated counterpart takes the form
\beq
Z_{CDT} = \sum _{\tr \in \trs_c} \frac{1}{C(\tr)} e^{i S_{ER}(\tr)},
\eeq
where the sum is now restricted only on the set of triangulations with causal structure, $\trs_c$. A $d$-dimensional causal triangulation consists of spatial hyper-surfaces of fixed topology $\Sigma ^{d-1}$ which are triangulated by equilateral $d-1$ simplices with link length $a_s$, and labelled by discrete proper time $t_n$. Successive spatial hyper-surfaces are connected by $d$-simplices such that they are arranged in layers  (see figure \ref{causal_4simplices}). 
\begin{figure}
\begin{subfigure}[b]{.5 \textwidth}
\includegraphics[scale=0.3]{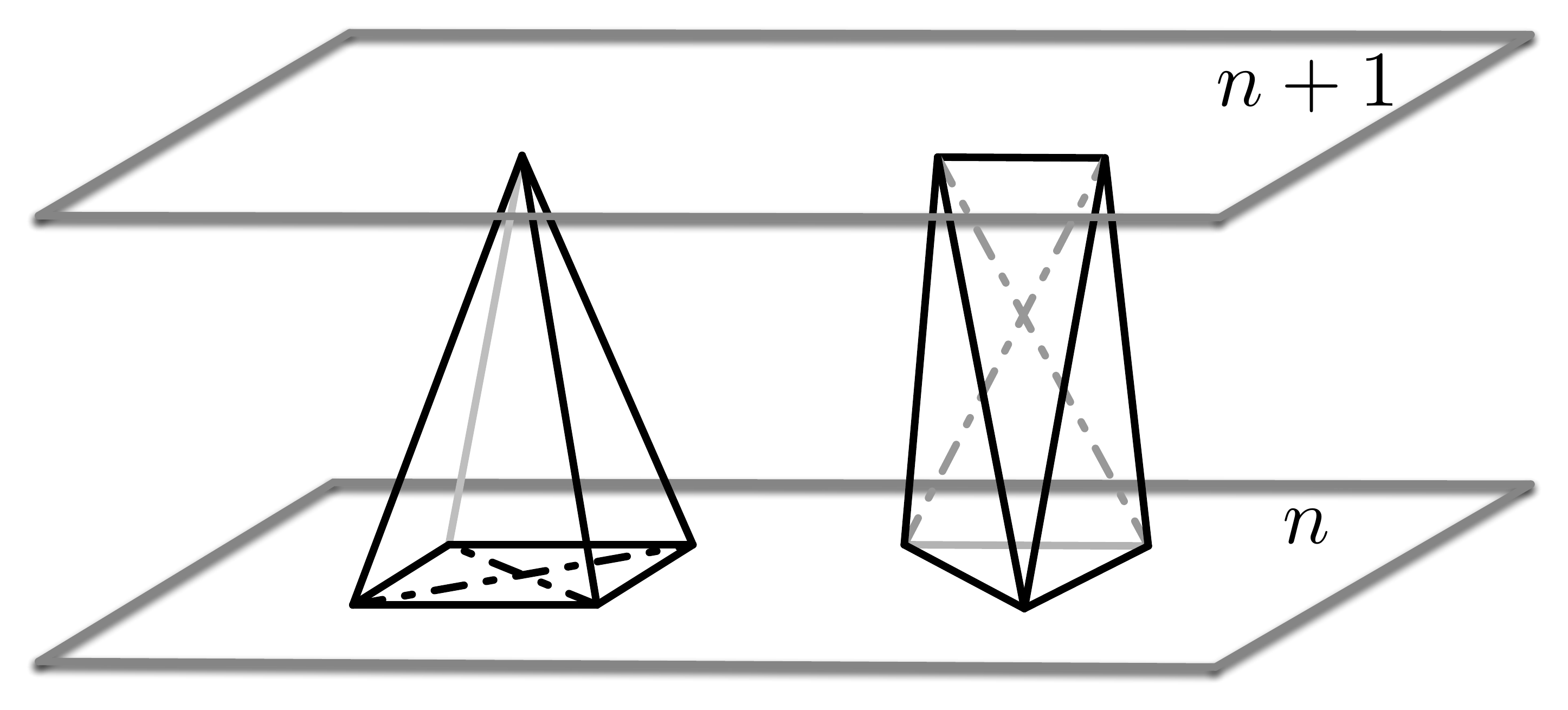}
\caption{The two building blocks of four-dimensional CDT; The (4,1) and (3,2) simplex are on the left and right respectively. Dotted lines correspond to three-simplices.}
\label{causal_4simplices}
\end{subfigure} \qquad
\begin{subfigure}[b]{.5 \textwidth}
\centering
\includegraphics[scale=0.2]{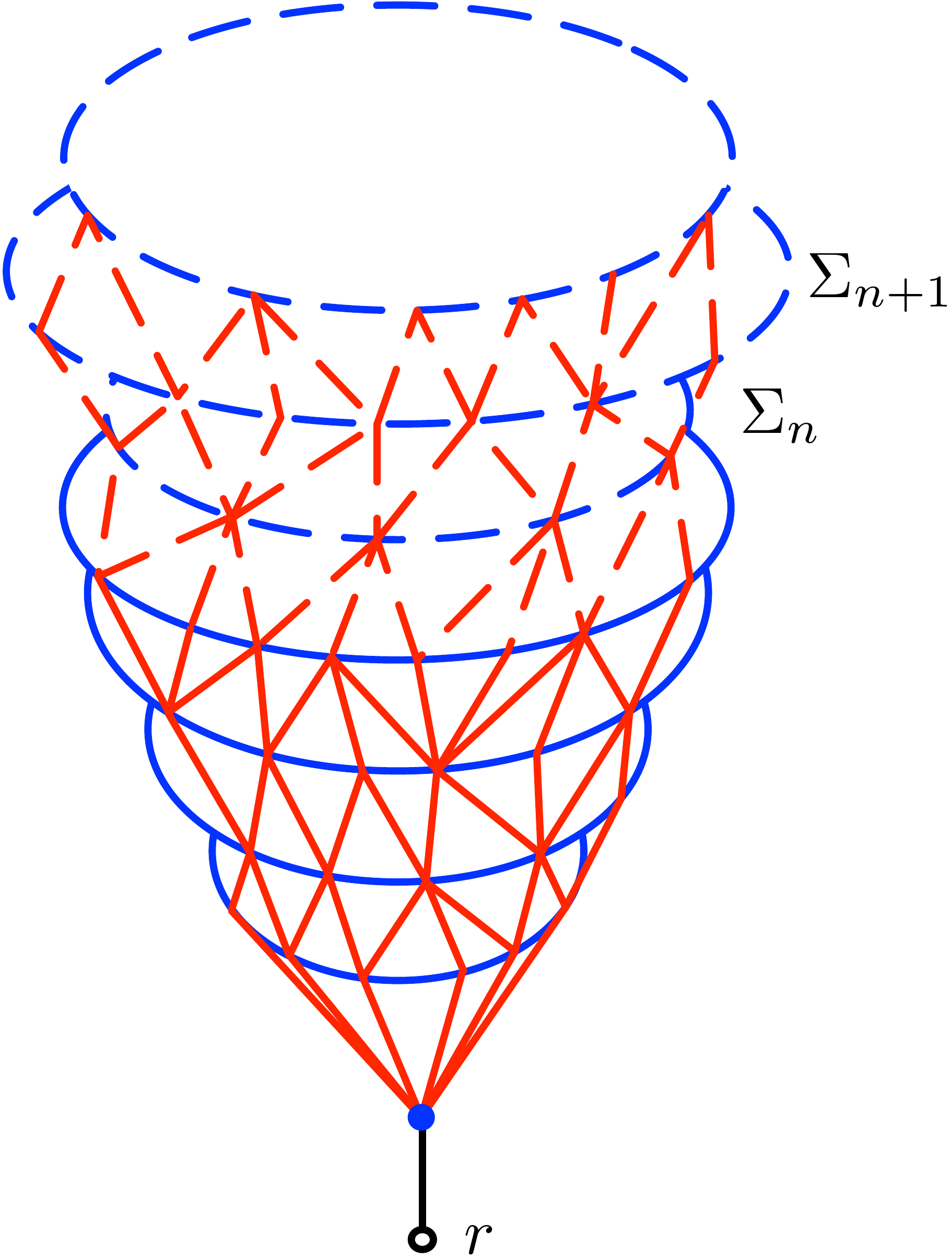}
\caption{A (rooted) two-dimensional causal triangulation with cylindrical topology. Spatial hyper-surfaces $\Sigma _n$ are depicted with blue line, whereas time-like links corresponds to red lines.}
\label{2d_cyl_causal_triangulation}
\end{subfigure}
\caption{The causal structure of a triangulated geometry.}
\label{causal structure}
\end{figure}
The total topology of the manifold is $I \times \Sigma^{d-1}$, which defines a global proper time foliation. From this construction, one observes two type of links; space-like links which lie in the spatial sectors of the triangulation and have square lattice length $a_s^2$, and time-like links which connect adjacent hyper-surfaces, and have lattice square spacing $a_t^2 = \alpha a^2_s$, for some relative scaling parameter $\alpha < 0$ (figure \ref{2d_cyl_causal_triangulation} shows a two-dimensional analogue). 

The importance of the causal assumption is twofold. Firstly, it forbids the formation of baby universes remedying the defects of DT. Secondly, it makes the Wick rotation from Lorentzian to Euclidean signature well defined for every triangulation, even though such rotation is not generally defined for continuum geometries. It is equivalent to a sign changing $\alpha \to - \alpha$ in the complex lower-half plane, i.e. $\sqrt{-\alpha} = - i \sqrt{\alpha}$, so that $a_s >0$ and $a^2_t = |\alpha| a^2_s >0$  
\footnote{Note that, after Wick-rotation, the Euclidean simplices put some lower bounds on the values of the parameter $|\alpha|$ \cite{Ambjorn:2012jv}.}.  
Under this change, the Einstein-Hilbert action is analytically continued to its Euclidean counterpart by $i S_{EH}(\alpha) \to - \hat S_{EH} (-\alpha)$, leading to
\beq \label{Euc_causal_path_int}
\hat Z_{CDT} = \sum _{\tr \in \trs_c} \frac{1}{C(\tr)} e^{- \hat S_{ER}(\tr)}.
\eeq
Having determined a well defined way to Wick-rotate the partition function to Euclidean signature, the problem has been reduced to a statistical physics problem which can been studied through Monte-Carlo simulations in $d>2$. 

As in the DT program, one has to search for critical behaviour in the phase-diagram of the theory \cite{Ambjorn:2005qt, Ambjorn:2010hu, Kommu:2011wd}. Taking the infinite volume limit by fine-tuning $\lambda$ to its critical value, the phase diagram of CDT exhibits three phases (figure \ref{CDT_phase_diagram}). 
\begin{figure}
\begin{center}
\includegraphics[scale=0.6]{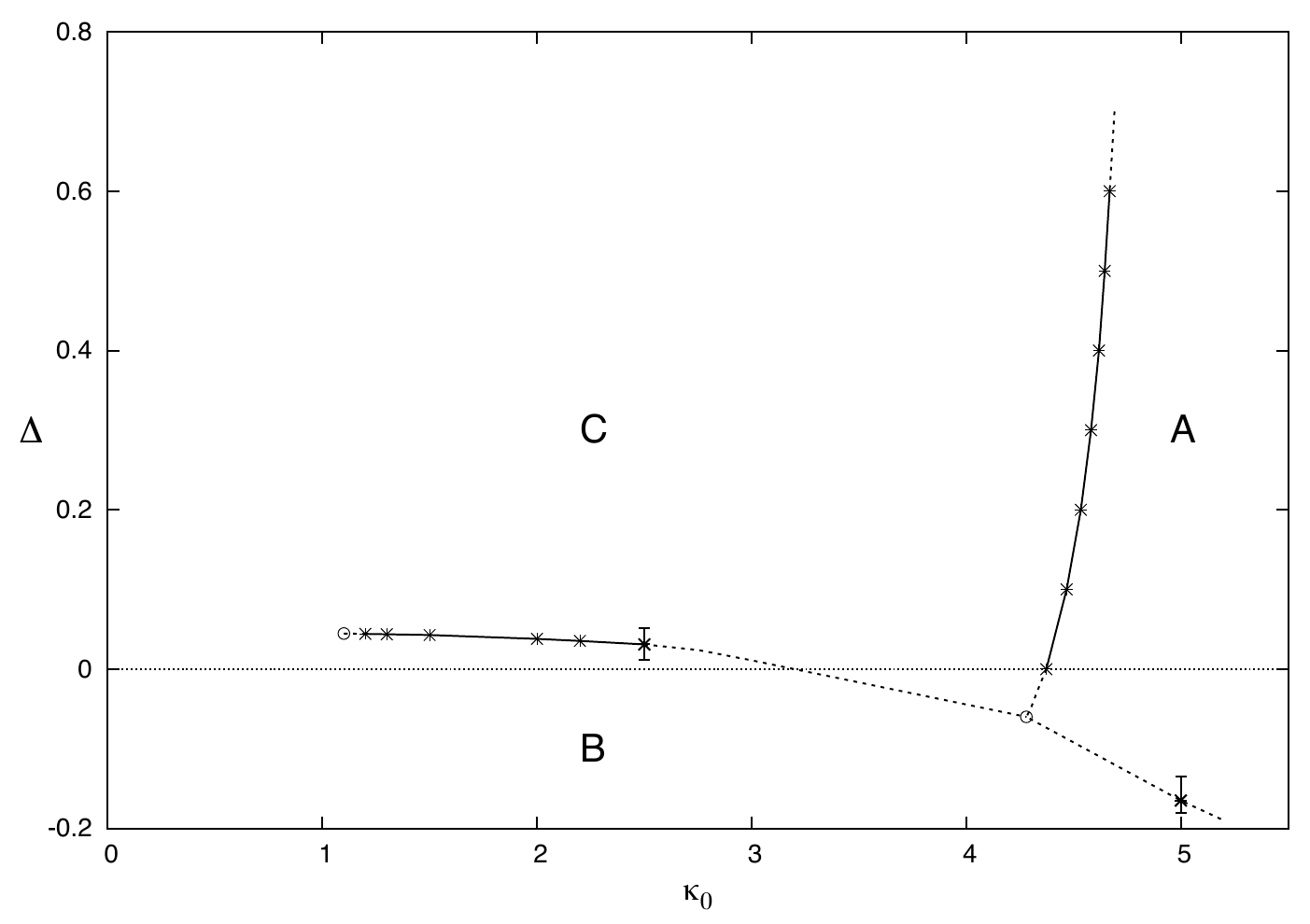}
\caption{Phase diagram of four-dimensional CDT in terms of the coupling constant $\kappa_0$, which replaces $\nu$, and the asymmetry parameter $\Delta$, which encodes the dependance of the Euclidean action on the relative length $\alpha$ between space-like and time-like links; $\Delta=0$ for $\alpha=1$. The coupling $\lambda$ does not appear because it has been tuned to its critical value \cite{Ambjorn:2010hu}. 
}
\label{CDT_phase_diagram}
\end{center}
\end{figure}
Phases A and B are non-physical and are believed to be remnants of the branched polymer and crumpled phases of DT respectively. In addition, contrary to DT, the phase diagram possesses a new phase, where an extended four-dimensional de-Sitter-like geometry emerges completely from quantum fluctuations. The existence of the latter phase is crucially important since the theory succeeds in reproducing samples of the observed universe in the semi-classical limit \cite{Ambjorn:2004qm, Ambjorn:2011ph}. 

Furthermore, one can extract more information from the phase diagram. From the theory of critical phenomena, it is well-known that a continuum limit of a lattice theory is described by a quantum field theory at a second or higher order phase transition. In \cite{Ambjorn:2010hu, Ambjorn:2011cg, Ambjorn:2012ij}, the authors studied the transitions between phases A-C and B-C numerically and found strong evidence that the former is of first order, whereas the latter is of second order. The appearance of the latter allows us to define a continuum theory of quantum gravity at the transition line making contact to other non-perturbative formulations of quantum gravity, e.g. the asymptotic safety scenario and Ho\v rava-Lifshitz gravity, depending on whether the UV fixed point is isotropic or anisotropic respectively \cite{Ambjorn:2010hu, Ambjorn:2013tki}. 

To summarise this brief introduction, one would state that CDT is a \textit{non- perturbative} and \textit{background independent} approach to approximate the gravitational path integral \eqref{Wheelers_suggestion} through a lattice regularisation. To conclude, we add another significant outcome of the theory, which is the dimensionality of the physical phase C. Computers simulations showed that the Hausdorff dimension of the physical phase is $\dha=4$ which is the expected result for a smooth classical geometry \cite{Ambjorn:2013tki}. In addition, while studying diffusion on ensembles of four-dimensional causal triangulations in \cite{Ambjorn:2005db} and a few years later independently in \cite{Kommu:2011wd}, the authors found that the spectral dimension of the universe is 4 in the classical, (IR), limit, which is regarded as another justification for the theory. Surprisingly they found that the spectral dimension reduces to the value of 2 in the UV limit. This phenomenon, known as \textit{dynamical dimensional reduction} or \textit{scale dependent spectral dimension}, was the first result suggesting that one notion of dimensionality of the universe does not remain constant but changes with the scale. In the next section we concentrate on definitions and details of the spectral dimension since it is an essential part of our research motivation.

\subsection{Scale dependent spectral dimension}
\label{sdsd}
Before discussing the reduction of the spectral dimension, which is a feature of  diffusion on \textit{quantum} geometries, we will first introduce definitions and properties of the spectral dimension of a \textit{classical} geometry, following the discussion in \cite{Ambjorn:2005db, Benedetti:2009ge}.

\textit{Spectral dimension of a classical geometry}: 
A non-rigorous and intuitive statement about the spectral dimension is that it corresponds to the dimension that a \textit{random walker} ``experiences" in a \textit{diffusion process}. In mathematical terms, consider a $d$-dimensional Riemannian manifold $M$, $d$ being the topological dimension, with metric $g_{\mu \nu}$. The diffusion process in governed by the diffusion equation 
\beq \label{diffusion_eq}
\frac{\partial K_g(y,y_0,\s)}{\partial \s} = \Delta _g K_g(y,y_0,\s)
\eeq 
with initial condition localised at a point $y_0$
\beq \label{initial_condition}
K_g(y,y_0,\s=0) = \frac{\delta ^d(y-y_0)}{\sqrt{g(y)}}
\eeq 
where $\s$ is the \textit{diffusion time}, $K_g(y,y_0,\s)$, the \textit{heat kernel}, is the probability density for diffusion from point $y_0$ to $y$ in diffusion time $\s$ and $\Delta _g = g^{\mu \nu} \nabla _{\mu} \nabla_{\nu}$ is the Laplace - Beltrami operator. However, here we should clarify that for diffusion in $d=3$ spatial dimensions, \eqref{diffusion_eq} becomes the heat equation and $\s$ corresponds to the physical time. However for $d$ space-time dimensions, $\s$ is no longer the real time, but rather should be considered as a fictitious time. The role of the latter will become clear in the following lines. 

The heat kernel can be equally expressed in terms of eigenvalues $\lambda_i$ and eigenvectors $\phi _i (y)$ of the operator $\Delta _g$
\beq \label{probability_density}
K_g(y,y_0,\s) = \langle y| e^{-\s \Delta _g} |y_0\rangle =\sum _{i} e^{-\lambda _i \s} \phi _i (y) \phi ^*_i (y_0).
\eeq
From the heat kernel we further define the return probability density in diffusion time $\s$, termed as the \textit{heat trace}, by
\beq \label{return_probability_density}
P_g(\s) = \frac{\int _{M} dy {\sqrt{g}} K_g (y,y,\s)}{V_g}
\eeq 
where we normalised over the volume of the manifold $V_g = \int _{M} dy {\sqrt{g}}$. In terms of eigenvalues (spectrum), \eqref{return_probability_density} takes the form
\beq \label{heat_trace}
P_g(\s) = \frac{\sum_i e^{-\lambda_i \sigma}}{V_g}
\eeq
which justifies the name heat trace. Expression \eqref{heat_trace} tells us that only small eigenvalues, of order $\lambda_i \lesssim 1/\s$, contribute to the return probability density, whereas large eigenvalues are exponentially suppressed. This sets a relationship between the scale which is being probed by the diffusion process and the diffusion time. Another important result for the return probability density is that it can be expressed in terms of curvature invariants of the geometry, the so-called \textit{heat trace expansion}
\beq \label{heat_trace_expansion}
P_g(\s) = \frac{1}{(4 \pi \sigma )^{d/2} V_g} \sum _{n=0} ^{\infty} a_n \sigma ^n ,
\eeq
where $a_0 = V_g$, $a_1 = \frac{1}{6} \int _M dy \sqrt{g} R(y)$ and $a_n, \ n\geq 2$, include higher order terms in scalar curvature $R$, Ricci tensor, $R_{\mu \nu}$, and Riemann tensor, $R_{\mu \nu \rho \tau} $. Having defined the return probability density, the spectral dimension of the geometry $M$ is defined by
\footnote{Our notation should not cause any confusion. We denote by $\ds$ the spectral dimension of discrete objects, like simplicial geometry and graphs, where the diffusion time is discrete. $D_s (\s)$ denotes the scale dependent spectral dimension of continuum geometries and is a function of continuous diffusion time $\s$.}
\beq \label{ds_def}
D_s(\s) := -2 \frac{d \ln P_g(\s)}{d \ln \s},
\eeq
which in terms of the curvature invariants \eqref{heat_trace_expansion} reads
\beq \label{ds_expansion}
D_s(\s) = d - 2 \frac{\sum _{n=1} ^{\infty} n a_n \sigma ^n}{\sum _{n=0} ^{\infty} a_n \sigma ^n}.
\eeq
For infinite flat manifolds, $a_n =0$, for all  $n\geq 1$, therefore the value of the spectral dimension agrees with the topological dimension for all diffusion times. Additionally, for a generic geometry with non-zero curvature, one observes that for small diffusion times $D_s (\s \gtrsim 0) \approx d$, whereas for larger diffusion times, the diffusion process ``explores" larger ``neighbourhoods" of the starting point which results in experiencing the curvature effects, thus $D_s(\s >> 1) < d$. In addition finite volume effects also alter the value of the spectral dimension. In this case, for $\s >> V^{2/d}_g$, only the zero eigenvalue contributes to \eqref{heat_trace} and the return probability density tends to one, meaning a reduction of the spectral dimension. 

\textit{Spectral dimension of a quantum geometry}: 
Our discussion so far has been restricted to a classical level, in the sense that we considered diffusion on a given manifold. To consider diffusion on a quantum geometry, we have to take into account all possible configurations (histories) of the geometry according to the spirit of the ``sum over all histories" approach described in previous sections. This means that we have to take the ensemble average of the return probability density defined as a gravitational path integral
\beq \label{P_ensemble_average}
\langle P (\s) \rangle _Z = \hat Z^{-1} \int \C{D} [g_{\mu \nu}] e ^{-\hat S_{EH}(g_{\mu \nu})} P_g (\s).
\eeq
Now the spectral dimension of the quantum geometry is defined through \eqref{ds_def} by replacing $P_g(\s)$ with the ensemble average $\langle P (\s) \rangle _Z$. For diffusion in a quantum geometry, $\s$ corresponds to the scale at which diffusion process probes the quantum geometry. Thus, large diffusion times, i.e. $\sigma \to \infty$, probe the infrared (IR) characteristics of the geometry, while the $\s \to 0$ limit probes the ultraviolet (UV) features of it. 

In order to determine \eqref{P_ensemble_average} within the CDT framework we write the triangulated analogue of it in the spirit of \eqref{Euc_causal_path_int}
\beq \label{P_ensemble_average_discr}
\langle P (\s) \rangle _Z = \frac{1}{\hat Z_{CDT}} \sum _{\tr \in \trs_c} \frac{1}{C(\tr)}e ^{-\hat S_{ER} (\tr)} P_{\tr}(\s).
\eeq
To compute the spectral dimension one should evaluate \eqref{P_ensemble_average_discr} using the partition function \eqref{Euc_causal_path_int} and then take the infinite volume limit in which $\lambda$ is tuned towards its  critical value and $\nu$ is expressed in terms of the  inverse renormalized Newton's constant $1/G$. At present this is analytically out of reach but Monte Carlo simulations of random walks (discrete diffusion) on four-dimensional CDTs of fixed $N_4$ \cite{Ambjorn:2005db} yield a scale dependent spectral dimension given by 
\bea \label{DsJan}
D_s(\sigma)=  4.02-\frac{119}{54+\sigma}=
\begin{cases}
1.80 \pm 0.25, & \sigma\to 0,\\
4.02 \pm 0.1, & \sigma \to \infty,
\end{cases} 
\eea
where the three constants were determined from the data range $\sigma \in [40,400]$. This range was chosen to avoid discreteness effects on one side and finite volume/curvature effects on the other side (see also figure \ref{scale_dependent_ds}). In addition, quantum effects are believed to be negligible when the spectral dimension reaches its large scale (topological) value.
\begin{SCfigure}
\includegraphics[scale=0.7]{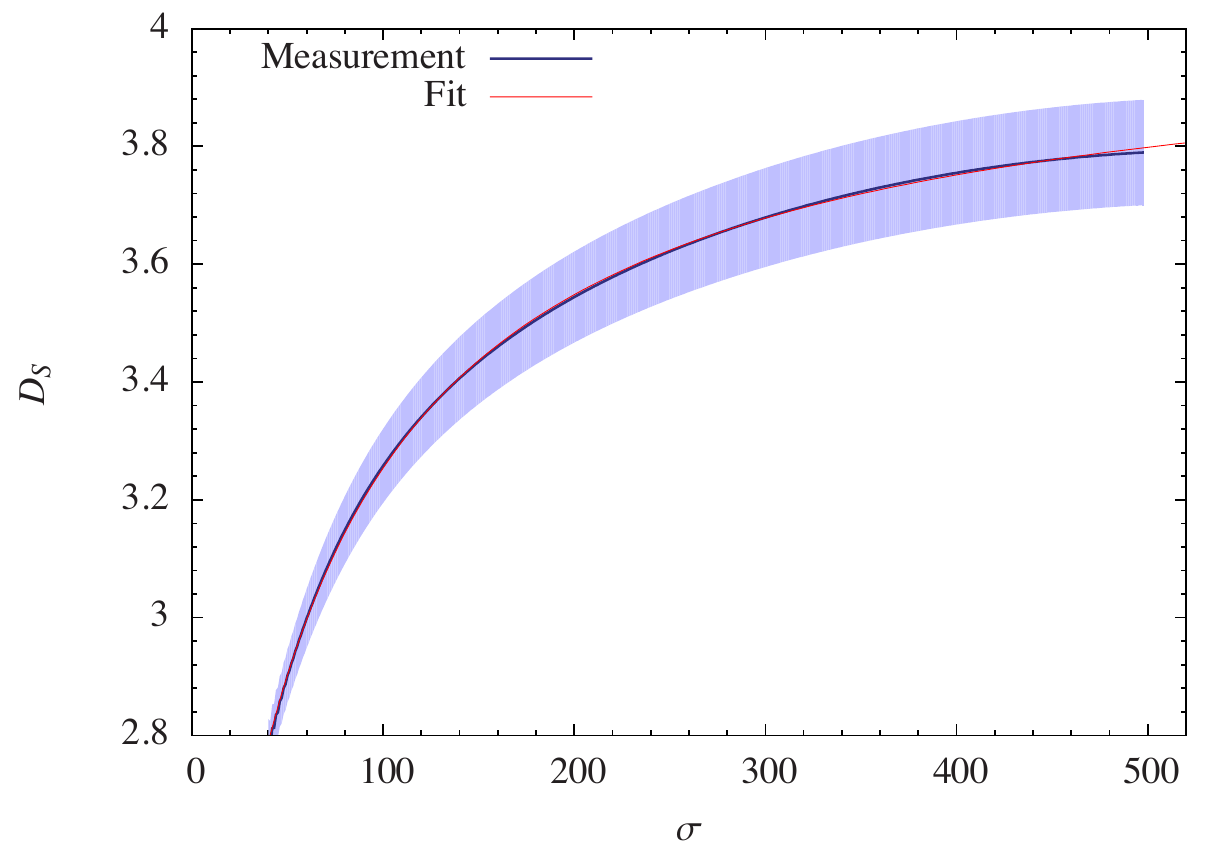}
\caption{The phenomenon of dynamical reduction of the spectral dimension in non-perturbative simplicial quantum gravity \cite{Ambjorn:2012jv}.}
\label{scale_dependent_ds}
\end{SCfigure}

A possible objection to \eqref{DsJan} is that the simulations are inevitably affected by finite size effects and the dimensional reduction observed might simply be an artefact of the discreteness scale. However, assuming that this expression can be extrapolated to continuum physics,  the return probability density \eqref{P_ensemble_average} for four-dimensional CDT in the \emph{continuum limit} \cite{Ambjorn:2005db} would be
\bea  \label{intro-P}
\avg{P(\sigma)}_Z \sim \frac{1}{\sigma^2}  \frac{1}{ 1 + const.\, G / \sigma},
\eea
where ``$\sim$" denotes equality up to multiplicative logarithmic corrections.

\section{Conclusion and outlook}
\label{motivation_outlook}
A quantitative confirmation of \eqref{DsJan} was given in \cite{Kommu:2011wd} using an independent computer code for the first time. Moreover, results from three-dimensional CDT \cite{Benedetti:2009ge} show qualitative agreement with the four-dimensional model. In particular, the classical limit of three-dimensional CDT gives rise to a de-Sitter-like space-time and the spectral dimension dynamically flows from the value of $3$ in the IR to the value of $2$ in the UV limit. However, we will postpone the details of the three-dimensional model for section \ref{3dim}.

Since the novel work in \cite{Ambjorn:2005db} and the first evidence of a scale-dependent spectral dimension, other authors have studied the behaviour of the spectral dimension and confirmed both quantitatively and qualitatively the results from Monte Carlo simulations. However their approaches were motivated by different models to quantum gravity, e.g. renormalisation group analysis \cite{Lauscher:2005qz} and Ho\v rava-Lifshitz gravity \cite{Horava:2009if} among others \cite{Carlip:2009km, Carlip:2009kf}. Despite the agreement with results from other models, little progress has been made in analytically understanding the numerical results coming from the CDT approach and showing that they remain valid when taking the continuum limit. This is the aim of this thesis. Based on work done in a series of published articles \cite{Atkin:2011ak, Giasemidis:2012qk, Giasemidis:2012rf, Giasemidis:2012pv} we introduce simplified graph toy models which capture features of CDT and exhibit the phenomenon of dynamical reduction of the spectral dimension. 

\end{chapter}
\clearemptydoublepage
\begin{chapter}{Preliminaries}
\label{graphs}

As explained in the previous chapter the main motivation for this thesis is to obtain an analytical understanding of the numerical simulations of the spectral dimensional flow in higher-dimensional CDT. Despite the lack of an analytical solution to the full CDT model at the moment, we attempt to achieve our goal by studying simplified graph models and   using ideas and concepts from the field of random infinite graphs. The reasons for studying random graphs are multiple \cite{Durhuus:2011co}. First, triangulations are graphs themselves. Second, random surfaces have a ``tree" phase (the branch polymer phase of DT). Studying the properties of trees, e.g. dimensionality, we probe properties of  the fractal phases of quantum geometry. Third, random trees encode information about random triangulations, via bijections (for example see figure \ref{bijection}). 
Fourth, many random graphs serve as toy models for studying further fractal aspects, like the dimensionality, of random surfaces. Specifically, studying random walks on random graphs is equivalent to the (discrete) diffusion process taking place in simulations of CDT. Choosing the appropriate graph ensemble for studying the diffusion might be a way to simplify the complexity of the problem, get better analytical understanding of the numerical results and approach our goal.  For all these reasons random graphs play an essential role in the study of random surfaces and the CDT approach to quantum gravity. We start this chapter by introducing the  basic concepts and tools of graph theory. Then, we proceed by emphasising the link between triangulations and other important graphs, especially the random trees and branching processes. 

\section[Basic elements of graph theory]{Basic elements of graph theory \footnote{We mostly follow the definitions in \cite{Harris:2008, Durhuus:2009zz}.}}
A graph $G$ consists of two sets; the set of vertices $V(G) = \{ v_1, v_2, \ldots ,v_N \}$ and the set of edges $E(G)$, which is the set of unordered pairs of vertices $\{(v_i, v_j)\}$. The \textit{order} and the \textit{size} of the graph $G$ are the cardinality of its vertex and edge set respectively. We will denote the size of the graph by $|G|$. Given any two vertices $u, v \in V(G)$, they are said to be \textit{adjacent} if $(u,v) \in E(G)$. Otherwise, they are \textit{non-adjacent}. In addition we say that an edge $e$ is incident to the vertex $v$ if $e = (v,u)$. The \textit{degree} of a vertex $v$, denoted by $\sigma (v)$, is the number of edges containing the vertex $v$. It is convenient in our analysis to pick a special vertex, called the \textit{root}, $r$, which has only one neighbour, i.e. $\sigma (r) =1$. A graph with a root vertex is called a \textit{rooted graph} \footnote{In some textbooks, e.g. \cite{Drmota:2009}, rooted graphs are considered graphs with a special vertex of arbitrary degree. Within this definition graphs with an extra special vertex of degree one are called \textit{planted graphs}.}.

A \textit{path} in $G$ is a sequence of different edges $\{(v_0, v_1), (v_1,v_2), \ldots, (v_{k-1}, v_k) \}$ where $v_0$ and $v_k$ are called the \textit{end vertices}. If the end vertices are the same, i.e. $v_0 = v_k$ the path forms a \textit{circuit} (or loop or cycle).    
Moreover, a graph is \textit{connected} if every pair of vertices can be joined by a path. We further denote by $d(v,u)$ the \textit{graph distance}  between two vertices $v$ and $u$ which is defined to be the minimal number of edges in a path connecting them.  
Additionally, we define the height $h(v)$ of a vertex $v$ to be the graph distance from root $r$ to $v$, i.e. $h(v) \equiv d (r,v)$, and $S_k(G)$ to be the set of all vertices of $G$ having height $k$. We denote by $B_R(G,v)$ the ball of radius $R$ centred at vertex $v$, which is the subgraph of $G$ spanned by vertices at graph distance at most $R$ from $v$.

Finally a graph is said to be \textit{planar} if it can be drawn in the plane in such a way that pairs of edges intersect only at vertices, if at all. A drawing of a planar graph $G$ in the plane in which edges intersect only at vertices is called \textit{planar embedding}. Two planar graphs are considered identical if one can be continuously deformed into the other in $\BB{R}^2$. 

In this thesis, we are mainly interested in rooted planar connected infinite graphs, i.e. both sets $E(G)$ and $V(G)$ have infinite elements. In particular the infinite graphs are constructed from the thermodynamic limit of fixed size $N$ graphs, as $N\to \infty$. Additionally, all graphs that will be considered are \textit{locally finite}, i.e. $\sigma (v)$ is finite for any $v \in V(G)$. We will also consider connected graphs with no circuits, called \textit{trees}, and graphs where we allow repeated elements in our set of edges, i.e. \textit{multigraphs} (see figure \ref{multigraph}).

\section{Random walks on graphs}
\label{random walks on graphs}
A \textit{walk} in a graph is a sequence of (not necessarily distinct) vertices $v_0, v_1, \ldots, v_k$ such that $(v_i, v_{i+1}) \in E(G)$ for $i = 0, 1, \ldots, k-1$.   
We shall denote such a walk by $\omega: v_0 \to v_k$ and call $v_0$ the \textit{origin} and $v_k$ the \textit{end} of the walk. In addition, the \textit{length} of the walk, $|\omega|$, is the number of its edges, counting repetitions, $\Omega _G$ is the set of all walks on $G$, and $\omega(t)$ is the location of the walk $\omega \in \Omega_G$ at time $t$. In a random walk $\omega$, at each time step, a walker at $\omega(t)$ steps to one of the neighbouring vertices $\omega(t+1)$ with probability given by the ratio 
\beq \label{prob_next_step}
\frac{\text{number of edges from} \ \omega(t) \ \text{to} \ \omega(t+1)}{\text{total number of edges incident to vertex} \ \omega(t)}. 
\eeq
In particular, in the case of tree graphs the walker steps with equal probability to one of the neighbouring vertices which is given by $1/\sigma (\omega(t))$. Note that a random walk which is at the root at time $t$ moves to vertex $1$ with probability one.  

Consider a random walk on planar rooted graphs with vertices of finite degree. For simplicity we assume that the random walker starts from the root (this choice does not affect the value of the spectral dimension, defined below).  Given a graph $G$, let $p_G(t)$ be the probability the walker is at the root after $t$ steps  and define $p_G^{1}(t)$ to be the probability that the walker returns to the root at time $t$ \textit{for the first time} after $t=0$, with the convention $p_G^{1}(0)=0$ \cite{Durhuus:2005fq}
\bea
\label{return_prob}
p_G(t)    &=& \sum _{\omega \in \Omega_G: |\omega| = t}\BB{P} \left (\{ \omega (t) = 0 | \omega(0) = 0\}\right ), \\
\label{first_return_prob}
p_G^1(t) &=& \sum _{\omega \in \Omega_G: |\omega| = t} \BB{P} \left (\{ \omega (t) = 0 | \omega(0) = 0, \omega (t') > 0, 0 <t'<t \}\right ).
\eea
These two probabilities are related to each other by decomposing the random walk into an arbitrary number of first returns to the root \cite{Jonsson:1997gk, Durhuus:2005fq}
\bea \label{return_prob_decomp}
p_G(t) &=& \delta _{t,0} + p_G^1(t) + \sum _{{t_1, t_2: \atop  t_1 + t_2 = t}} p_G^1 (t_1) p_G^1 (t_2) + \dots + \sum _{{t_1, \ldots, t_n : \atop t_1 + \ldots +t_n = t}} p_G^1 (t_1)  \cdot \ldots \cdot p_G^1 (t_n) +\ldots \nn\\
           &=& \delta _{t,0} + \sum _{n=1}^{\infty} \sum _{{t_1, \ldots, t_n : \atop t_1 + \ldots + t_n =t}} \prod _{k=1} ^{n}p_G^1 (t_k) 
\eea
We should note that there are graphs, such as the combs, trees and multigraphs, in which both return probabilities $p_G(t)$ and $p_G^1(t)$ vanish for odd time steps because the walker needs an even number of steps for return to the root. However this does \textit{not} hold for causal triangulations.

\section{Spectral and Hausdorff dimension of graphs}
\label{ds and dh}
There are two key notions of dimension which characterise the fractal properties of a graph $G$.  The first, the \textit{spectral dimension} $\ds$, is related to a random walk (discrete diffusion) and is given by the asymptotic behaviour of the return probability at large times, $p_G(t) \sim t ^{-\ds/2}$, as $t\to \infty$. For the formal definition we write \cite{Durhuus:2009zz}
\beq \label{ds_graph_def}
\ds := - 2 \lim _{t \to \infty} \frac{\ln p_G(t)}{\ln t}
\eeq 
provided that the limit exists. If the graph is finite then $\ds=0$, because the random walk will return to the root with probability 1 after infinitely many steps. On the other hand if the graph is infinite one has $\ds \geq 1$. 

The second essential notion is the Hausdorff dimension $\dha$ which is defined 
\beq \label{dh_def}
\dha :=  \lim _{R \to \infty} \frac{\ln |B_R(G,v)|}{\ln R}
\eeq
provided the limit exists. For finite graphs $\dha = 0$, whereas if $G$ is infinite then $\dha \geq 1$ \cite{Durhuus:2009zz}. It is important to note that the existence and value of the limit do not depend on the choice of the vertex $v$. For this reason from now on we will consider balls of radius $R$ centred at the root  denoting them by $B_R(G)$ or $B_R$ as a shorthand notation.

\textit{$\dha$ vs $\ds$}: As we will see in detail in the next sections, these two definitions of dimensionality do not always agree, because they are related to different characteristics of the graph. The Hausdorff dimension is related to the volume growth, whereas the spectral dimension is sensitive to the graph's connectivity. In particular, it was proved in \cite{Coulhon:2000rw} that for fixed graphs under certain assumptions the following inequalities should be satisfied
\beq \label{ds_inequalities}
\dha \geq \ds \geq \frac{2 \dha}{ 1+\dha}.
\eeq
 We will further comment on these inequalities in section \ref{from triangulations to trees}.

\subsection{Spectral dimension via generating function}
Next we introduce the generating functions of the return and first return probabilities respectively being \cite{Durhuus:2005fq}
\beq \label{Q_and_P_def}
Q_G(s) = \sum _{t=0}^{\infty} p_G (t) s^t \qquad \text{and} \qquad P_G(s) = \sum_{t=0}^{\infty} p_G^1 (t) s^t 
\eeq 
Following the decomposition \eqref{return_prob_decomp} the generating functions are related by (see figure \ref{P_and_Q_fig} for a graphical representation)
\begin{figure}
\begin{center}
\includegraphics[scale=0.5]{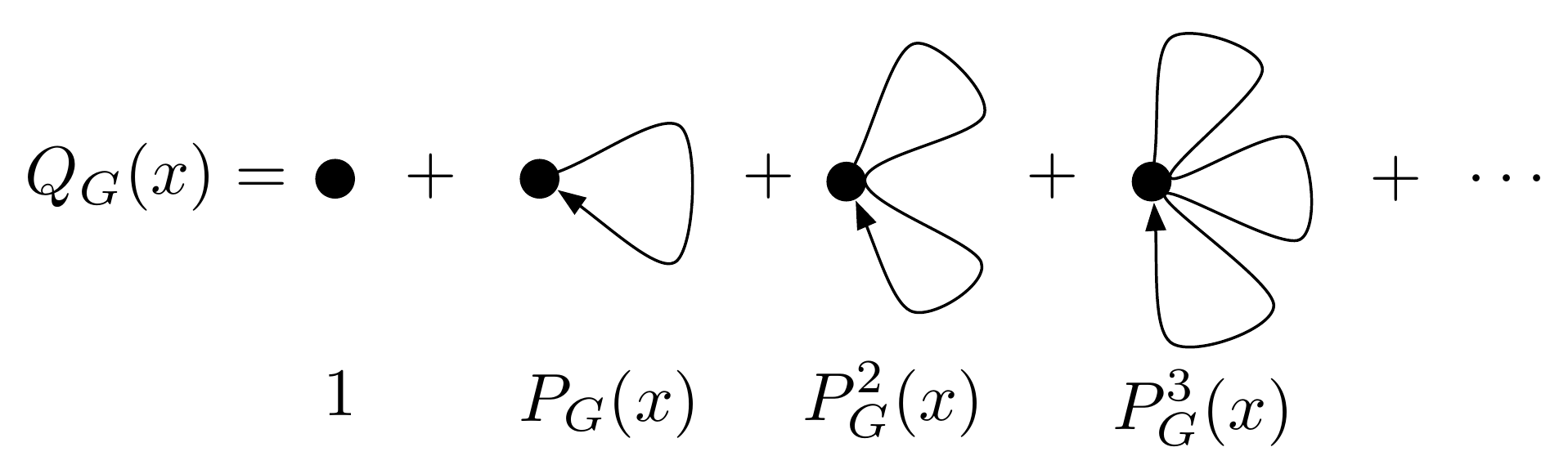}
\caption{Decomposition of the generating function $Q_G(x)$ into first returns.}
\label{P_and_Q_fig}
\end{center}
 \end{figure}
\beq \label{Q_decomp_of_P}
Q_G(s) = \sum _{n=0}^{\infty} ( P_G(s) )^n = \frac{1}{1- P_G(s)}.
\eeq
From the definition \eqref{Q_and_P_def} we notice that $P_G(s)$ is analytic in the unit disc and $|P_G(s)|<1$ for $|s|<1$. 
It is convenient to change variable, $s^2 = 1- x$, assuming $0 \leq x \leq1$, and use the definitions
\beq
\label{Q_and_P_def_x}
Q_G(x) = \sum _{t=0}^{\infty} p_G (t) (1-x)^{t/2} \qquad \text{and} \qquad P_G(x) = \sum_{t=0}^{\infty} p_G^1 (t) (1-x)^{t/2}. 
\eeq
The asymptotic behaviour of $Q_G(s)$ and $P_G(s)$ as $s \to 1$ is equivalent to the asymptotics of $Q_G(x)$ and $P_G(x)$ as $x \to 0$.
The generating function method is a useful tool for extracting the spectral dimension  from the asymptotic behaviour of $Q_G(x)$, as $x\to 0$, provided it diverges. Indeed a Tauberian theorem \cite[chapter XIII]{Feller:1965b} tells us that  
\beq \label{ds_via_Q}
p_G(t) \sim t^{-\ds/2}, \ \ t\to\infty \ \ \Leftrightarrow \ \ Q_G(x) \sim x ^{-1 +\ds/2} \qquad \text{as} \qquad x\to 0,
\eeq
given that $Q_G(x)$ diverges at this limit. We should emphasise that  by $f(x)\sim g(x)$ it is meant that there is a constant $0< x_0 \leq1$ such that  
\beq 
g(x) \psi_{-}(x)\le f(x)\le g(x) \psi_{+}(x) \quad \text{for} \qquad 0<x<x_0,
\eeq
where $\psi_-(x)$ and $\psi_+(x)$ are slowing varying functions at $0$
\footnote{A function $\psi (x)$ is slowly varying at $x_s$ if $\lim _{x \to x_s}\psi (\lambda x)/\psi(x) =1$ for any $\lambda>0$. An example of a slowly varying function at $0$ is the logarithm, e.g. $\psi(x) = (\log (x))^{c}$, where $c$ is an arbitrary constant.}.

\textit{Example.} For instructive reasons we proceed with the calculation of the spectral dimension of the simplest graph, a half line with vertices $\{r, s_1, s_2, \ldots\}$ which is denoted by the subscript $\infty$ (this pedagogical example has also been  considered in \cite{Durhuus:2005fq, Stefansson:2012sd, Giasemidis:2012rf}). We focus on the calculation of the \textit{first return} generating function $P_{\infty} (x)$. From its definition we observe that every step of the walker contributes a factor $s=\sqrt{1-x}$ to $P_{\infty} (x)$. Firstly, the walker leaves the root with probability 1 going to vertex $s_1$. There, it has two possibilities; either the walker steps back to the root with probability 1/2 and the diffusion is over contributing to $P_{\infty} (x) = \half (1-x)$. Or, being at $s_1$, the walker could leave to the right with probability 1/2 and diffuse in the semi-infinite graph until he/she returns to $s_1$ for the first time. The generating function of the latter diffusion is identical to  $P_{\infty} (x)$ because the random walk from the root and back is identical with the walk from vertex $s_1$ to the right and back again. Then the walker can either go to the root with probability 1/2 and then the process is over contributing $P_{\infty} (x) = \half (1-x) + \half (1-x)P_{\infty} (x)$, or step to the right and so on. Therefore one gets the following recursion relation for the generating function
\beq \label{P_half_line}
P_{\infty} (x) = \half (1-x) \sum_{n=0} ^{\infty} \left (\half P_{\infty}(x) \right )^n = \frac{1-x}{2- P_{\infty} (x)}
\eeq
which can be solved for $P_{\infty}(x)$ and gives
\beq \label{P_infty}
P_{\infty}(x) = 1 - \sqrt{x} \Rightarrow Q_{\infty} (x) = \frac{1}{\sqrt{x}}.
\eeq
So, using the definition \eqref{ds_via_Q} we conclude that $d_s = 1$ as expected since the half line is the simplest graph without any structure and its spectral dimension should be equal to the topological dimension. 

\textit{Recurrent vs transient graphs.} From definition \eqref{ds_via_Q} one observes that when the generating function $Q_G(x)$ diverges as $x\to0$ then $\ds \leq 2$.   
In this case, expression \eqref{Q_decomp_of_P} implies $P_G(0) =1$ which implies, in the view of definition \eqref{Q_and_P_def_x}, $\sum_{t=1}^{\infty} p^1_G(t)=1$, which means that the probabilities of first return at any time  sum to one and therefore the random walker always returns to the root. Random walks on such graphs are called \textit{recurrent}. Examples of recurrent graphs are the generic random tree \cite{Durhuus:2006vk}, random combs \cite{Durhuus:2005fq}, random brushes \cite{Jonsson:2008rb} and non-generic trees \cite{Jonsson:2011cn, Stefansson:2012sd}.

One the other hand, when $\ds\geq2$ the generating function $Q_G(0)$ is finite which implies that $P_G(0) <1$
\footnote{The boundary case $\ds=2$ is sensitive to the slowing varying terms $\psi(x)$ and might correspond either to recurrent or transient random walks.}.
 Thus there is a non-zero probability that the random walker escapes to infinity and the random walk is characterised as \textit{non-recurrent} or \textit{transient}. Examples of the latter are the biased random walks on combs \cite{Elliot:2007br} and multigraphs induced by higher-dimensional causal triangulations (see chapter \ref{multigraphs}) \cite{Giasemidis:2012rf}. In the case of a transient random walk we define the spectral dimension through the derivative of $Q_G(x)$ of lowest degree which is diverging via the relation
\beq
Q_G^{(k)} (x) \sim x^{-(k+1) + \ds/2}, \qquad \text{as} \qquad x\to0,
\eeq  
for $2k \leq \ds < 2 (k+1)$.
In section \ref{resistance} we will discuss recurrence in terms of \text{graph resistance}.

\section{Random (infinite) graphs}
\label{random infinite graphs}
In the previous sections we introduced several definitions of the properties of a single graph $G$. However, we are also interested in the fractal properties  of an \textit{ensemble} of graphs, or a \textit{random graph}, which is a set $\C{G}$ of graphs equipped with a probability measure $\mu$. Measurable subsets $\C{A}$ of $\C{G}$ are called \textit{events} and $\mu (\C{A})$ is the probability of the event $\C{A}$. We also denote by $\avg{\cdot}_{\mu}$ the expectation with respect to the measure $\mu$. Furthermore we define the \textit{annealed} Hausdorff dimension $\bar \dha$ and \textit{annealed} spectral dimension $\bar \ds$ of an ensemble  $\R{G} = (\C{G},\mu)$ by \cite{Durhuus:2009zz}
\bea
\bar \dha &=&\lim_{R\to \infty} \frac{\ln \avg{|B_R(G)|}_{\mu}}{\ln R} ,\\
 \bar \ds  &=& -2 \lim _{t \to \infty} \frac{\ln \avg{p_G(t)}_{\mu}}{\ln t}. 
\eea
respectively, provided that the limits exist.  Using generating function techniques, $\bar \ds$ is also determined via the ensemble average of the generating function defined as
\beq
\bar Q (x) \equiv \avg{Q_G(x)}_{\mu} := \int Q_G(x) d\mu \sim x^{-1+\bar\ds/2},  
\eeq
provided it is diverging as $x\to0$. 

Additionally to the annealed dimensions, one can make stronger statements about the values of $\dha$ and $\ds$ for \textit{almost all} graphs in the ensemble. What we mean by saying ``almost all" is that there is only a finite set of graphs with different values of the Hausdorff or spectral dimensions. Thus the probability of ``picking'' such a graph from an infinite set is zero. Or, in other words, the probability of finding a graph in the ensemble with the value $\dha$ is one. Formally we say that the Hausdorff dimension, for example,  of a random graph is almost surely $\dha$, if there exists a subset $\C{G}_0$ of $\C{G}$ such that $\mu(\C{G}_0) =1 $ and such that every $G \in \C{G}_0$ has Hausdorff dimension $\dha$ 
\footnote{By abuse of notation, we denote the annealed dimensions by $\ds$ and $\dha$ too in the rest of the thesis, even though the value of the annealed dimension might be different from the ``almost sure'' value (see for example \cite{Jonsson:2011cn}). However this should not cause further confusion because it will be clear from the discussion to what dimension we refer to.}.

As we have already mentioned we are mostly interested in ensembles of infinite graphs where the measure $\mu$ is obtained as a limit of measures $\mu _N$ defined on sets of finite graphs of size $N \in \BB{N}$. 
This limit must be understood in the weak sense, meaning that 
\beq
\int_{\C{G}} f d\mu_N \stackrel{N\to\infty}{\longrightarrow}  \int_{\C{G}} f d\mu
\eeq
for all bounded continuous functions $f$ on $\C{G}$ \cite{Durhuus:2009zz}.

\section{From causal triangulations to trees}
\label{from triangulations to trees}
The discussion and definitions so far have been applied to arbitrary graphs. In this section we focus on random triangulations, which are the building blocks for the discretised models of gravity considered here, and show their relation to planar tree graphs. This relation arises naturally if one considers two-dimensional discretised models of gravity. It was soon realised, that the problem of calculating the loop functions of those models is reduced to a purely combinatorial problem of counting the number of distinct triangulations with given $S^2$-topology and boundaries. This problem can be solved either by using Tutte's recursive decomposition \cite{Tutte:1963ac, Tutte:1968en} or by matrix model techniques used by physicists \cite{Brezin:1977sv, Ambjorn:1997} or by bijective proofs used by mathematicians \cite{Bouttier:2004pm}. The latter method exploits the bijection between discretised surfaces and specific classes of trees, which have an easier enumeration. A well-known bijection relates quadrangulations of the sphere to well-labelled trees \cite{Schaeffer:1998cd, Chassaing:2002rp}. This bijection can be extended between triangulations and labelled mobile trees \cite{Bouttier:2004pm, Ambjorn:2013csx}. In other words, discretisations of the sphere without any causal structure are in one-to-one correspondence with labelled trees 
\footnote{We skip the details of those bijections since they are not of any particular interest for the discussion in the following chapters and the interested reader should consult the references mentioned above.}.
\begin{SCfigure}
\centering
\includegraphics[scale=0.35]{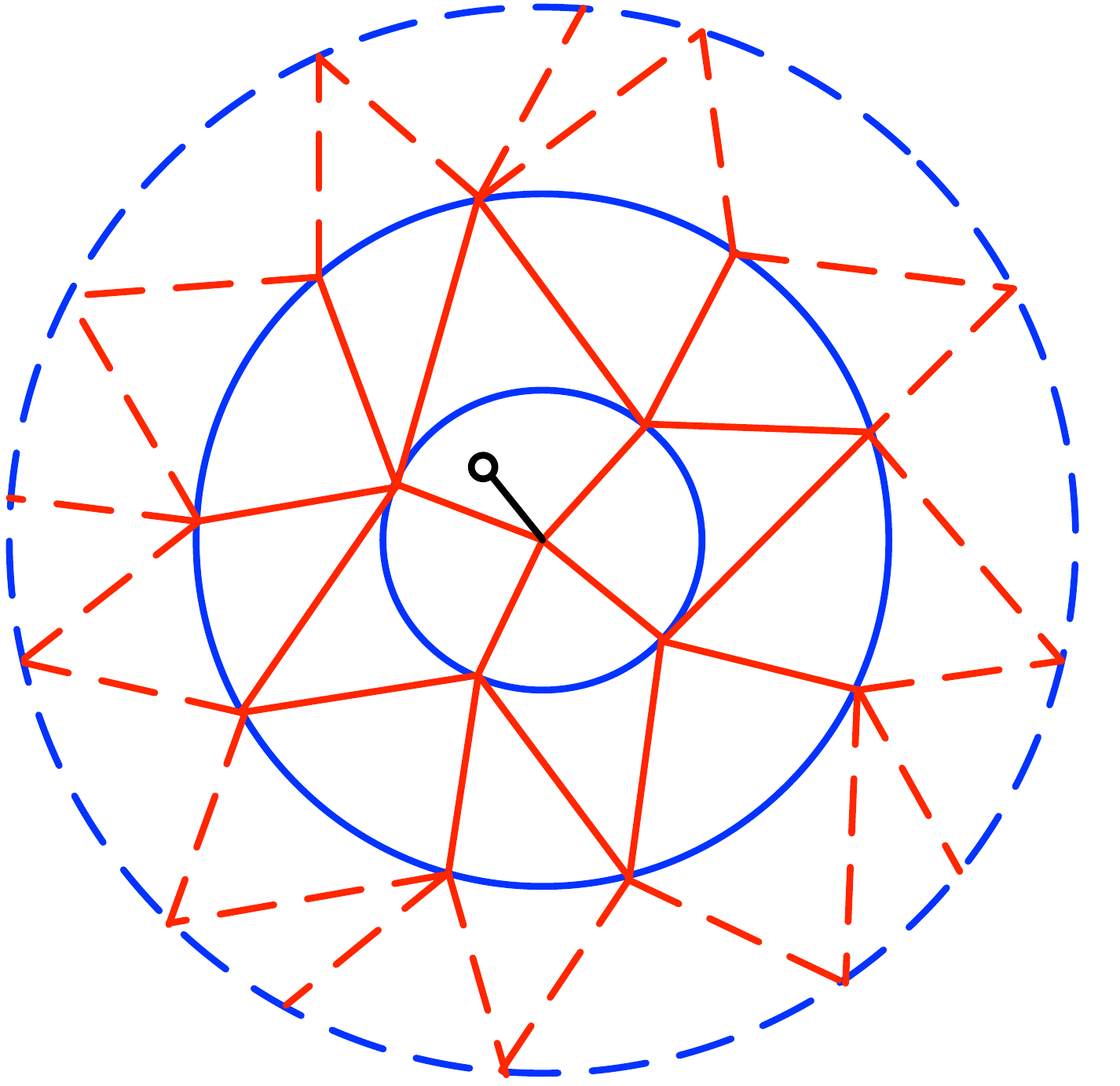} 
\caption{An example of a planar causal triangulation. Blue and red links correspond to space-like and time-like edges respectively.}
\label{planar_causal_triangulation}
\end{SCfigure}

By introducing the causal constraint, causal triangulations are in bijection with (unlabelled) trees. The bijection was studied in detail in \cite{Durhuus:2009sm} and here we elaborate because it plays an essential role in the following discussion. In particular, a causal triangulation $\tr$ consists of two sets of edges; the space-like edges which link vertices at the same height $k$ from the root, $S_k (\tr)$, and form a cycle, and the time-like edges which link vertices at heights $k$ and $k+1$  (see figure \ref{planar_causal_triangulation}). 
From this construction, every triangle consists of one space-like edge and two time-like edges. We further denote $\Sigma _k$ to be the subgraph of the triangles bounded by the cycles $S_k(\tr)$ and $S_{k+1}(\tr)$. The number of triangles in a strip is given by $\Delta (\Sigma_k) = |S_k(\tr)| + |S_{k+1}(\tr)|$. If the causal triangulation is finite, it is decorated appropriately \cite{Durhuus:2009sm}. Now the bijection is defined as follows. 

\textit{From causal triangulations $\tr$ to rooted planar trees $T$} (see figure \ref{bijection}). First, we mark one edge from $S_0(\tr)$ to $S_1(\tr)$ and attach a new edge $(r,S_0)$ such that the marked edge is the rightmost to $(r,S_0)$. Second, we retain all the edges from $S_0(\tr)$ to $S_1(\tr)$. Third, for $k \geq1$, delete the rightmost time-like edge from vertex $v \in S_k(\tr)$ to $S_{k+1}(\tr)$. Apply this process for any vertex $v\in S_k(\tr)$ and repeat for any $k$. Any decorations are deleted. Fourth, delete all the space-like edges. The resulting graph is a rooted (unlabelled) tree graph. 
\begin{figure}[t]
\centering
\begin{subfigure}{0.4\textwidth}
\includegraphics[scale=0.25]{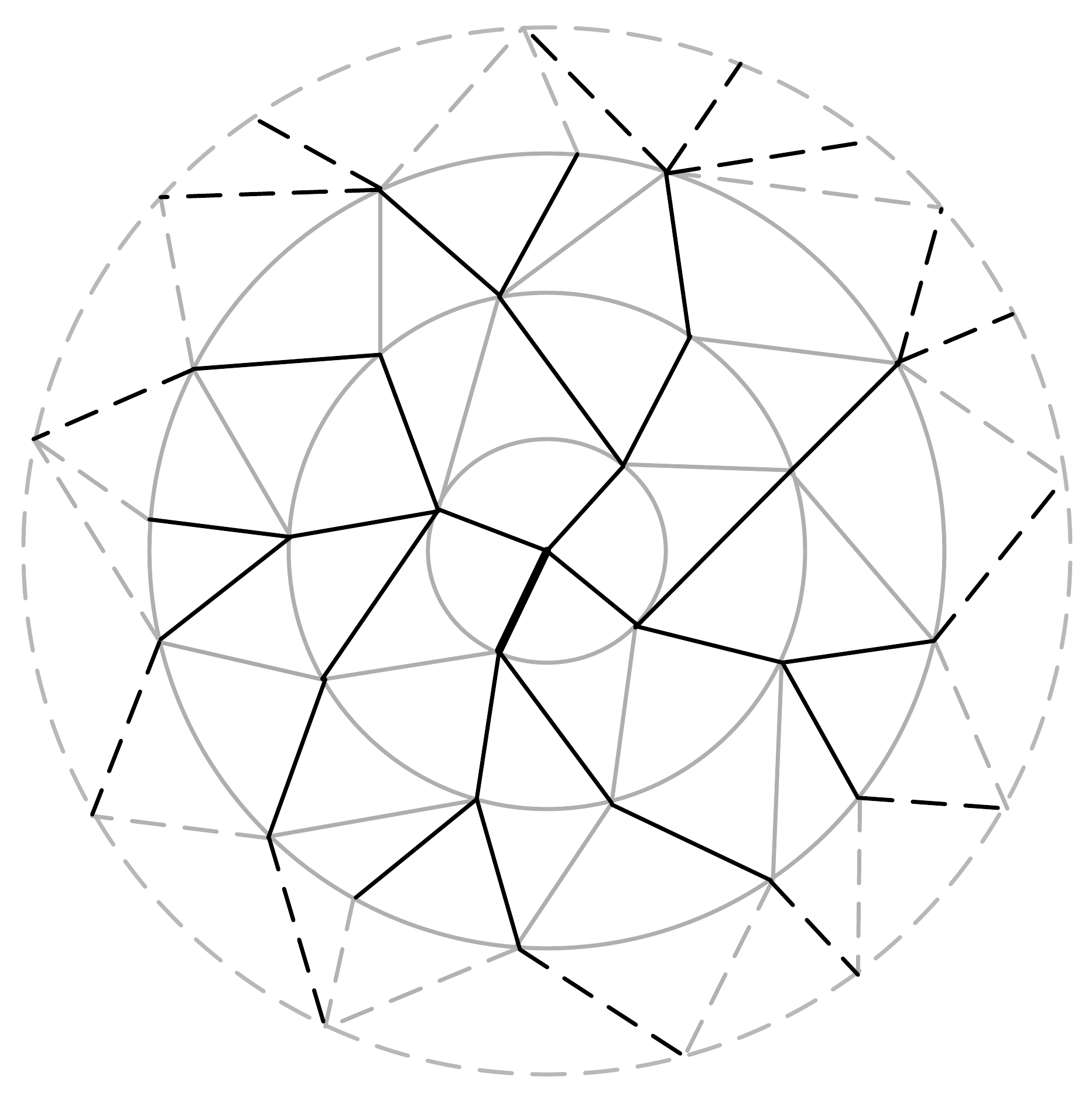} 
\centering \caption{A causal triangulation}
\label{causal_triang}
\end{subfigure}
\qquad \qquad 
\begin{subfigure}{0.4\textwidth}
\includegraphics[scale=0.25]{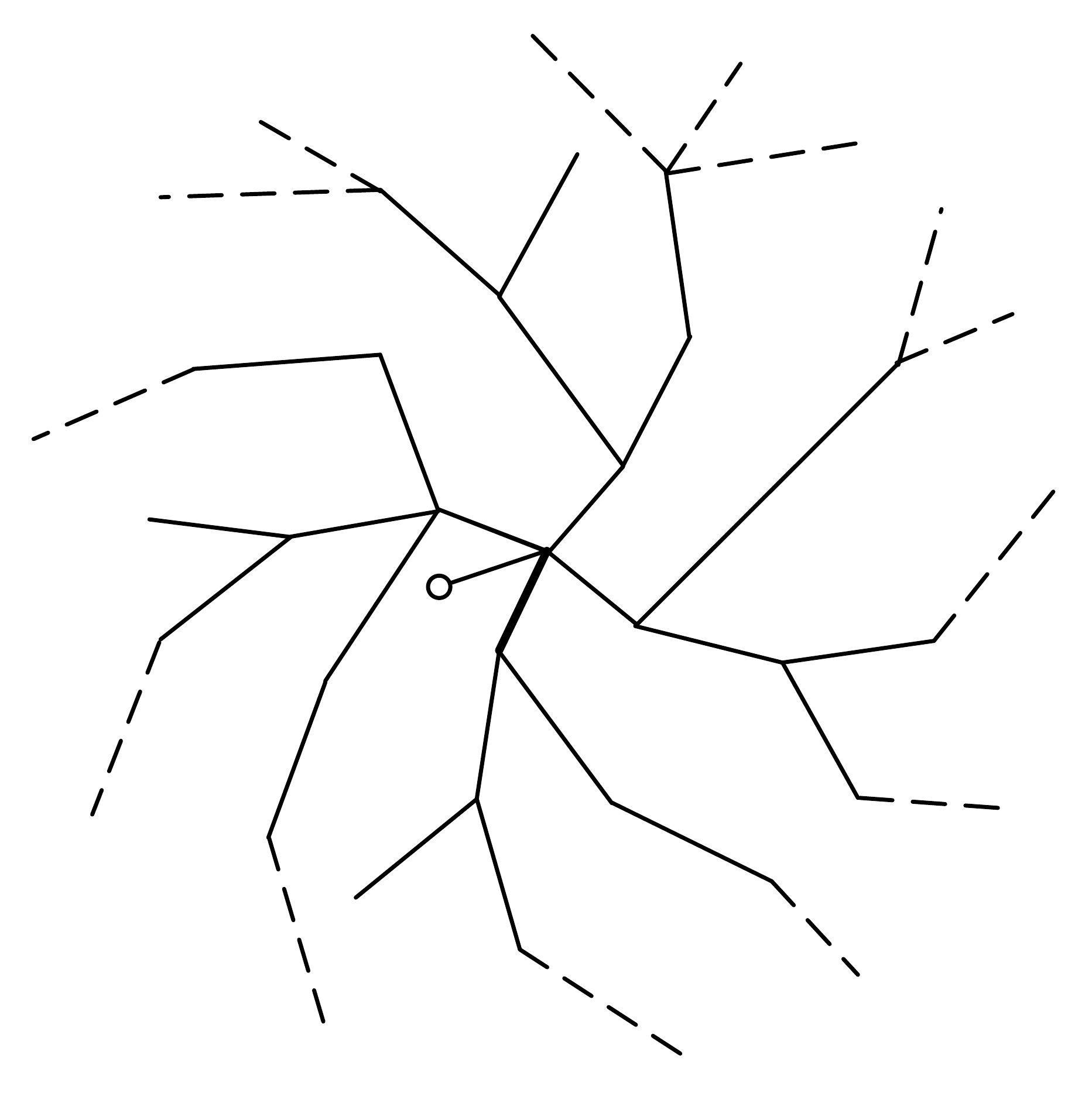} 
\centering \caption{Tree bijection}
\label{tree}
\end{subfigure}
\caption{From causal triangulation to trees. The bijection works as follows using the definitions of Section \ref{from triangulations to trees}; delete the rightmost edge of each vertex from $S_k$ to $S_{k+1}$ (time-like edges) and all edges connecting vertices at height $k$ (space-like edges) - deleted edges are drawn as grey lines. Add a vertex $r$ (empty circle) and mark an edge from $S_0$ to $S_1$ which is the rightmost edge with respect to the edge $(r, S_0)$ (fat line). Dashed lines encode the fact that both graphs are infinite and extend beyond finite height.}
\label{bijection}
\end{figure}

\textit{From rooted planar trees $T$ to causal triangulations $\tr$}. First, mark the rightmost edge from $S_1(T)$ to $S_2(T)$ and delete the root and the edge attached to it. Second, for every vertex $v \in S_{k}(T)$, $k\geq2$, add a new edge at the rightmost of $v$, which links $v$ with a vertex in $S_{k+1}(T)$ such that the new edge does not cross any existing edges. Third, draw edges linking vertices at the same height $k$, i.e. $S_{k}(T)$, for any $k$ and add decorations to the cycle of maximum height.  The resulting graph is a planar causal triangulation. 

Additionally, we denote by $D_k(T)$ the set of edges in a tree $T$ connecting vertices at height $k-1$ to vertices at height $k$. Clearly $|D_k(T)| = |S_k(T)|$. 

Bijections are not only useful for enumeration problems. One can also extract information about the measure of an ensemble which is inherited from the measure of its bijection. We will return to this point and elaborate in the next section. However, the two sides of the bijection do not necessarily share every feature, a simple example being the dimensionality. A particular example is the uniform infinite causal triangulation (UICT), essentially a planar CDT constrained never to die out, which has $\dha = 2$  and $\ds \leq 2$, both on average and almost surely \cite{Durhuus:2009sm}. On the other side of the bijection we get the generic random tree (GRT) which has $\dha = 2$  and $\ds = 4/3$, both on average and almost surely,   \cite{Durhuus:2006vk, Durhuus:2009sm, Stefansson:2012sd}, saturating the right hand side of \eqref{ds_inequalities} (see section \ref{CGRT} for more details). Therefore one observes that the GRT induces the value of the Hausdorff dimension into the UICT due to the nature of the bijection \cite[Theorem 3]{Durhuus:2009sm}, whereas the bijection carries no information about the spectral dimension. 

Another useful mapping, first considered in \cite{Durhuus:2009sm}, is the following. Starting from a causal triangulation we contract all the vertices at the same height into a new vertex, but retaining all the time-like edges. The resulting graphs are considered as reduced models of causal triangulations, so-called \textit{multigraph ensembles}, which are constructed from the discretised half line by allowing multiple edges between adjacent vertices with some probability distribution (figure \ref{multigraph})
\footnote{We postpone the description of this map and further details about multigraphs till chapter \ref{multigraphs}.}.
\begin{SCfigure}
\centering
\includegraphics[scale=0.35]{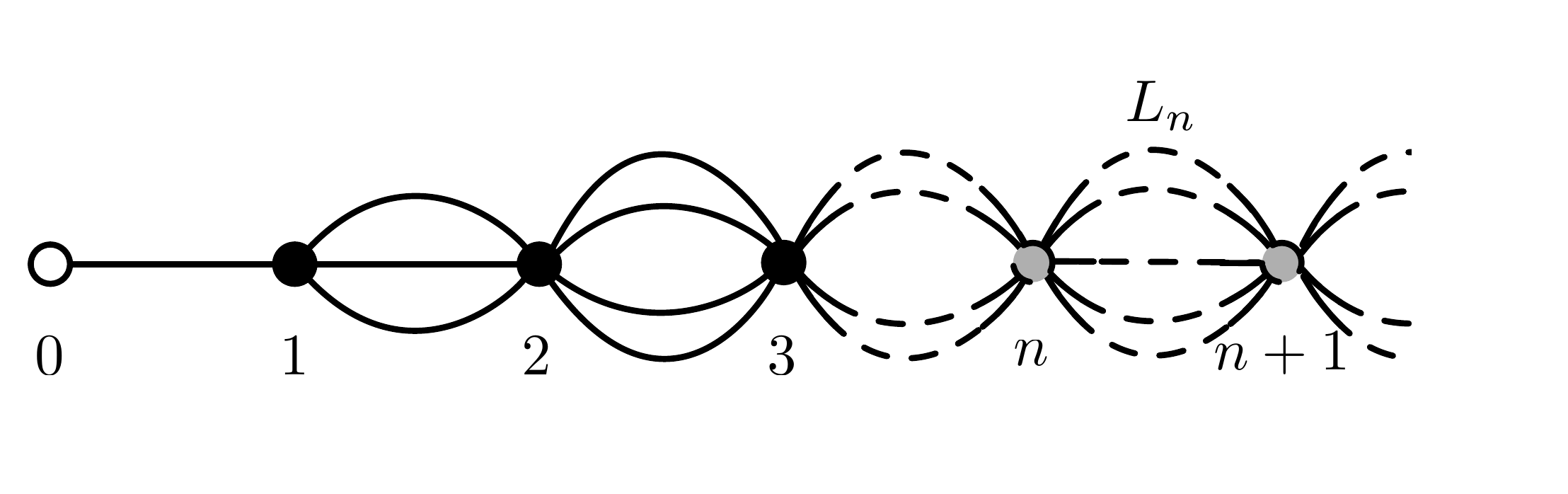} 
\caption{An example of a multigraph}
\label{multigraph}
\end{SCfigure}
The authors in \cite{Durhuus:2009sm} first proved that these reduced ensembles of the UICT have $\dha = \ds= 2$, both on average and almost surely saturating the left hand side of \eqref{ds_inequalities}. Then they showed that the spectral dimension of UICT is bounded above by the spectral dimension of multigraphs and therefore the UICT are recurrent (their argument becomes clear in the next paragraph). Here we see an example of an injection, where information about the spectral dimension is not completely lost. For the time being, a complete proof for the lower bound of the spectral dimension of the UICT is lacking. However numerical simulations indicate that it should be equal to 2. While the multigraph ensemble only gives an upper bound for the spectral dimension of the UICT, it is believed that both have the same spectral dimension. The UICT is consistent with the upper bound \eqref{ds_inequalities} even though the causal triangulations do not necessarily satisfy the assumption under which \eqref{ds_inequalities} is derived.

\subsection{Graph resistance}
\label{resistance} 
To understand further features of the mappings between causal triangulations, trees and multigraphs we introduce the notion of {\it graph resistance} $R_G$. It is defined by considering the graph as an electric network where each edge has resistance one \cite{Lyons:2012}. One distinguishes two cases: The recurrent case ($\ds\leq 2$) where a random walker ``faces'' infinite resistance to escape to infinity and returns to the root with probability 1; and the transient (or non-recurrent) case ($\ds \geq 2$) where finite resistance to infinity implies return probability strictly less than one (see also section \ref{ds and dh}). 

By Rayleigh's monotonicity principle 
\footnote{Intuitively, Rayleigh's monotonicity principle states that removing an edge increases the effective resistance, whereas contracting an edge, i.e. identifying its two endpoints and removing the resulting loop, decreases the effective resistance.}
\cite[section 2.4]{Lyons:2012}, the resistance from the root to infinity of the two-dimensional causal triangulation $R_{\tr}$ is bounded by
\beq \label{resistance_ineq}
R_{M} \leq R_{\tr} \leq R_{T},
\eeq
where $R_{M} $ and $R_{T}$ are the resistances of the corresponding multigraph and tree respectively. Given that the resistance of recurrent multigraphs is infinite this inequality serves as a proof that the two-dimensional UICT is recurrent and $\ds \leq 2$ almost surely. Furthermore it  implies that the recurrent multigraph ensemble and the generic tree ensemble are two extreme cases used to bound the spectral dimension of UICT and saturate the left and right hand side of \eqref{ds_inequalities} respectively, as we have already mentioned. It is believed that the spectral dimension of UICT is two and that thus multigraphs provide a tight bound.

\section{Galton-Watson and simply generated trees}
\label{GW trees and simply generated trees}
In the previous section we found that random planar trees have a special role in two-dimensional quantum gravity models due to their bijection with two-dimensional CDT. In particular we are interested in trees which are generated via a branching process. For this reason, in this section we focus on Galton-Watson \cite{Harris:1963, Athreya:2004} and simply generated trees and study their properties (we mostly follow \cite{Janson:2011sg, Jonsson:2011cn}).

Let us first fix some notation. Let $\C{T}_N$ be the set of all ordered rooted trees of size $N$ and denote by $\C{T}_f$ and $\C{T}_{\infty}$ the set of all rooted finite and infinite trees respectively. Then the set of all rooted trees, both finite and infinite is $\C{T} = \C{T}_f \cup \C{T}_{\infty}$. We also denote by $T_N$ and $T$ an element of the sets $\C{T}_N$ and $\C{T}$ respectively. 

\subsection{Galton-Watson trees}
\label{GW trees}
A \textit{Galton-Watson} process is a branching process which is determined by the \textit{offspring probabilities} $p_k, \ k=0,1,2,\ldots$, which are the probabilities that a single ``parent" will have $k$ ``children". The offspring probabilities satisfy the condition  $\sum _{k=0}^{\infty} p_k =1$. The process describes the evolution of a population of particles. It starts at time 0 with $Z_0$ particles each of which splits independently of the others into a random number of offspring according to the probability law $p_k$. The total number of particles in the first generation is denoted by $Z_1$. Every particle in the first generation splits independently according to $p_k$ and produces the second generation, with total number of particles $Z_2$, and so on. We should mention that the number of ``children" produced by a single ``parent" at any given time is independent of the history of the process and of the existence of other particles in the ``parent's" generation. If we represent the offspring by vertices and link them to the parent vertex we generate a tree. To consider rooted trees we start with population 1 of the zero-generation, i.e. $Z_0=1$, and attach a link to the single parent whose other end is incident to the root. Then notice that the offspring at the $n$-th generation correspond to vertices at height $n+1$ in the corresponding rooted tree, i.e. $Z_n = |D_{n+1} (T)|$. 

A very useful tool in the analysis of the Galton-Watson process is the generating function of the offspring probabilities
\beq
f(s) := \sum _{k=0}^{\infty} p_k s^k,
\eeq
analytic for $|s| \leq 1$. Let $P_n(i,j)$ be the probability that starting from $i$ particles, the $n$-th generation has population $j$. Then the following identities are in order \cite{Athreya:2004}
\bea
f_1(s) \equiv f(s) = \sum _{k=0}^{\infty} P_1 (1,k) s^k, \qquad f_n(s) :=   \sum _{k=0}^{\infty} P_n (1,k) s^k, \\ 
\sum _{k=0}^{\infty} P_1(i,k) s^k = \left ( f(s)\right )^i,  \qquad \sum _{k=0}^{\infty} P_n(i,k) s^k = \left ( f_n(s)\right )^i.
\eea
From these expressions we should note that the generating function $f(s)$ iterates as $f_{n+1} (s) = f(f_n(s))$. In addition the expectation value of offspring in the next generation is given by
\beq \label{mean_offspring}
m =  f' (1) = \sum _{k=1}^{\infty} k p_k.
\eeq 
The mean offspring value $m$ categorises the process into three cases \cite{Athreya:2004}: 
\begin{itemize}
\item \textit{Subcritical}, $f'(1) <1$, where the process dies out exponentially fast with probability one, therefore the corresponding tree is finite.
\item \textit{Critical}, $f'(1) =1$, where the process dies out with probability one, therefore the corresponding tree is finite too.
\item \textit{Supercritical}, $f'(1) >1$, where the process has a non-zero probability to survive to infinity, therefore the corresponding tree might be infinite.
\end{itemize} 
Additionally, the process is called \textit{generic} if the generating function $f(s)$ is analytic at $s=1$ \cite{Durhuus:2006vk, Durhuus:2009zz}. Next, we present an important lemma for the critical Galton-Watson tree \cite{Athreya:2004}.
\begin{lemma}
If $m = f'(1) =1$ and $\sigma ^2 = f''(1) <\infty$, then
\beq
\frac{1}{1-f_n(s)} = \frac{1}{1-s} + \frac{n f''(1)}{2} + o (n), \qquad \text{as} \ n\to \infty,
\eeq
uniformly for $0\leq s<1$. 
\end{lemma}
Applying this lemma for $s=0$ we deduce the probability of survival of the $n$-th generation
\beq
\BB{P} (Z_n > 0) = 1- f_n(0) = \frac{2}{nf''(1)} + o (1/n), \qquad \text{for large} \ n.
\eeq

The Galton-Watson process yields a probability measure on the set of all finite trees, $\C{T}_f $,
\beq \label{muGW}
\mu _{GW} (T) = \prod_{v \in V(T)\backslash r} p_{\sigma (v)-1}, \qquad \text{for} \qquad T \in \C{T}_f 
\eeq
and the ensemble $(\C{T}_f, \mu _{GW})$ is called a Galton-Watson tree. 
If we \textit{condition} the size of the tree to be $|T|=N$, then these random trees are called \textit{conditioned} Galton-Watson trees and the probability distribution $\mu_N$ on the set of trees with fixed size $N, \ \C{T}_N$, is given by
\beq \label{muN}
\mu_{N} (T)= \frac{\mu _{GW} (T)}{\mu(\C{T}_N)}, \qquad \text{for} \qquad T \in \C{T}_N 
\eeq
where $\mu(\C{T}_N) = \sum _{T\in \C{T}_N}\mu _{GW} (T)$. 

\textit{Uniform trees}: Of special interest is the case where the offspring probabilities are given by \cite{Durhuus:2003pa, Durhuus:2009sm}
\beq \label{pk_induce_uniform_measure}
p_k = 2^{-k-1}, \qquad k\geq 0.
\eeq
One can easily verify that $f(1) = f'(1) =1$, therefore the process is generic and critical. The interesting feature of this process is that it induces a \textit{uniform measure} on the set of trees with fixed size $N$, $\C{T}_N$. This can be seen from \eqref{muGW} and \eqref{muN} as follows. Noting that $\sum _{v \in V(T_N)\backslash r} \s(v) = 2N-1$, we get
\bea
\mu _{GW} (T) &=& 2^{-2N+1}, \\
\mu (\C{T}_{N}) &=& \sum _{T\in \C{T}_N}\mu _{GW} (T) = 2^{-2N+1} C(N-1),
\eea
where $C(N-1) = \frac{1}{N}{2N-2 \choose N-1}$ is the $N-1$-th Catalan number which enumerates the number of trees of size $N$ \cite[sections 1.2.2, 3.1.2]{Drmota:2009}.  
Thus
\beq \label{uniform_measure}
\mu_{N} (T) = \frac{1}{C(N-1)}, \qquad \text{for} \qquad T \in \C{T}_N,
\eeq
 is a uniform measure on the set $\C{T}_N$. The ensemble  $(\C{T}_N, 1/C(N-1))$ generated by \eqref{pk_induce_uniform_measure} is called the uniform generic critical random tree.
  
Under the prescription of section \ref{from triangulations to trees}, we can relate a planar causal triangulation to a rooted planar tree. In \cite{Durhuus:2009sm}, it was shown that the measure $\rho _N(\tr)$ on the set of causal triangulations of area $2N-2$, $\trs _{2N-2}$, equals the uniform measure \eqref{uniform_measure} on the set of trees of size $N$. The thermodynamic limit of $\rho _N(\tr)$, $N \to \infty$, exists and converges to the uniform measure on the set of infinite causal triangulation, which we have called the uniform infinite causal triangulation (UICT) ensemble.


\subsection{Simply generated trees}
\label{simply generated trees}

Simply generated trees are tree ensembles which are generated by a sequence of non-negative number, $\{w_n\}, \ n\geq1$, the \textit{branching weights}. We then define the \textit{weight} of a finite tree $T\in \C{T}_f$ to be
\beq
w(T) := \prod _{v \in V(T)\backslash {r}} w_{\sigma (v)}
\eeq
To avoid trivialities we assume $w_1 > 0$ and $w_n>0$ for some $n \geq 3$.   
Now we introduce the \textit{finite volume partition function} defined by
\beq
Z_N = \sum _{T\in \C{T}_N} w(T)
\eeq
By this construction, it becomes clear that the probability of picking an element $T$ from the set of trees of size $N$, $\C{T}_N$, is given by the distribution
\beq \label{Prob_of_T_N}
\nu _N(T)=  \frac{w(T)}{Z_N}, \qquad \text{for} \qquad T \in \C{T}_N.
\eeq
Therefore the set of trees $\C{T}_N$ equipped with the probability law \eqref{Prob_of_T_N} defines a tree ensemble. 

A special case of branching weights is when $\sum _{n=1}^{\infty} w_n = 1$, and the sequence $\{w_n\}$ is a probability weight sequence. In this case the probability of picking a finite tree from the set $\C{T}_f$ is given by its weight, i.e. $\nu _f(T) = w(T)$. In addition the finite volume partition function $Z_N$ becomes the probability of picking a tree of size $N$ from the set $\C{T}_f$, i.e. $Z_N(T) = \nu_f(\{ T : |T|=N\}) = \nu (\C{T}_N)$. Therefore it is easy to notice that the simply generated random tree ($\C{T}_N, \nu_N(T)$) with probability weight sequence is equivalent to a Galton-Watson tree $(\C{T}_N,\mu_N(T))$ conditioned on size $N$ with the following identification of probability distributions
\beq \label{GW_equiv_sgt}
\nu_f(T) = \mu _{GW} (T), \qquad \nu (\C{T}_N) \equiv \mu (\C{T}_N), \qquad \nu _N(T) \equiv \mu _N (T).
\eeq

An important feature of simply generated trees is that the branching weights can be rescaled without changing the simply generated random tree. In particular, let $a, b >0$ and rescale the branching weights by $\tilde w _n := a b^{n-1} w _n$. Then the weight for every tree in $\C{T}_N$  is shifted by $\tilde w (T) = a^N b^{N-1} w (T)$. Moreover the finite volume partition function changes at the same amount, $\tilde Z_N = a^N b^{N-1} Z_N$, leaving the probability \eqref{Prob_of_T_N} unchanged. Therefore the rescaled weight sequence $\{\tilde w_k\}$ defines the same simply generated random tree $(\C{T}_N, \nu_N(T))$ as $\{w_k\}$. We use this equivalence to rescale tree ensembles $(\C{T}_N, \nu_N(T))$ such that $\sum _{k=1}^{\infty}\tilde w_k =1$ and relate them to Galton-Watson trees according to \eqref{GW_equiv_sgt}. 

\subsection{Critical generic random tree}
\label{critical generic random tree}
\label{CGRT}
To study the thermodynamic limit, $N\to\infty$, and classify infinite random trees we need to introduce the generating functions of the finite volume partition functions and branching weights being 
\bea
\label{partition_gen_fun}
\C{Z} (s) &=& \sum _{N=1}^{\infty} Z_N s ^N, \\ 
\label{weights_gen_fun}
g(z) &=& \sum _{n=0} ^{\infty} w_{n+1} z^n,
\eea
with radius of convergence $s_0$ and $\rho$ respectively. These two generating functions are related through the recursion relation (see figure \ref{Z_decomposition} for a graphical interpretation)
\begin{SCfigure}
\includegraphics[scale=0.5]{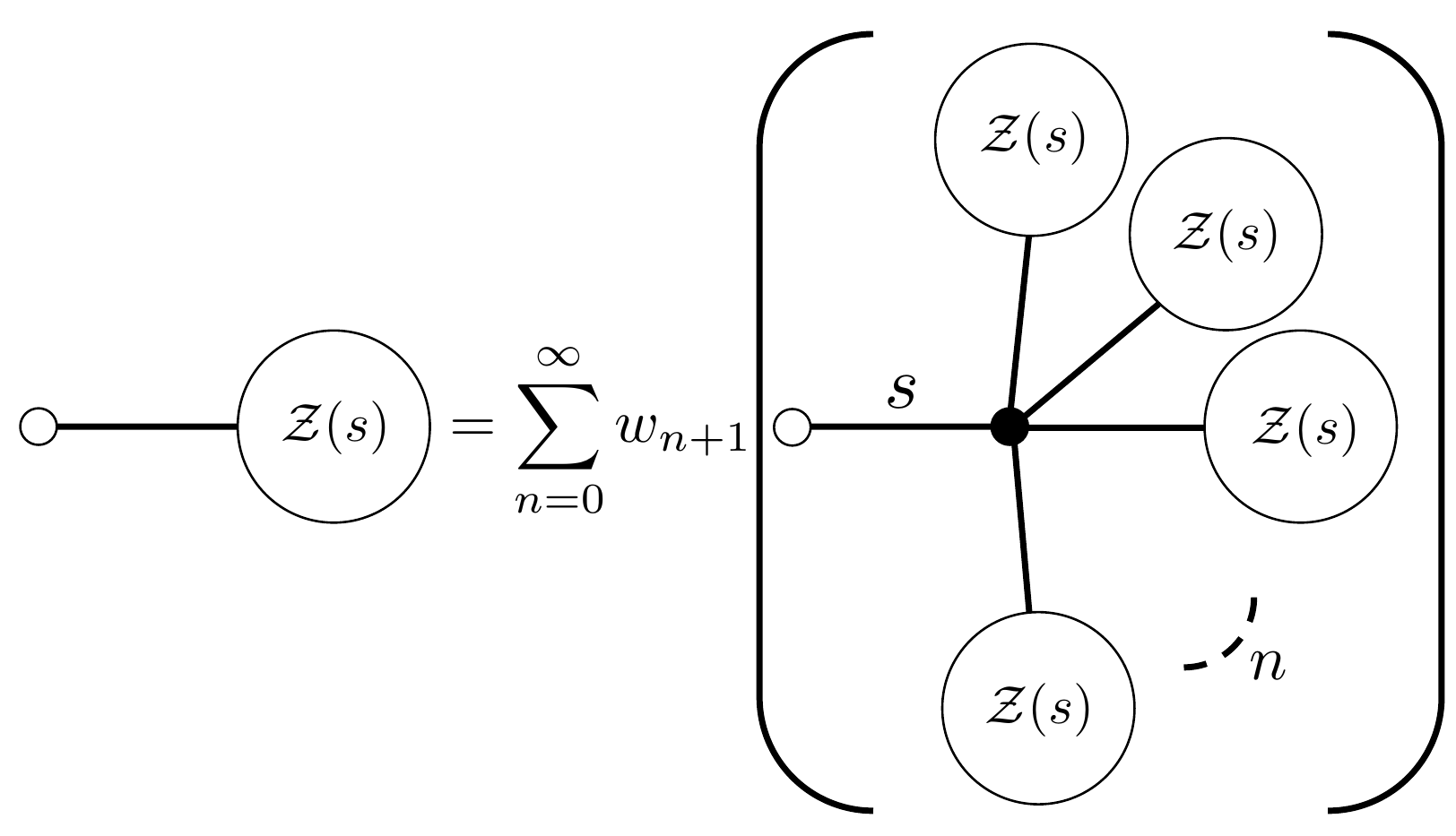}
\caption{A graphical interpretation of eq. \eqref{Z_and_g_relation}.}
\label{Z_decomposition}
\end{SCfigure}
\beq \label{Z_and_g_relation}
\C{Z} (s) = s g (\C{Z} (s))
\eeq
We also define $\C{Z}_0 = \lim_{s \to s_0} \C{Z}(s)$. Now we would like to find the equivalent probability weights which correspond to Galton-Watson trees as explained previously. To do so, we use the following lemma \cite{Janson:2011sg}. 
\begin{lemma}
There exists a probability weight sequence equivalent to $\{w_n\}$ if and only if $\rho >0$. In this case, the probability weight sequences equivalent to $\{w_n\}$ are given~by
\beq
p_n = \frac{z^n w_{n+1}}{g(z)}
\eeq
for any $z>0$ such that $g(z)<\infty$.
\end{lemma}
We apply this lemma for $z = \C{Z}_0$ and find
\beq \label{offspring_prob}
p_n = \frac{\C{Z}_0^nw_{n+1}}{ g \left (\C{Z}_0 \right )} \stackrel{\eqref{Z_and_g_relation}}{=} s_0 w_{n+1} \C{Z}_0^{n-1}.
\eeq
Thus the generating function of the offspring probabilities \eqref{offspring_prob} is given by
\beq \label{gener_fun_prob_weights}
f(z) = \sum _{n=0}^{\infty} p_n z^n = s_0 \sum _{n=0}^{\infty} w_{n+1} \C{Z}_0^{n-1} z^{n} = s_0 \C{Z}_0^{-1} g(\C{Z}_0 z)
\eeq
As we have already commented, if the generating function \eqref{gener_fun_prob_weights} is \textit{analytic in a neighbourhood of the unit disc},  the corresponding tree ensemble is called \textit{generic}. From the right hand side of \eqref{gener_fun_prob_weights} we observe that this definition is equivalent to $\C{Z}_0 < \rho$. Ensembles that do not meet this criterion are called \textit{non-generic}. Also note, that if $\rho = \infty$ then we always have a generic ensemble. 

As we have explained the mean offspring value is given by $m=f'(1)$ and now takes the form
\beq \label{mean_offsping}
m = \sum _{n=1}^{\infty} n p_n = s_0 \sum _{n=1}^{\infty} n w_{n+1} \C{Z}_0 ^{n-1} = s_0 g'(\C{Z}_0) = \C{Z}_0\frac{g'(\C{Z}_0)}{g(\C{Z}_0)} = 1 -\frac{g(\C{Z}_0)}{\C{Z}'(s_0)}
\eeq
where the last two equalities are obtained by using \eqref{Z_and_g_relation} at $s=s_0$. From the very last expression we notice that the process can be either critical or sub-critical. Additionally, for generic trees one can differentiate \eqref{Z_and_g_relation} with respect to $\C{Z}$ and get $\C{Z}_0 g'(\C{Z}_0)=g(\C{Z}_0)$ at $s=s_0$ since $\C{Z}_0<\rho$.  Looking at \eqref{mean_offsping} we conclude that generic trees always have $m=1$ and therefore correspond to critical Galton-Watson process. On the other hand, non-generic trees can be  either critical or sub-critical Galton-Watson processes.

An important and crucial feature of the critical size-conditioned  Galton-Watson tree  is that in the thermodynamic limit it converges to an infinite tree ensemble $\C{S}$, with a unique infinite path from the root (to infinity), called the spine, where the outgrowths from the vertices along the spine are independently and identically distributed (i.i.d.). Formally the following theorem holds \cite{Durhuus:2006vk, Durhuus:2009zz, Durhuus:2009sm}.  
\begin{theorem}
Assume that $\mu_N$ is defined as above as a probability measure on $\C{T}$ where $\{p_n\}$ defines a generic and critical Galton-Watson process. Then, $\mu_N$, converges weakly as $N\to \infty$ to a probability measure $\mu_{\infty}$ on $\C{T}$ concentrated on the set of infinite trees with a \textit{unique infinite spine} $\C{S}$. The outgrowths from the vertices on the spine are independent critical Galton-Watson trees with offspring probabilities given by \eqref{offspring_prob}. In addition, the probability law for finding a vertex on the spine with $k$ finite branches is
\beq
\BB{P}_s (k) = s_0 (k+1) w_{k+2} \C{Z}_0^{k}
\eeq
which has a generating function
\beq
\sum _{k=0}^{\infty} \BB{P}_s (k) x^{k} = f'(x). 
\eeq
\end{theorem}

Critical Galton-Watson trees with finite variance $f''(1)<\infty$ define a generic ensemble, whereas those with infinite variance, $f''(1)=\infty$ belong to the critical non-generic ensembles. 

We call  the ensemble $(\C{S}, \mu _{\infty})$ a generic random tree (GRT). Some useful results on the attributes of the tree ensembles are summarised by
\begin{lemma}
\label{Critical_GW_StandardResults}%
For any critical Galton-Watson ensemble and related  generic random tree ensemble,
\bea 
\expect {\abs{D_k}}{\infty} &=& (k-1)f''(1)+1,\quad k\ge 1,\label{Done}\\
\expect{|B_k|}{\CGW}&=&k,\quad k\ge 1,\label{BGW}\\
\expect{|B_k|}{\infty}&=&\half k(k-1)f''(1)+k,\quad k\ge 1,\label{BT}\\
 \expect {\abs{D_k}^{-1}}{\infty} &=& \mu_{\CGW}(  D_k(T)> 0),\quad k\ge 1.\label{Dinv}\eea
\end{lemma}
\begin{proof}The proofs of \eqref{Done},  \eqref{BGW} and \eqref{BT}  are given in \cite{Durhuus:2006vk}, Appendix 2. The result \eqref{Dinv} uses Lemma 4 and  the proof of  Lemma 5 of \cite{Durhuus:2006vk}. \end{proof}

We will find it particularly useful to consider the \textit{generalised uniform} process $U$ for which%
\bea  \label{Udist}
p^U_k &=&
\begin{cases} 
b,&k=0,\\      
b^{k-1}(1-b)^2,& k\geq 1,
\end{cases}
\eea
with $0<b<1$ 
\footnote{The special case $b=1/2$  reduces \eqref{Udist} to \eqref{pk_induce_uniform_measure} and defines the uniform generic critical random tree and the corresponding UICT \cite{Durhuus:2009sm}.} \cite{Harris:1963, Athreya:2004}. The generating function is

\bea f^U(x)=\sum_{k=0}^\infty p^U_k x^k=\frac{b+(1-2b)x}{1-bx}.\eea
The $r$'th iterate of $f^U$ is  
 \bea f^U_r(x)=\frac{rb - (rb + b - 1)x}{1 - b + rb - rbx} \label{Riterate}\eea
and 
\beq {f^U}'(1) =1,\qquad {f^U}''(1)=\frac{2b}{1-b}.\eeq
We will denote by $\infty U$ the GRT measure associated with $U$ (equivalently this is the Galton-Watson process described by $U$ and constrained never to die out).
It follows from Lemma \ref{Critical_GW_StandardResults} and \eqref{Riterate} that  
\bea \expect {\abs{D_k}^{-1}}{\infty U} &=& \frac{1}{1+ (k-1)b(1-b)^{-1}} =\frac{1}{1+{f^U}''(1)(k-1)/2},\,\, k\ge 1.\label{U:Dinv}\eea

Before we close this section, for completeness sake, we would like to comment on the thermodynamic limit of the measure of conditioned sub-critical non-generic trees. The authors in \cite{Jonsson:2011cn}  proved the existence of the infinite volume measure and showed that it is concentrated on the set of trees of finite diameter with precisely one vertex of infinite degree and the rest of the tree is distributed as a sub-critical Galton-Watson process.

\section{Conclusion and outlook}
\label{graphs_outlook}
In this chapter we argued why the theory of random infinite graphs is essential for studying fractal aspects of random geometry, e.g. the spectral and Hausdorff dimensions. We introduced a useful bijection between two-dimensional causal triangulations and planar trees and concentrated on a particular ensemble of causal triangulations, the UICT, and its relation to the GRT and the theory of branching processes. However, the definitions we presented are not sufficient to deal with the scale dependent spectral dimension problem. Definition \eqref{ds_graph_def} is valid only at large diffusion time, when the walk is much longer than the discretised scale (the length of the edge), and therefore probes only the long distance features of the graph/geometry. In fact, the values we presented for the several ensembles of graphs remain fixed, describe the long scale characteristics of the random graph and are irrelevant to the short scales. Additionally, one would like to perform the continuum limit, probe the short distance physics and show that the dynamical reduction of the spectral dimension is not an artefact of the discretisation. 

The above considerations raise two immediate questions. Firstly, whether there is a rigorous definition of the continuum limit and scale dependent spectral dimension in the context of graph ensembles and secondly whether there might be a reduced version of the full CDT model which is analytically tractable and yet displays behaviour similar to the full model, at  least as far as the spectral dimension is concerned. The first question is answered in the next chapter, while in chapters \ref{multigraphs} and \ref{physics} we deal with the second point.  
\end{chapter}

\clearemptydoublepage
\begin{chapter}{Continuum Random Combs}
\label{combs}

The terminology ``combs" refers to simplified models of trees with comb-like structure (see figure \ref{comb}). They may not be closely related to random triangulations, but they serve as useful toy models to apply new methods to the study of the fractal properties  of random geometry. For example, the authors in \cite{Durhuus:2005fq} used these simplified models of trees to study the spectral dimension of combs by introducing generating function methods. This work was the starting point for applying generating function techniques for the spectral dimension to other graph ensembles too. Random combs were also used as simplified models to explore the thermodynamic properties of the ensemble measure, i.e. the convergence of the probability measure on the set of finite graphs of size $N$, as $N\to \infty$ \cite{Durhuus:2009zz}. The comb-like structure was also important for the study of biased random walks on random combs which serve as an example of transient walks \cite{Elliot:2007br}. 

One can argue that combs serve as an instructive playground for employing new methods. This principle is followed in \cite{Atkin:2011ak}, which we also discuss in this chapter, where random combs were used as simplified toy-models to search for a rigorous definition of the continuum limit of graph ensembles and of the phenomenon of scale dependent spectral dimension.

The basic idea behind any rigorous formulation is very simple. Consider a comb which has a half-line structure up to some characteristic distance and beyond this scale the structure changes drastically. As we have mentioned, the definition \eqref{ds_graph_def} of the spectral dimension probes only the long distance characteristics of the graph, which implies that the value of the spectral dimension encodes information only for the extra structure of this comb. If we could have a definition which is sensitive in both regions of this comb separated by the characteristic scale, then a random walk within the first region would ``feel" a one-dimensional structure, whereas when the walk passes into the second region ``explores" extra structure, which implies different (larger) spectral dimension. This is a heuristic picture of what we would like to model. In addition, we should keep in mind that the short distance characteristics are sensitive to the cut-off scale, i.e. the lattice spacing $a$, which should be taken to zero, i.e. $a\to0$, defining the continuum limit. These two points will be our guide to the study for formal definitions. 

In this chapter we review the results of \cite{Atkin:2011ak}. Particularly, we rigorously formulate the heuristic arguments made above and then apply this formulation to a simple comb ensemble and present results for another two random combs. We finally conclude with a few remarks.

\section{Definitions}
\label{combs_definitions}

The structure of a comb consists of a half line regarded as a graph, called the \textit{spine} of the comb and denoted by $S = \{s_0, s_1, s_2, \ldots\}$, and finite or infinite linear chains of vertices, called the teeth of the comb $T_n = \{t_{n_0}, t_{n_1}, \ldots \},  \ n = 1, 2,\ldots$, which are attached to the vertices $s_n, \ n=1, 2,\ldots$, on the spine except the root $r=s_0$ (figure \ref{comb}). We further assume for convenience that every tooth is attached to the spine at its endpoint. One can see from this structure that the root is the only vertex on the spine which has degree 1 and any other vertex on the spine has degree 3 at most.  

\begin{figure}[t]
  \begin{center}
    \includegraphics[scale=0.5]{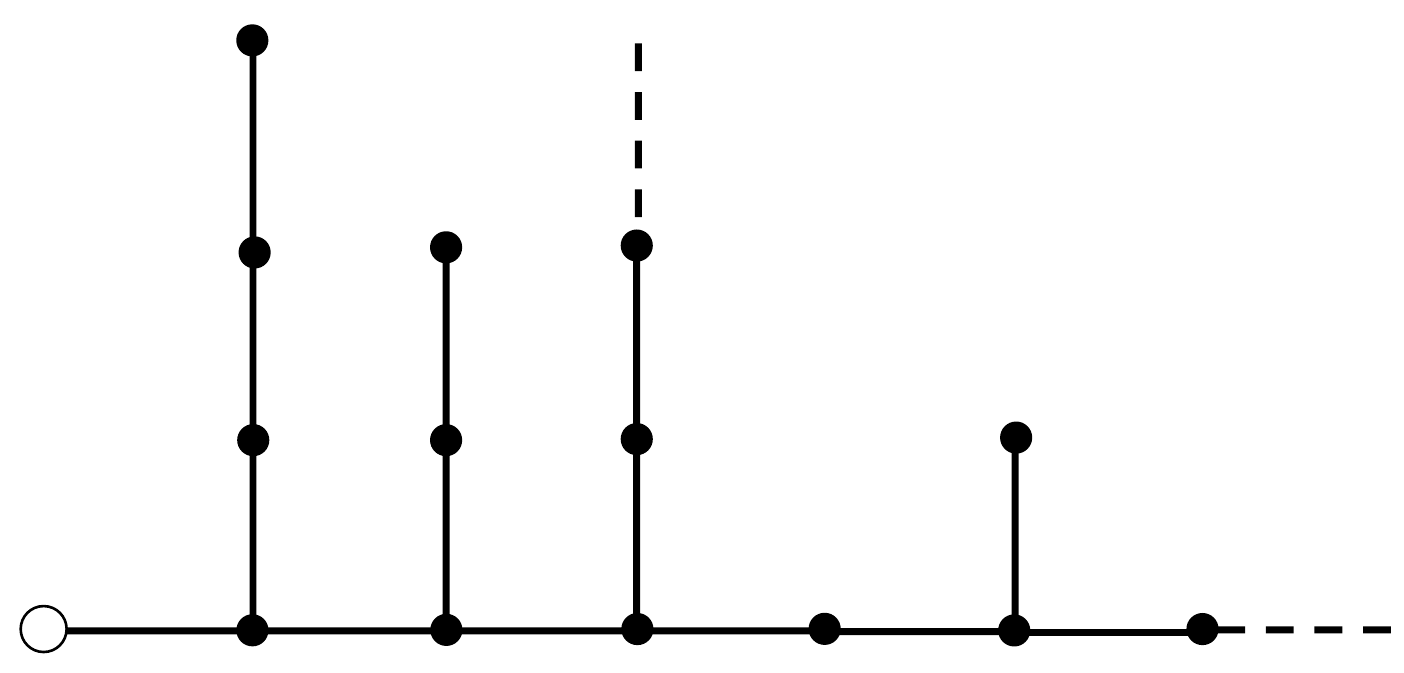}
     \caption{An example of a comb.}
    \label{comb}
  \end{center}
\end{figure}

In addition, we denote by $C_k$ the truncated comb obtained by removing the links $(s_0, s_1), \ldots, (s_{k-1}, s_k)$, the teeth $T_1, \ldots, T_k$ and relabelling the remaining vertices on the spine such that the $k$-vertex of $C$ becomes the root of $C_k$.  For convenient purposes we denote by $C=\infty$  the empty comb in which the spine has no teeth, i.e. a half-line, and $C=*$  the full comb in which every vertex on the spine is attached to an infinite tooth, e.g. $T_k = \infty$ for all $k\geq1$.

A random comb $\R{C} = (\C{C}, \mu)$ is defined as the set $\C{C}$ of all combs equipped with a probability measure $\mu$ on $\C{C}$. The latter is determined by letting the length of the teeth be identically and independently distributed by the measure $\varpi$.  Thus the set $\C{A}$ which consists of combs with teeth at vertices $s_1, s_2,\ldots ,s_k$ having lengths $\ell_1,\ell_2,\ldots ,\ell_k$ has measure
\beq
\mu (\C{A})=\prod_{j=1}^k \varpi(\ell_j).
\eeq

\subsection{Random walks on combs}

In section \ref{ds and dh}, we introduced the generating functions  \eqref{Q_and_P_def} for random walks on graphs. In particular, for a fixed comb $C$, the generating function for first return has a useful decomposition; a random walk steps from the root $r$ to vertex $s_1$ with probability one. At vertex $s_1$, the random walker might step to one of its neighbour vertices with probability $1/3$. If he returns to the root the random walk is over contributing $ (1-x)/3$ to $P_C(x)$. However, being at vertex $s_1$, his next step might be either into the tooth $T_{1}$ or into the truncated comb $C_1$. In either case, the random walker diffuses until he reaches vertex $1$ for the the $n$-th time, contributing to the generating function a factor $\left (\third \left (P_{T_1}(x)+P_{C_1}(x)\right) \right )^n$, where he steps back to the root with probability $1/3$ and the random walk is over. This process implies the recursion relation
\beq \label{rec_rel_P_c}
P_C(x) = \frac{1}{3}(1-x) \sum _{n=0}^{\infty} \left (\third \left (P_{T_1}(x)+P_{C_1}(x)\right) \right )^n = \frac{1-x}{3- P_{T_1}(x)-P_{C_1}(x)}.
\eeq
This result can be applied to both the empty comb (half line) and the full comb. The latter case \eqref{rec_rel_P_c} reads
\beq
P_{*}(x) = \frac{1-x}{3-P_{\infty}(x)-P_{*}(x)}
\eeq
which can be solved for $P_{*}(x)$, since we know $P_{\infty}$(x) from \eqref{P_infty}, yielding
\beq
Q_{*}(x) = x ^{-1/4}, \qquad x \to 0.
\eeq

It is now convenient to state three lemmas which were introduced and proven in \cite{Durhuus:2005fq} and will be useful in our discussion soon. 

\begin{lemma}[Monotonicity lemma A]
\label{MonoLem1}
The generating function $P_C(x)$ is a monotonic increasing function of $P_{T_k}(x)$ and $P_{C_{k}}(x)$ for any $k\geq 1$.
\end{lemma}
\begin{lemma}[Monotonicity lemma B]
\label{MonoLem2}
The generating function $P_C(x)$ is a decreasing function of the length, $\ell_k$, of the tooth $T_k$ for any $k\geq 1$.
\end{lemma}
\begin{lemma}[Rearrangement lemma]
\label{RearrangeLem1}
Let $C'$ be the comb obtained from $C$ by swapping the teeth $T_k$ and $T_{k+1}$. Then $P_C(x)>P_{C'}(x)$ if and only if $ P_{T_k}(x)> P_{T_{k+1}}(x)$.
\end{lemma}

As a consequence of these Lemmas one sees that the generating function of any fixed comb $C$ lies between the half line and the full comb, i.e.
\beq
\label{P_and_Q_c_bounds} 
P_*(x) \leq P_C(x)\leq P_\infty(x),  \quad \text{and} \quad x^{-\quarter}\leq Q_C(x)\leq x^{-\half},
\eeq
as $x\to0$, which implies that its spectral dimension (provided it exists) should satisfy $1\leq\ds \leq 3/2$ \cite{Durhuus:2005fq}.

\subsection{Two-point functions}
Let us introduce the probability $p_C(t;n)$ that a random walk that starts at the root at time $t'=0$ is at vertex $s_n$ on the spine at time $t'=t$ having not returned to the root in the time interval $0< t' \leq t$. The two-point function $G_C(x;n)$ is defined as the generating function for these probabilities, i.e.  
\bea
G_C(x;n) = \sum^{\infty}_{t=1} (1-x)^{t/2} p_C(t;n).
\eea
$G_C(x;n)$ has a useful decomposition in terms of first return generating functions of the truncated combs $C_k, \ 0 \leq k \leq n-1$ \cite{Durhuus:2005fq}. The idea is the following. Any random walk $\Omega$ on a comb $C$ which contributes to the two-point function can be decomposed into a sequence of $n$ random walks $\Omega_1, \Omega_2, \ldots, \Omega_n$. A walk contributing to $\Omega_k$ consists of the part of the walk in $\Omega$ from $s_{k-1}$ to $s_k$ which does not revisit $s_{k-1}$ and might reach $s_k$ multiple times. We now observe that if we add a last step to $\Omega_k$ back to vertex $s_{k-1}$ we reconstruct a walk which returns back to the vertex $s_{k-1}$ for the first time. This walk contributes to $P_{C_{k-1}} (x)$. The extra step contributes to the generating function a factor $(1-x)^{1/2}/\sigma(k)$. To compensate the addition of the extra step we should divide by the same amount. Therefore we write 
\bea \label{G_decomp_P}
G_C(x;n)= \prod _{k=0}^{n-1} \frac{P_{C_k}(x)/\sigma(k)}{(1-x)^{1/2}/\sigma(k+1)} = 
\sigma(n)(1-x)^{-n/2}\prod_{k=0}^{n-1}P_{C_k}(x).
\eea
Using Lemma \ref{MonoLem2} we obtain the bounds,
\bea \label{G_bounds}
\frac{G_*(x;n)}{3} \leq \frac{G_C(x;n)}{\sigma(n)}\leq \frac{G_\infty(x;n)}{2}.
\eea

In addition, it is useful to introduce the probability, $r_C(t;n)$, that a random walk that starts at the root at time $t'=0$ is at the vertex $s_n$ on the spine for the \textit{first} time at time $t'=t$ without visiting the root in the intermediate time $0< t'\leq t$. The corresponding \textit{modified} two-point function is thus defined by 
\bea 
G^{(0)}_C(x;n) = \sum^{\infty}_{t=1} (1-x)^{t/2} r_C(t;n)
\eea
which satisfies the bounds   
\bea \label{mG_bounds}
G^{(0)}_{*}(x;n)\leq G^{(0)}_{C}(x;n) \leq G^{(0)}_{\infty}(x;n)
\eea
and is analogous to \eqref{G_bounds}. The proof of \eqref{mG_bounds} is similar to the proof of \eqref{G_bounds} and is described in Appendix \ref{Appendix_combs}.


\section{Defining the continuum limit}
\label{defining the continuum limit}

In this section we formulate the heuristic picture we described in the introduction. Clearly if we consider short random walks on a graph we only see the local discrete structure and cannot expect any scaling behaviour. For this reason, we need a model in which all walks are long in graph units but there must be a characteristic distance scale $\La$ which is continuously variable and sets a distance scale relative to which walks can either be short or long. Hence $\La$ should also be large in graph units. We introduce the characteristic scale in the ensemble through the measure $\varpi(\ell ; \Lambda_i)$ which might be a function of more than one characteristic scales $\Lambda_i$. 
 
To define the continuum limit we scale the discrete quantities and relate them to their continuum counterparts, i.e. $x=a\xi$ and $\La _i= a^{-\Delta_i} \lambda_i^{\Delta_i}$. Then the continuum generating function is defined by \cite{Atkin:2011ak, Atkin:2011ks}

\begin{lemma}\label{Lemma:scaling}
Assume that there exist constants $\Delta_\mu$ and $\Delta$ such that
\beq  \label{tildeQ:definition}
\tilde Q(\xi;\lambda)=\lim_{a\to 0} a^{\Delta_\mu} \avg{Q\left(x=a \xi ;\Lambda=a^{-\Delta}\lambda^\Delta\right)}_{\mu}
\eeq
exists and is non-zero and the combination $\xi\lambda$ is dimensionless. Then \\
i) there exists a $\tau_0(\xi;\lambda)$ such that 
\bea
 \tilde Q(\xi;\lambda)(1-e^{-\xi\lambda})< \lim_{a\to 0} a^{\Delta_\mu}\sum_{t=0}^{\lfloor\tau_0/a\rfloor}
\BB{P}\left(\{\omega(t)=n\, \vert \, \omega(0)=n\right\})\,(1-a \xi )^{\half t} <  \tilde Q(\xi;\lambda) \nn\\
\eea
and $\tau_0(\infty;\lambda)=\lambda$;\\
ii) there exists a $\tau_1(\xi;\lambda)$ such that 
\bea
\tilde Q(\xi;\lambda)-e\lambda^\half< \lim_{a\to 0} a^{\Delta_\mu}\sum_{t=\lceil\tau_1/a\rceil}^{\infty}
\BB{P}\left(\{\omega(t)=n\, \vert \, \omega(0)=n\right\})\,(1-a \xi )^{\half t} <  \tilde Q(\xi;\lambda) \nn\\
\eea
and $\tau_1(0;\lambda)=\lambda$.
\end{lemma}

So, as $\xi\to\infty$ we see that $ \tilde Q(\xi;\lambda)$ describes walks of continuum duration less than $\lambda$ and that as $\xi\to 0$, $ \tilde Q(\xi;\lambda)$ describes walks of continuum duration greater than $\lambda$ (provided it diverges in that limit). The spectral dimensions $d_s^0$ and $d_s^\infty$ in the short and long walk limits respectively are then defined by
\bea 
\label{dsShort}
d^0_s&=& 2\left(1+\lim_{\xi\to\infty} \frac{\log(\tilde Q(\xi;\lambda))}{\log \xi}\right),\\
\label{dsLong}
d^\infty_s&=& 2\left(1+\lim_{\xi\to 0} \frac{\log(\tilde Q(\xi;\lambda))}{\log \xi}\right),
\eea
provided these limits exist. Lemma \ref{Lemma:scaling} and \eqref{dsShort}-\eqref{dsLong} yield the spectral dimension in the case of recurrent random walks. In chapter \ref{multigraphs} we need to apply this formalism for transient graphs with  $\ds\leq 4$. In this case an exactly analogous result relates $ \partial_\xi \tilde Q(\xi;\lambda)$ and $ \partial_x  Q(x;\Lambda)$ such that 
\bea 
d^0_s&=& 2\left(2+\lim_{\xi\to\infty} \frac{\log|\partial_\xi\tilde Q(\xi;\lambda)|}{\log \xi}\right),\nn\\
d^\infty_s&=& 2\left(2+\lim_{\xi\to 0} \frac{\log|\partial_\xi\tilde Q(\xi;\lambda)|}{\log \xi}\right).\label{trdsdef}
\eea

We skip the proof of Lemma \ref{Lemma:scaling} and the interested reader is referred to \cite{Atkin:2011ak, Atkin:2011ks} for further details. Here we are mainly interested in investigating graph ensembles for which the continuum limit \eqref{tildeQ:definition} exists and hence the spectral dimension at short and long distances can be rigorously defined via \eqref{dsShort} and \eqref{dsLong}  respectively. In the next section we will demonstrate in detail how this process works in a  simple comb ensemble. 


\section{A simple continuum comb - the first toy model} 
\label{SimpleExample}
We now apply the above description to the first \textit{toy model} of a random comb and show that walks on different scales give rise to different values of the spectral dimension at long and short distances. 
We start with a simple random comb which can have either infinite or no teeth. Hence the measure of the length of the teeth is given by
\bea \label{easycomb}
\varpi(\ell;\La)&=&\begin{cases}
1-\frac{1}{\La},&\ell=0,\\ \frac{1}{\La},& \ell=\infty,\\
0,&\rm{otherwise.}
\end{cases}
\eea
The introduction of the characteristic length scale indicates that the infinite teeth have an average separation of $\La$. Intuitively we would expect that if a random walker did not move further than a distance of order $\La$ from the root it would not ``feel'' the teeth and therefore would ``experience'' a half-line structure, i.e. $\ds=1$. If however it were allowed to explore the entire comb it would see something roughly equivalent to a full comb and so ``feel'' a larger spectral dimension. 
To formulate this intuition we proceed in two steps. First we sufficiently bound $\bar{Q}(x;\La)$ from above and below so that these bounds depend on $\La$. Second we apply the continuum limit to obtain bounds for $\tilde{Q}(\xi; \lambda)$. 

To determine the lower bound we apply Jensen's inequality 
\footnote{Let $f$ be a convex function and $X$ a random variable with probability distribution $u$ and expectation $\avg{X}$. Then Jensen's inequality states \cite[Section V.8]{Feller:1965b}
\beq \label{Jensens_ineq}
\avg{f(X)} \geq f(\avg{X}).
\eeq}
to \eqref{rec_rel_P_c} and get a lower bound on the ensemble average generating function 
\beq \label{barP}
\bar P(x ; \La) \geq \frac{1-x}{3 - \bar P_T(x;\La) - \bar P(x; \La)},
\eeq
where $\bar P_T(x; \La) = \sum _{\ell=0}^{\infty} \varpi(\ell; \La) P_\ell(x)$ is the mean first return probability  generating function of the teeth of the comb defined by $\varpi(\ell; \La)$. Rearranging \eqref{barP} we obtain
\beq
\bar P(x ; \La) \geq 1- \sqrt{1+x-\bar P_T(x ; \La)}
\eeq
which leads to a lower bound of the mean return probability generating function \cite{Durhuus:2005fq}, \cite[Lemma 6]{Atkin:2011ak},
\beq \label{barQ_lower_bound1}
\bar Q(x ; \La) \geq \left (1+x-\bar P_T(x ; \La) \right )^{-\half}.
\eeq
For the random comb defined by the measure \eqref{easycomb} we get
\bea
\bar{P}_T(x;\La) = 1 - \frac{1}{\La}(1-P_{\infty}(x)) = 1 - \frac{\sqrt{x}}{\La} 
\eea
 which implies
\bea \label{barQ_lower_b}
\bar{Q}(x;\La) \geq \left(\frac{\sqrt{x}}{\La}+x\right)^{-\half}.
\eea



We now proceed with the determination of an upper bound on  $\bar{Q}(x;\La)$. We follow the line of proof in \cite{Durhuus:2005fq} and 
compare a typical comb in the ensemble with the comb consisting of a finite number of infinite teeth at regular intervals. For this reason we define the event
\bea
\C A(D,k) = \{C:D_i \leq D:i=0,...,k\},
\eea
where $D_i$ is the distance between the $i$ and $i+1$ infinite teeth and then write,
\bea 
\label{Qintegral}
\bar{Q}(x;\La) &=& \int_{\C{C}} Q_C(x;\La) d\mu \nn\\
 &=& \int_{\C{C}\backslash\C{A}(D,k)} Q_C(x;\La) d\mu + \int_{\C{A}(D,k)} Q_C(x;\La) d\mu.
\eea
Since the $D_i$ are independently distributed, the probability of the event $\C{A}(D,k)$ is 
\beq \mu(\C A(D,k))=(1-(1-1/\La)^D)^k.\eeq

Consider a comb $C \in \C A(D,k)$; then by Lemmas \ref{MonoLem1},  \ref{MonoLem2} and  \ref{RearrangeLem1}, 
\beq P_C(x;\La) \leq P_{C'}(x),\eeq
where $C'$ is the comb obtained by removing all teeth beyond the $k$ tooth and moving the remaining teeth so that the spacing between each is $D$. Now we split the walks contributing to $P_{C'}(x)$ into  two sets. The set $\Omega_1$ consists of all walks which go no further than the vertex at distance $Dk-1$  from the root on the spine. The second set, $\Omega_2$, consists of those walks which go at least as far as the $Dk$ vertex on the spine. Hence
\beq 
P_{C'}(x) = P^{(\Omega_1)}_{C'}(x) + P^{(\Omega_2)}_{C'}(x).
\eeq 

Noting that the walks contributing to $P^{(\Omega_1)}_{C'}(x)$ do not go beyond the last tooth, using a slight modification of the Lemmas \ref{MonoLem1}-\ref{RearrangeLem1} (see \cite[Lemmas 1-3]{Atkin:2011ak, Atkin:2011ks}) we have
 \beq
 P^{(\Omega_1)}_{C'}(x) \leq P_{*D}(x),\label{bd1}
 \eeq 
 where $*D$ is the comb having infinite teeth regularly spaced and separated by a distance $D$. 
 
 To bound $P^{(\Omega_2)}_{C'}(x)$ from above we use a result proven in \cite{Durhuus:2005fq} and \cite[Lemma 4]{Atkin:2011ak} which reads
 \beq \label{P2bound}
 P_{C}^{(\Omega_2)}(x) \leq 3 x^{-1/2} G_{C}^0(x;N)^2.
 \eeq 
Using \eqref{bd1}, \eqref{P2bound} and \eqref{mG_bounds} we have,
\beq\label{PCupper1}
P_C(x;\La) \leq P_{{*D}}(x) +  3x^{-\half} G^{(0)}_\infty(x; Dk)^2
\eeq
uniformly in $\C A$. $P_{*D}(x)$ and $G^{(0)}_\infty(x; n)$ are given in Appendix \ref{Appendix_combs}. 
Now set $D = \lfloor\tilde{D}\rfloor$ and $k=\lceil\tilde{k}\rceil$, where,
\beq 
\tilde{D} = 2 \La|\log x \La^2|, \qquad \tilde{k}= (x\La^2)^{-1/2}.
\eeq
Since $G^{(0)}_\infty(x; n)$ is manifestly a monotonic decreasing function of $n$  and $ P_{*D}(x)  $ an increasing function of $D$,
\beq\label{bd100}
\bar{Q}(x;\La) \leq x^{-1/2} (1-(1-(1-1/\La)^{\tilde{D}-1})^{\tilde{k}+1}) + Q_{U}(x)(1-(1-1/\La)^{\tilde{D}})^{\tilde{k}}
\eeq
where we have used \eqref{P_and_Q_c_bounds} and
\beq Q_U(x) = \left[1- P_{*\tilde{D}}(x) - 3x^{-\half} G^{(0)}_\infty(x; (\tilde{D}-1)\tilde{k})^2\right]^{-1}.
\eeq

Having bounded $\bar Q(x)$ on both sides, we now apply the continuum limit by setting $x=a \xi$ and $\Lambda = a^{-\half} \lambda ^{\half}$. After a few lines of algebra it becomes evident that the most singular term is of order $a^{-1/2}$ as $a \to 0$. Therefore the continuum limit \eqref{tildeQ:definition} is applied to \eqref{barQ_lower_b} and \eqref{bd100}  for  $\Delta_\mu = 1/2$ and using \eqref{G0_cont} and \eqref{1-P_cont} 
it gives
\beq \label{SimpleUB} 
\xi ^{-\half} \left (1+(\xi \lambda)^{-\half} \right )^{-\half} \leq \tilde{Q}(\xi;\lambda) \leq \xi^{-1/2} F(\xi\lambda),
\eeq
where
\bea F(\xi \lambda)=
\begin{cases}
1+o\left ((\xi \lambda)^{-1}\right), &\xi \lambda\to \infty, \\
(\xi \lambda)^\quarter\sqrt{2\abs{\log (\xi \lambda)}}+O\left ((\xi \lambda )^\half \right), &\xi \lambda \to 0. 
\end{cases}
\eea
It then follows from \eqref{dsShort}, \eqref{dsLong} and \eqref{SimpleUB} that
\bea d_s^0=1,\qquad
d_s^\infty=\threehalves.
\eea
This is the first important result of this chapter. Starting with a random comb with a relative simple structure we demonstrated how our intuition about short and long \textit{continuum} walks can be rigorously formulated giving rise to a scale dependent spectral dimension which varies according to the probing scale. 


\section{Combs with power law measures}
In this section we attempt to generalise the measure on the teeth of the random comb to a power law of the form,
\bea \label{powerdist}
\varpi(\ell;\La)&=&
\begin{cases}
1-\frac{1}{\La},&\ell=0,\\ 
\frac{1}{\La}C_{\alpha} \ell ^{-\alpha},&\ell>0,
\end{cases}
\eea
where $C_\alpha$ is a normalisation constant and as before $\La$ plays the role of a characteristic distance scale.  We restrict our attention in the range $1<\alpha<2$ as it is known that for $\alpha\geq 2$ the random comb has $\ds=1$ in the sense of \eqref{ds_via_Q} \cite{Durhuus:2005fq} and therefore it is not possible to get a spectral dimension deviating from 1 on any  scale. 

The process proceeds with the same steps as described in the previous section with minor modifications. To keep the discussion as simple as possible we are going to skip the technical details of the proofs from now on and we only sketch the guidelines of the argument. We refer the interested reader to \cite{Atkin:2011ak, Atkin:2011ks} for further details. 


The starting point for computing a lower bound on the return probability generating function is expression \eqref{barQ_lower_bound1}, from which it becomes evident that one needs an upper bound on $1-\bar{P}_T(x)$ for the measure \eqref{powerdist}.  
Such an upper bound is determined by the cumulative probability function $\chi(u;\La_i) = \sum^{[u]}_{\ell=0} \varpi(\ell;\La_i)$ \cite[Lemma 7]{Atkin:2011ak}. For the measure \eqref{powerdist} it takes the explicit form
\bea
1-\bar{P}_T(x;\La) \leq  m_\infty(x)\sqrt{x}\left(\frac{b_1}{\La}m_\infty(x)^{\alpha-2}+\frac{b_2}{\La}m_\infty(x)^{\alpha-1}+\frac{b_3}{\La}\right),
\eea
where  $b_{1,2,3}$ are constants depending only on $\alpha$ with $b_1>0$ and $m_\infty(x) \equiv \half \log \frac{1+\sqrt{x}}{1-\sqrt{x}}$.


An upper bound is obtained by a slight modification of the proof of the upper bound in the previous section. Here we replace the infinite tooth with a  \textit{long} tooth defined as the tooth whose length is greater than $H$. The probability that a long tooth occurs is   
\beq \label{prob_long_tooth}
p = \sum^{\infty}_{\ell = H+1} \varpi(\ell;\La_i).  
\eeq
We define the event as before
\bea
\label{Aevent}
\C A(D,k) = \{C:D_i \leq D:i=0,...,k\}
\eea
where now $D_i$ is the distance between the $i$ and $i+1$ \textit{long} teeth.  Since the $D_i$ are independently distributed
\beq
\label{nuD} \mu(\C A(D,k))=(1-(1-p)^D)^k.
\eeq
The proof proceeds as before, starting from \eqref{Qintegral}, bounding $P_{C\in \C{A}(D,k)}(x;\La_i)$ from above by the comb whose all teeth except the first $k$ long teeth have been removed and the remaining teeth have been truncated to have length $H$ at equal distances $D$. Further, we split the random walks in two sets, those which never reach vertex $Dk$ on the spine and those which go beyond $Dk$.  
We finally arrive at \cite[Lemma 8]{Atkin:2011ak}
\beq
\label{GeneralUB}
\bar{Q}(x;\La_i) \leq x^{-1/2} (1-(1-(1-p)^{{D}})^{{k}}) + Q_{U}(x)(1-(1-p)^{{D}})^{{k}} ,
\eeq
where
\beq
\label{GeneralQU}
Q_U(x) = \left[1- P_{{H},*{D}}(x) - 3x^{-\half} G^{(0)}_\infty(x; {D}{k})^2\right]^{-1},
\eeq
and $C=H, *D$ denotes  the comb with teeth of length $H$ equally spaced at intervals of $D$.

We now specialise to the power law measure \eqref{powerdist} and set $H=\lfloor\tilde{H}\rfloor$, $D = \lfloor\tilde{D}\rfloor$ and $k=\lceil\tilde{k}\rceil$, where
\bea
\label{LambdaK}
\tilde{H} &=& x^{-1/2}, \nn\\
\tilde{D} &=& (\Delta'+1) \frac{\alpha-1}{C_\alpha} x^{\Delta'-1/2} \La |\log x \La^{1/\Delta'}|,\\
\tilde{k} &=& (x \La^{1/\Delta'})^{-\Delta'}.\nn 
\eea

Choosing  $\La= a^{-\Delta'} \lambda^{\Delta'}$ with $\Delta' = 1-\alpha/2$, scaling the expressions $P_{H, *D}$ and $G_\infty^{(0)}$ and taking the continuum limit, yields the continuum return generating function,
\beq
\label{PowerLB}
\xi ^{-1/2}\left(1+b_1 (\xi\lambda)^{-(1-\alpha/2)}\right)^{-1/2} \leq \tilde{Q}(\xi ; \lambda)\leq \xi^{-1/2} F(\xi \lambda)
\eeq
where
\beq
F(\xi \lambda) =
\begin{cases}
1+o\left ((\xi \lambda)^{-1}\right), &\xi \lambda \to \infty,\\
c\, (\xi \lambda)^{1/2-\alpha/4}  \sqrt{\abs{\log (\xi \lambda)}}+O\left((\xi \lambda)^{\Delta'}\right), &\xi \lambda \to 0.
\end{cases}
\eeq
Expression \eqref{PowerLB} implies that the spectral dimension varies from the value $(4-\alpha)/2 >1$ at the long walk limit, i.e. $\xi \to 0$, to  $1$ at short distances, i.e. $\xi \to \infty$. In summary, the random comb with the power law measure \eqref{powerdist} for the tooth length exhibits 
\beq 
d_s^0=1,\qquad d_s^\infty=2-\frac{\alpha}{2}.
\eeq


\section{Multiple Scales}
\label{multiple_scales}

In the light of the results so far we investigate the case of a random comb which has more than one characteristic scale. This is achieved through a generalisation of \eqref{powerdist}  having a double power law distribution,
\bea \label{2powerdist}
\varpi(\ell;\La_i)&=&
\begin{cases}
1-\La_1^{-1}-\La_2^{-1},&\ell=0,\\ 
\frac{1}{\La_1}C_{1} \ell ^{-\alpha_1}+\frac{1}{\La_2} C_{2} \ell ^{-\alpha_2} ,& \ell>0.
\end{cases}
\eea
To keep the discussion simple and unambiguous we assume without loss of generality that the continuum length scales $\lambda_i$ satisfy the hierarchy $\lambda_1< \lambda_{2}$ and that $1<\alpha_i<2$.  

To find the continuum generating function, we follow the steps of the previous section subject to modifications due to the measure \eqref{2powerdist}.  For example the lower bound on $\bar Q(x;\Lambda)$ depends on the upper bound
\beq 
1-\bar{P}_T(x;\La_i) \leq  m_\infty(x)\sqrt{x}\sum_{i=1}^2\left(\frac{b_{1i}}{\La_i}m_\infty(x)^{\alpha_i-2}+\frac{b_{2i}}{\La_i}m_\infty(x)^{\alpha_i-1}+\frac{b_{3i}}{\La_i}\right).
\eeq
An upper bound is obtained by repeating the discussion from expression \eqref{prob_long_tooth} to \eqref{GeneralUB} subject to the measure \eqref{2powerdist} and setting $H=\lfloor\tilde{D}\rfloor$, $D = \lfloor\tilde{D}\rfloor$ and $k=\lfloor \tilde{k}\rfloor$, where
\bea
\label{DoubleLambdaK}
\tilde{H} &=& x^{-1/2} \nn\\
\tilde{D} &=& \beta x^{-1/2} G(x \La_1^{1/\Delta'_1},x \La_2^{1/\Delta'_2})^{-1} |\log x \La_1^{1/\Delta'_1}|, \\
\tilde{k} &=& G(x \La_1^{1/\Delta'_1},x \La_2^{1/\Delta'_2}) \nn 
\eea
where we have introduced the function,
\beq
G(\upsilon_1,\upsilon_2) = \frac{C_1}{\alpha_1-1} \upsilon_1^{-\Delta'_1}+\frac{C_2}{\alpha_2-1} \upsilon_2^{-\Delta'_2} .
\eeq

Scaling $x=a \xi$  and $\La_i=a^{-\Delta'_i} \lambda_i^{\Delta'_i}$, where $\Delta'_i = 1-\alpha_i/2$, we get the continuum return generating function
\bea
\label{DoubleLB}
{\xi^{-1/2}\left(c_0 +c_1 (\xi \lambda_1)^{-(1-\alpha_1/2)}+c_2( \xi\lambda_2)^{-(1-\alpha_2/2)} \right)^{-1/2} \leq \tilde{Q}(\xi ; \lambda_i)} \leq \nn\\
\xi^{-1/2}\Bigg[1-(1-(\xi \lambda_1)^{-s\beta})^G  \nn\\
+ \frac{ (1-(\xi \lambda_1)^{-s\beta})^{G-1} }{-\gamma+\sqrt{\gamma^2+1+2\gamma \coth(|\log(\xi \lambda_1)^\beta|/G)} - 3\mathrm{cosech}^2(|\log(\xi \lambda_1)^\beta|(1-1/G))} \Bigg] 
\eea
in which we have suppressed the arguments of $G(\upsilon_1,\upsilon_2)$.  

We carefully examine  \eqref{DoubleLB} and deduce the behaviour of $\tilde{Q}(\xi;\lambda_i)$ on various length scales. 
\begin{itemize}
\item[1.] Taking the walk length to be smaller than the smaller characteristic scale $\lambda _1$, meaning $\xi\gg\lambda_1^{-1}$, both upper and lower bounds of $\tilde Q(\xi;\lambda_i)$ are dominated by the $\xi^{-\half}$ behaviour. Hence taking the short walk limit, $\xi\to\infty$, gives $d_s^0=1$ as in the previous sections independent of the relative relation between $\alpha_1$ and $\alpha_2$.

\item[2.] Consider the case with $\alpha_1 <  \alpha_2$. Then  both upper and lower bounds of $\tilde Q(\xi;\lambda_i)$ are dominated by the $\xi^{-\alpha_1/4}$ behaviour in the long walk limit, $\xi\to 0$,  leading to $d_s^\infty=2-\alpha_1/2 $. There is no regime in which $\alpha_2$ controls the behaviour.  

\item[3.] Finally, consider the case where $\alpha_2 <  \alpha_1$. The behaviour of $\tilde Q(\xi;\lambda_i)$ is now dominated by $\xi^{-\alpha_2/4}$ in long walk limit. However, one observes an intermediate regime $\lambda_3^{-1}\ll \xi\ll \lambda_1^{-1}$, where the scale $\lambda_3^{-1}$ is given by
\beq 
\lambda_3^{-1}=\lambda_1^{(2-\alpha_1)/(\alpha_1-\alpha_2)}\lambda_2^{(2-\alpha_2)/(\alpha_2-\alpha_1)}.
\eeq
For continuum random walks much larger than $\lambda _1$, but still much shorter than $\lambda_3$, the  asymptotic behaviour $\xi^{-\alpha_1/4}$ is dominant
\footnote{There will be corrections of order $\xi ^\beta$, but we can choose $\beta$ to suppress those corrections.}.  
In this case, although the system exhibits a reduction from $d_s^\infty = 2-\alpha_2/2>1$ to $d_s^0=1$, there seems to be an additional regime of walks of length $\lambda_1\ll\xi ^{-1}\ll\lambda_3$ in which the system manifests a spectral dimension $\delta_s=2-\alpha_1/2$. We will refer to a spectral dimension that appears in this way as an {\emph{apparent spectral dimension}} and denote it by $\delta_s$ rather than $d_s$. The reason is that any statement for the behaviour of $\tilde Q(\xi; \lambda_i)$ in the intermediate scales is weak by itself, since the upper and lower bounds might differ in these scales. However, the hierarchy of scales $\lambda _1 \ll \lambda _2$ gives rise to a constant value of the (apparent) spectral dimension over an extended regime, which is considered only an intermediate plateau between $d_s^\infty$ and $d_s^0$.

This scenario, in which the spectral dimension exhibits an intermediate plateau, has also been considered in the asymptotic safety scenario of gravity \cite{Reuter:2011ah} and multifractional spacetimes \cite{Calcagni:2012qn, Calcagni:2012rm} (further details are presented in section \ref{other_approaches}). However, the latter approaches possess a different spectral dimension profile, where the value of the spectral dimension in the intermediate plateau is less than the values  $d_s^\infty$ and $d_s^0$~
\footnote{We comment later that diffusion on multifractional spacetimes can imitate a monodically increasing profile with one intermediate plateau by an appropriate choice of the fractional charges (section \ref{other_approaches}).}. 
\end{itemize}

\section{Conclusion and outlook}
\label{combs_outlook}
In this chapter we obtained the continuum formulation and results which are significant for several reasons. First, we showed analytically that there exists indeed a class of graph models which exhibit the phenomenon of scale dependent spectral dimension. This is a rather non-trivial result, if one considers the level of randomness -- we are dealing with \emph{random} walks on \emph{random} graphs. Second, the fact that the running spectral dimension survives the continuum limit indicates that this phenomenon is not due to discretisation artefacts weakening the criticism of the numerical results of CDT. 

Despite the importance of these results, one might ask how relevant the random combs are to quantum gravity models. The point is that combs are not realistic toy models, since they include no information about the dynamics of the random geometry. To elaborate this point, consider both the two-dimensional causal triangulations (the UICT) and generic random tree (GRT). Both ensembles are evolved by a \textit{local} weight, which implies that the evolution of the spatial volume with time is generated by a Hamiltonian. In contrast, the random combs with power law tooth length distributions considered in this chapter have no local growth law and the number of vertices at height $k$ does not induce any information on the set of vertices at height $k+1$, disallowing a Hamiltonian description of the random comb
\footnote{However we should notice that the random comb with tooth length distribution of the form $\varpi(\ell) \propto e^{-\ell}$ admits a local growth and Hamiltonian description, but it does not exhibit any reduction of the spectral dimension.}.

Therefore, one realises the need of proceeding with more realistic graph toy models which capture more features of the CDT geometry and include dynamics too. Being armed with the formalism described in this chapter, this is now possible in the view of mappings defined in section \ref{from triangulations to trees}.
\end{chapter}

\clearemptydoublepage
\begin{chapter}{Continuum Random Multigraphs}
\label{multigraphs}

At the end of chapter \ref{graphs} we set two important questions. In chapter \ref{combs}, we  answered the first question by developing a rigorous definition of the continuum limit of random graphs which can yield a scale dependent spectral dimension. However the toy models we studied there can be considered as ``kinematical" because they do not have a completely \textit{local growth} rule and are therefore not directly related to any model of quantum gravity. In principle, one would like to apply those techniques to more realistic toy models and address the second question raised previously. It is thus the aim of this chapter to introduce such realistic models, so called multigraph ensembles, which are approximations of CDTs, inherit some of their features and include dynamics too. We analyse properties of the spectral dimension of a variety of multigraph ensembles.    
We start the discussion by defining the multigraph ensemble, its probability measure and stating some of its properties. Then we apply the continuum limit formalism for the recurrent multigraphs.

In section \ref{transient case} we study the fractal properties of transient multigraphs, reduced models of higher-dimensional CDT. We then proceed by discussing the lessons from the instructive recurrent case where we have full analytical control on the ensemble measure. This discussion serves as the motivation for the assumptions we adopt in the study of the continuum limit of the transient multigraph ensemble.

\section{Multigraphs, causal triangulations and trees}
A multigraph $M$ is defined by introducing a mapping which acts on a rooted infinite causal triangulation $\tr$ by collapsing all space-like edges at a fixed distance $k, k\geq1$ from the root and identifying all vertices at this distance $k$
\footnote{Notice that the multigraph construction is valid for both two-dimensional and higher-dimensional CDT due to the time-slicing structure.}.
In the resulting multigraph a vertex $k$ has neighbours $k\pm1$, except the vertex $0$ (the root $r$) which has $1$ as a neighbour, and there are $L_k(M) \geq1$ (time-like) edges connecting $k$ and $k+1$ (see figure \ref{ct_to_multigraphs}). 
\begin{figure}
\centering
\includegraphics[scale=0.3]{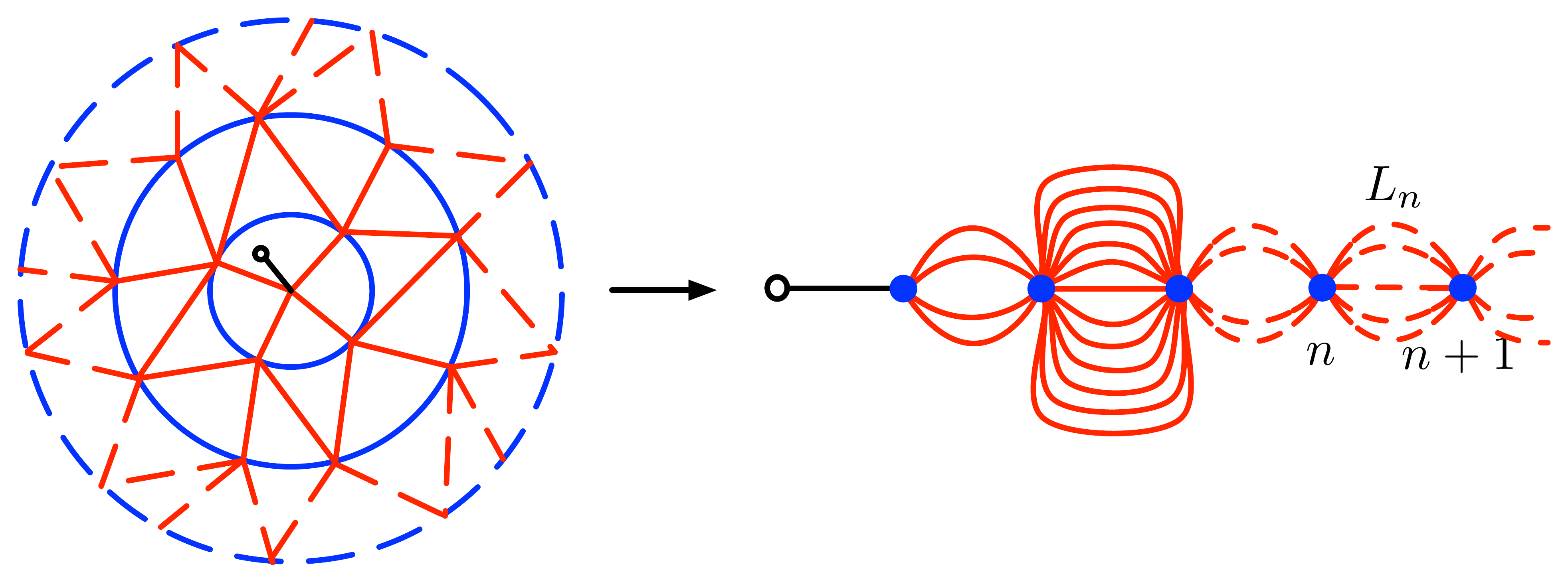}
\caption{Illustration of how to obtain a multigraph from a causal triangulation.}
\label{ct_to_multigraphs}
\end{figure}
We denote the set of multigraphs by $\C M$.  A multigraph $M\in \C M$ is completely described by listing the number of edges $\{L_k(M), k=0,1,\ldots\}$.  
A multigraph ensemble $\R M=\{\C M,\chi\}$ consists of the set of graphs $\C M$ together with a probability measure $\chi(\{M\in \C M: \C A\}) $ for the (finite) event $\C A$. 

Given the bijection between causal triangulations and trees described in section \ref{from triangulations to trees}, multigraphs can also be obtained by defining the mapping $\gamma: \C T_\infty\to\C M$  which acts on a tree $T$ by  identifying all vertices $v\in S_k(T)$ with the single vertex $v$ but retaining all the edges. The resulting multigraph ensemble inherits its measure from the measure on $\C T_\infty$ so that  for  integers $0\le k_1<\dots<k_n$ and positive integers $m_1,\dots, m_n$ 
\bea \label{RRequivGRT}
\lefteqn{\chi(\{M\in \C M: L_{k_i}(M)=m_i,\, i=1\ldots n\})}  \nn\\
&=&\mu(\{T\in\C T_\infty:\abs{D_{k_i+1}(T)}=m_i,\, i=1\ldots n\}).
\eea
This mapping is well defined provided that the measure on $\C T_\infty$ is supported on trees with a unique path to infinity.  It is convenient to use trees to define some of  the ensembles of multigraphs we will be considering because we can thereby exploit many standard results presented in section \ref{GW trees and simply generated trees}. To clarify the terminology, whenever we use the term ``recurrent multigraph ensemble'', we refer to the random multigraph, which is induced by the GRT having measure $\infty U$ and the corresponding two-dimensional infinite causal triangulation and exhibits recurrent random walks, in contrast to the ``transient multigraphs'' which are approximations of higher-dimensional CDT and exhibit transient random walks, as we explore in later sections.
 
One should notice that different causal triangulations or planar trees might induce the same multigraph. Therefore when the mapping is applied some information is lost  and the inverse mapping does not correspond to a unique causal triangulation or tree. As a result those mappings define an \textit{injection} from causal triangulations and trees to multigraphs.

\section{Random walks on multigraphs}
\label{random walk on mulrigraphs}

Given a fixed multigraph $M$ the probability for a random walker at $n$ to step next to $n+1$ is given by $p_n(M) = L_n/(L_{n-1} + L_n)$, according to \eqref{prob_next_step}, and the probability that the next step is to $n-1$ is $1-p_n(M)$ (note that the probability to move from the root to vertex $1$ is one). Then we decompose the random walk into two pieces; a step from vertex  $n$ to  $n+1$, then a random walk returning to $n+1$ and a final step from $n+1$ to n at time $t$. Because of this decomposition and the chain structure the generating function of first return to the root satisfies the following recursion relation \cite{Durhuus:2009sm}
\beq \label{P_recurr_mgraph}
P_{M_n}(x)=\frac{(1-x)(1-p_{n+1}(M))}{1-p_{n+1}(M)P_{M_{n+1}}(x)},
\eeq
where $M_n$ is the multigraph obtained from $M$ by removing the first $n$ vertices and all edges attached to them and relabelling the remaining multigraph 
\footnote{For instructive reasons, we apply this formalism to rederive some known results in Appendix \ref{Appendix_multigraphs}.}.
Introducing the notation $\eta _{M_n} \equiv Q_{M_n}/L_n$, expression \eqref{P_recurr_mgraph} can be rearranged to give
\beq \label{eta_rec_rel}
\eta _{M_n}(x) =\eta _{M_{n+1}}(x) + \frac{1}{L_n} -x L_{n} \eta _{M_n}(x)\eta _{M_{n+1}}(x).
\eeq

Recursion relation \eqref{P_recurr_mgraph} leads  to the following useful results.
\begin{lemma}{\bf (Monotonicity)}
\label{monotonicity} 
For any $M\in\C M$ and $M'\in\C M$ which are identical except that  $L_k({M'})=L_k({M})-1$ for some $k>0$,
\beq P_{M}(x) < P_{M'}(x).\eeq
\end{lemma}
\begin{proof} First note that from \eqref{P_recurr_mgraph} $P_{M_{k-1}}(x)$ is a monotonically increasing function of $P_{M_{k}}(x)$ and therefore $P_M(x)$ is a monotonically increasing function of $P_{M_{k}}(x)$; it then suffices to prove the $k=1$ case. This is easily done by applying \eqref{P_recurr_mgraph} twice in succession to express both $P_M(x)$ and $P_{M'}(x)$ in terms of $P_{M_2}(x)$, $L_0(M)$, $L_1(M)$  and  computing the difference. Note that the lemma is not true for the case $k=0$.
\end{proof}
\begin{lemma}{\bf (Rearrangement)}
\label{rearrangement}
For any $M\in\C M$ and $M'\in\C M$ which are identical except that $L_k(M) = L_{k+1} (M')$ and $L_{k+1}(M) = L_{k} (M')$, $k>0$, then
\beq
P_M(x) \leq P_{M'} (x)  \qquad \text{only if} \qquad L_k (M) \geq L_{k+1} (M). 
\eeq
\end{lemma}
\begin{proof}
The proof is similar to the proof above. It suffices to prove the case $k=1$. We apply twice \eqref{P_recurr_mgraph} and get both $P_M(x)$ and $P_{M'}(x)$ in terms of $P_{M_3}(x)$, $L_i(M)$, $i = 0, \ldots, 3$  and  computing the difference. Note that the lemma is not true for the case $k=0$.
\end{proof}
An immediate consequence is that
\beq \label{eta:upper}
\eta_M(x) \equiv \frac{Q_M(x)}{L_0(M)}<\frac{1}{ x^{1/2}}.
\eeq
To obtain this we reduce all $L_{n>0}$ to 1, repeatedly applying Lemma \ref{monotonicity} to obtain the upper bound 
\beq P_M(x)<P_{M^*}(x)\eeq
where $M^*$ is the graph with edge numbers  $\{L_0(M),1,1,1,\ldots\}$ and then compute  $P_{M^*}(x)$ explicitly by the methods of \cite{Durhuus:2005fq}. 
 Since by definition $Q_M(x)\ge 1$ we also have 
that
\beq
\eta_M(x)\ge\frac{1}{L_0(M)}.\label{eta:lower}
\eeq


\section{Scale dependent spectral dimension in the recurrent case}  
\label{Scale dependent spectral dimension in the recurrent case}

The model we consider in this section is a multigraph ensemble whose measure is related to the generalised uniform GRT measure $\infty U$ through \eqref{RRequivGRT} and whose graph distance scale is therefore set by the parameter $b$. The weighting of the multigraphs in this ensemble can be related to an  action for the corresponding  CDT ensemble which contains a coupling to the absolute value of the scalar curvature. To see this first note that the probability for a finite tree $T$ in the $U$ ensemble, defined by \eqref{Udist}, is given by
\bea 
b^{\sum_{v\in T}\abs{\sigma(v)-2}}\,(1-b)^{2\sum_{v\in T}1-\delta_{\sigma(v),1}}.
\eea
The quantity $\abs{\sigma(v)-2}$ is in fact the absolute value of the two-dimensional scalar curvature; small $b$ suppresses all values of vertex degree except $\sigma(v)=2$. The causal triangulation $\tr$ which is in bijection to this tree  has probability
\bea 
b^{\sum_{v\in \tr}\abs{\sigma (f(v))-2}}\,(1-b)^{2\sum_{v\in \tr}1-\delta_{\sigma (f(v)),1}},
\eea
where $\sigma(f(v))$ is the number of  `forward' edges connecting vertex $v$ to vertices whose distance from the root is one greater than that of $v$.  Taking into account the causal constraint the effect of small $b$ on the triangulation is to suppress all vertex degrees except $\sigma(v)=6$ so that $-\log b$ plays the role of a coupling to a term $\sum_v\abs{R_v}$, essentially the integral of the absolute value of the scalar curvature, in the action for the CDT.

The main result of  this section is
\begin{theorem}
\label{VariableSpectralDimensionRecurrent}
The scaling limit of the multigraph ensemble  with measure $\infty U$
has spectral dimension $d_s^0=1$ at short distances and  $d_s^\infty=2$ at long distances.
\end{theorem}

\begin{proof}  Taking $\Lambda=b^{-1}$ and $\Delta_\mu=\Delta=\half$ in \eqref{tildeQ:definition} we obtain
   \bea\label{newthm1}
\tilde Q(\xi;\lambda)&\sim&
\begin{cases}
 \xi^{-\half},&\xi\gg\lambda^{-1},\\ \lambda^\half\abs{\log \lambda\xi},& \xi \ll\lambda^{-1},
\end{cases}
\eea
the proof of which follows. The theorem then follows from definitions in section \ref{defining the continuum limit},   
\eqref{dsShort} and \eqref{dsLong}.
\end{proof}

\subsection{Short distance behaviour: $\xi\to\infty$}

By the monotonicity lemma we have 
\beq Q(x;b)<x^{-\half}\eeq
from which it follows immediately that
\beq \tilde Q(\xi;\lambda) < \xi^{-\half}.\eeq
We can obtain a lower bound for the expectation value by considering only the contribution of graphs for which the first $N$ vertices have only one edge and walks which get no further than $N$; then using \eqref{PhalflineLf} we get
\beq 
\expect{Q(x;b)}{\infty U}>(1-b)^{2N}x^{-\half}\frac{(1+x^\half)^N-(1-x^\half)^N}{(1+x^\half)^N+(1-x^\half)^N}.
\eeq
Setting $N=b^{-1}$ we find 
\beq  \tilde Q(\xi;\lambda) > \xi^{-\half} e^{-2}\tanh(\sqrt{\xi\lambda})\eeq
which establishes  the first part of \eqref{newthm1}.

\subsection{Lower bound as $\xi\to 0$}

From now on to improve legibility  we will adopt the following simplified notation whenever it does not lead to ambiguity; for $\eta_{M_n}(x)$ we will write $\eta_n(x)$ and for $L_n(M)$ we will write $L_n$. We will also suppress the second argument $b$ in $\eta_n(x)$.

Rearranging \eqref{eta_rec_rel} with respect to $\eta_n(x)$ we obtain
\beq 
\eta_n(x)=\frac{\eta_{n+1}(x)+\frac{1}{L_n}}{1+xL_n\eta_{n+1}(x)}\label{start}
\eeq
which can be iterated to give 
\bea 
\eta_n(x)&=&\frac{\eta_{N}(x)}{ \prod_{k=n}^{N-1}(1+xL_k\eta_{k+1}(x))}   +\sum_{k=n}^{N-1}\frac{1}{L_k}\frac{1}{\prod_{m=n}^{k}(1+xL_m\eta_{m+1}(x))}\nn\\
		&>&\sum_{k=n}^{N}\frac{1}{L_k}\exp \left (-\sum_{m=n}^{k}xL_m\eta_{m+1}(x)\right ),\label{EtaLower}
\eea
where we have used $\eta_N(x)\ge 1/L_N$. Using the monotonicity lemma \eqref{eta:upper}  and Jensen's inequality \eqref{Jensens_ineq}  gives
\beq 
\expect{\eta_0(x;b)}{\infty U}>
\sum_{n=0}^{N}\frac{1}{\expect{L_n}{\infty U}}\prod_{k=0}^{n}e^{-\sqrt{x}\expect{L_k}{\infty U}}.
\eeq
Using the results in Lemma \ref{Critical_GW_StandardResults} and \eqref{U:Dinv} gives 
\bea 
\expect{\eta_0(x;b)}{\infty U}&>&
\sum_{n=0}^{N}\frac{1}{1+{f^U}''(1)n/2}
e^{-\sqrt{x}\sum_{k=0}^{n}\left (1+{f^U}''(1)k \right )}\\
&=&\sum_{n=0}^{N}\frac{1}{1+{f^U}''(1)n/2}
e^{-\sqrt{x}(n+1)(1+\half {f^U}''(1)n)}\nn\\
&>&e^{-\sqrt{x}(N+1)(1+\half {f^U}''(1)N)}  \sum_{n=0}^{N}\frac{1}{1+{f^U}''(1)n/2}\nn\\
&>&\frac{2}{{f^U}''(1)} e^{-\sqrt{x}(N+1)(1+\half {f^U}''(1)N)}\log \left(1+{f^U}''(1)N/2\right).
\eea
Now let $N=\lfloor b^{-\half} x^{-\quarter}\rfloor$ and set  $b=a^\half\lambda^{-\half}$, $x=a\xi$ so
\bea 
\tilde Q(\xi;\lambda)&=&\lim_{a\to 0}a^\half \expect{\eta_0\left (x=a\xi; b = a^{1/2}\lambda^{-1/2} \right)}{\infty U} \\
&>&
{\lambda^\half}e^{-1-(\xi\lambda)^\quarter}\log\left( 1+\frac{1}{(\xi\lambda)^\quarter}\right)\label{longlower}
\eea
which diverges logarithmically as $\xi\to 0$.

\subsection{Upper bound as $\xi\to 0$}

This proceeds by a fairly standard argument.
Let $p_t(r;n)$ be the probability that a walk starting at the root at time zero is at vertex $n$ at time $t$ and define
\beq 
Q(x;n)=\sum_t p_t(r;n)  (1-x)^{\half t}.
\eeq
Then
\beq \sum_{n\in B(R)} Q(x;n) <\frac{2}{x}\eeq
so it follows that there is a vertex $v\le R$ such that
\beq Q(x;v)< \frac{2}{x R}.\eeq
Now consider the walks contributing to $Q_M(x)$ and split them into two sets; $\Omega_1$ consisting of those reaching no further than $v-1$,  and $\Omega_2$ consisting of those reaching at least as far as $v$  (see Figure \ref{omega}). Then we have
\begin{figure}
\begin{center}
\includegraphics[width=14cm]{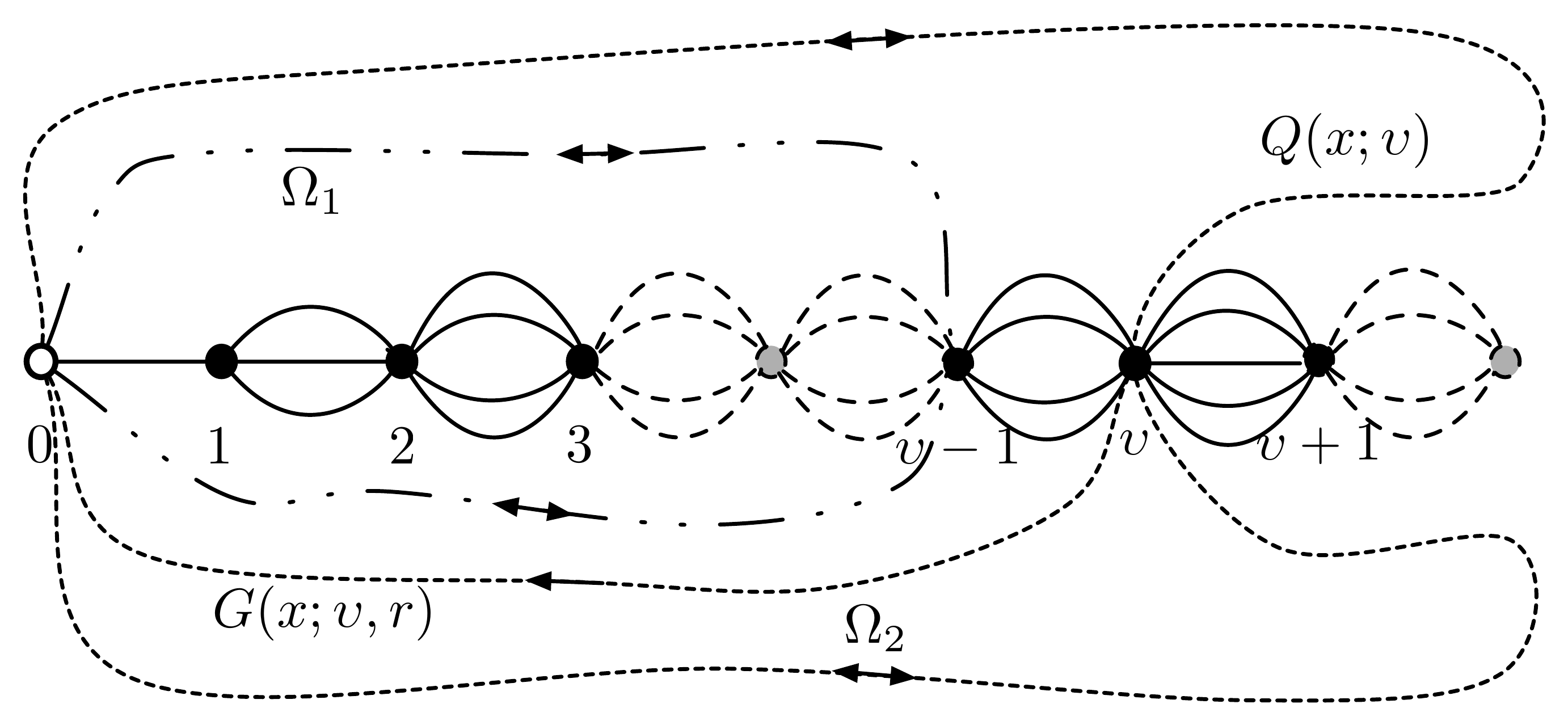}
\caption{Decomposition of a random walk into sets $\Omega _1$ and $\Omega _2$; $\Omega _1$ (two dots-one dash line) corresponds to walks that do not move beyond vertex $\upsilon -1$. $\Omega _2$ (dashed line) corresponds to walks that reach at least vertex $\upsilon$. Double arrows encode the fact that random walk can return to the point of origin and single arrow means that the walker cannot revisit its starting point.}
\label{omega}
\end{center}
\end{figure}
\beq 
Q_M(x) = Q_M^{(\Omega_1)}(x)+ Q_M^{(\Omega_2)}(x).
\eeq
$Q_M^{(\Omega_2)}(x)$ can be written
\beq 
Q_M^{(\Omega_2)}(x)=Q(x;v)\frac{L_{v-1}}{L_v+L_{v-1}} G(x;v,r)
\eeq
where $G(x;v,r) $ generates walks which leave $v$ and return to $r$ without visiting $v$ again. $G(x;v,r)$ can be bounded by decomposing the walks following similar arguments that led to \eqref{G_decomp_P}, i.e. leave $v$, go to $v-1$, do any  number of returns to $v-1$, leave $v-1$ for the last time, go to $v-2$ etc which gives
\beq 
G(x;v,r)=(1-x)^{-\half}P(x;v)\frac{L_{v-2}}{L_{v-1}}G(x;v-1,r)
\eeq
where $P(x;v)$ is the first return generating function for walks that leave $v$ towards the root. Iterating gives (where the root is labelled 0)
\bea 
G(x;v,r)	&=&\prod_{k=v}^2(1-x)^{-\half}P(x;k)\frac{L_{k-2}}{L_{k-1}} G(x;1,r)\\
			&=& \frac{L_{0}}{L_{v-1}}(1-x)^{-\half(v-2)}\prod_{k=v}^2P(x;k),
\eea
where $G(x;1,r) = \sqrt{1-x}$. One can then use the monotonicity Lemma to bound the $P(x;k)$ by reducing the multigraph to  $M_k^*=\{1,1,1,...,L_{k-1}(M),L_{k}(M),...\}$  which yields  
\beq
P(x;k) \leq P^*(x;k)=\frac{(1-x)\frac{ L_{k-1}}{1+L_{k-1}} }{1-\frac{1}{1+L_{k-1}} P_{k-1}(x)  } \leq (1-x)
\eeq 
where in the last inequality we used that $P_{k-1}(x)$, the first return generating function for walks on the line segment  of length $k-1$, is  bounded above by 1. Thus we have
\bea 
Q_M^{(\Omega_2)}(x) &<& \frac{2}{x R}\frac{L_{0}}{L_v+L_{v-1}}          (1-x)^{v/2} \\
					   &<& \frac{1}{x R}.   
\eea
$Q_M^{(\Omega_1)}(x)$ is bounded using \eqref{start} by
\beq Q_M^{(\Omega_1)}(x)=\eta_ 0^{(\Omega_1)}<\eta_ {v-2}^{(\Omega_1)}+\sum_{n=0}^{v-3}\frac{1}{L_n}.\eeq
Now
\bea\eta_{ v-2}^ {(\Omega_1)}	&=&\frac{1}{L_{v-2}}\;\frac{1}{1-\frac{L_{v-2}(1-x)}{L_{v-1}+L_{v-2}}} \\
			&=&\frac{1}{L_{v-2}} \;\frac{L_{v-1}+L_{v-2}}{L_{v-1}+xL_{v-2}}\\
			&<&\frac{1}{L_{v-2}}+\frac{1}{L_{v-1}}
\eea
so
\beq 
Q_M^{(\Omega_1)}(x)	<\sum_{n=0}^{v-1}\frac{1}{L_n} < \sum_{n=0}^{R}\frac{1}{L_n} 
\eeq
and altogether
\beq \label{upperbound}
Q_M(x)<\frac{1}{x R} +\sum_{n=0}^{R}\frac{1}{L_n}.
\eeq
Taking expectation values
\bea 
\expect{Q(x;b)}{\infty U}	&<&\frac{1}{x R}   %
+\sum_{n=0}^{R}\frac{1}{1+{f^U}''(1)n/2}\\
			&<& \frac{1}{x R}    %
			+1+\frac{2}{{f^U}''(1)}\log( 1+{f^U}''(1)R/2 ).
\eea
Finally let $R=bx^{-1}$   and set $b=a^\half\lambda^{-\half}$, $x=a\xi$ so  
\bea \tilde Q(\xi;\lambda)&=&\lim_{a\to 0}a^\half \avg{Q \left (x=a\xi; b = a^{1/2}\lambda^{-1/2} \right)}_{\infty U}\\
&<&{\lambda^\half}\left(1+ \log\left(   1+\frac{1}{\xi\lambda}  \right) \right).\eea
Together with \eqref{longlower} this establishes the second part of  \eqref{newthm1}.

\section{Properties of spectral dimension in the transient case}
\label{transient case}
 
Throughout this section we will assume that with measure $\mu=1$  there exist constants $c$ and $N_0$ such that
\beq |B_{2N}(M)|<c |B_N(M)|,\quad N>N_0.\label{Ball:condition}\eeq
This rules out exponential growth for example. 
Recall from definition \eqref{dh_def} that a graph $G$ has Hausdorff dimension $\dha$ if
\beq |B_N(G)|\sim N^{\dha}\eeq
and an ensemble $\R M=\{\C M,\mu\}$ has annealed Hausdorff dimension $\dha$ if
\beq \expect{|B_N(M)|}{\mu}\sim N^{\dha}.\eeq
We will confine our considerations to the case $2\le \dha \le 4$ as being the regime of most physical interest. 

\subsection{Resistance in transient multigraphs}
\label{Resistance in transient multigraphs}

Here we note a lemma which describes the relationship between the resistance to infinity on $M$ and transience.

\begin{lemma}
The electrical resistance, assuming each edge has resistance 1, from $n$ to infinity is given by 
\beq 
\label{multigraphs_resistance}
\eta_n(0)=\sum_{k=n}^\infty \frac{1}{L_k}.
\eeq
If $\eta_n(0)$ is finite the graph is transient.
\end{lemma}
\begin{proof} This is a well known property of graphs (recall section \ref{resistance}) but we give an explicit proof here specifically for multigraphs as we will need \eqref{multigraphs_resistance} later. Clearly the right hand side of \eqref{multigraphs_resistance} is the resistance by the usual laws for combining resistors in parallel and series.  
 Setting $R= x^{-1}\abs{\log x}$ in \eqref{upperbound}  gives
\beq Q_M(x)<\frac{1}{\abs{\log x}}+\sum_{n=0}^\infty\frac{1}{L_n}-\sum_{n=R}^\infty \frac{1}{L_n}\eeq
and
\beq Q_M(0)=\eta_0(0)\le \sum_{n=0}^\infty\frac{1}{L_n}\label{etazero:upper}\eeq
which shows that if the rhs of this expression is finite the graph is definitely transient. Assuming this is the case, noting from \eqref{start}  
 that 
\beq \eta_1(0)=\eta_0(0)-1/L_0,\eeq
and proceeding by induction we see that 
\bea \eta_k(0)&=&\eta_0(0)-\sum_{n=0}^{k-1}\frac{1}{L_n}.\label{eta:reln}
\eea
Setting $k=\infty$ in  \eqref{eta:reln} gives 
 \beq  \eta_{0}(0)\ge \sum_{n=0}^\infty\frac{1}{L_n}\eeq
 which together with \eqref{eta:reln} and \eqref{etazero:upper} gives \eqref{multigraphs_resistance}. Finally it follows from the definition of $\eta_M(x)$ in \eqref{eta:upper} that if $\eta_n(0)$ is finite, so is $Q_n(0)$ and therefore $M$ is transient. 
 \end{proof}
 Note that by  Jensen's inequality
  \beq \sum_{n=0}^\infty\frac{1}{L_n}=\lim_{N\to\infty}\sum_{n=0}^N\frac{1}{L_n}> \lim_{N\to\infty}\frac{N^2}{\sum_{n=0}^N L_n}\eeq
  so that only  if $\dha\ge 2$   can  the graph be transient.

The solvable case $L_n=(n+2)(n+1)$ is discussed in Appendix \ref{Solvable} as a simple illustration of all these properties.

\subsection{Universal bounds on $\eta_n'(x)$}

Differentiating the recursion \eqref{start} and iterating we obtain
\bea
\label{second}
\! \! \! \! \! \! \!  \abs{\eta_0'(x)}&=&\abs{\eta_{N}'(x)} \prod_{k=0}^{N-1}\frac{(1-x)}{(1+xL_k\eta_{k+1}(x))^2}+\nn\\
\! \! \! \! \! \! \!& &+\sum_{n=0}^{N-1} (L_n\eta_{n+1}(x)^2+\eta_{n+1}(x) )(1-x)^{-1}\prod_{k=0}^{n}\frac{(1-x)}{(1+xL_k\eta_{k+1}(x))^2} \\
\label{first}
\! \! \! \! \! \! \!\qquad\quad &=&\abs{\eta_{N}'(x)} \prod_{k=0}^{N-1}\frac{(1-xL_k\eta_{k}(x))^2}{1-x}+
\nn\\
\! \! \! \! \! \! \!& &+\sum_{n=0}^{N-1} (L_n\eta_{n+1}(x)^2+\eta_{n+1}(x) )(1-x)^{-1}\prod_{k=0}^{n}\frac{(1-xL_k\eta_{k}(x))^2}{1-x}
\eea
so
\bea \label{eta_prime_upper}
\abs{\eta_0'(x)}
&<&\abs{\eta_{N}'(x)} (1-x)^{-N}e^{-2x\sum_{k=0}^{N-1} L_k\eta_{k}(x)} +
\nn\\
& &+\sum_{n=0}^{N-1}( L_n\eta_{n+1 }(x)^2+\eta_{n+1}(x)) (1-x)^{-n-2}e^{-2x\sum_{k=0}^{n-1} L_k\eta_{k}(x)} 
 \label{etaprime:upper} 
 \eea
and
\bea \label{eta_prime_lower}
\abs{\eta_0'(x)}&>&\abs{\eta_{N}'(x)} (1-x)^N e^{-2x\sum_{k=0}^{N-1} L_k\eta_{k+1}(x)}+
\nn\\
& &+\sum_{n=0}^{N-1} (L_n\eta_{n+1}(x)^2+\eta_{n+1}(x) )(1-x)^n e^{-2x\sum_{k=0}^{n} L_k\eta_{k+1}(x)}.
 \label{etaprime:lower} 
 \eea
We see that the upper and lower bounds are essentially of the same form.  Defining
\bea 
F_N(x)&=&\sum_{k=0}^{N} L_k\eta_{k+1}(x), \label{FN:defn}\\
G_N(x)&=&\sum_{k=0}^{N} L_k\eta_{k+1}(x)^2+\eta_{k+1}(x)\label{GN:defn}
\eea
we have
\begin{lemma}\label{lemma:etaprime}
For any  $M\in \C M$
\bea 
\abs{\eta_0'(x)}&>& c\,  G_{N^-(x)-1}(x)\label{new:bound}
\eea
for $x< x_0<1$, where $c$ is a constant and $N^{-}(x)$ is the integer such that
\bea x F_{N^-(x)}(x)& >& 1\geq x F_{N^-(x)-1}(x).\label{FNminus:defn}\eea
\end{lemma}

\begin{proof} 
Setting $N=N^-(x)$ in   \eqref{EtaLower}, and using \eqref{FNminus:defn}
we have 
\bea 
\eta_n (x)	&>& e^{-1} \sum_{k=n}^{N^-(x)-1}\frac{1}{L_k}\label{eta:lower2}
\eea
so that (using Jensen's inequality and the condition \eqref{Ball:condition}, see \eqref{A:result1})
\bea 
\frac{1}{x}>F_{N^-(x)-1}(x)	&>& e^{-1} \sum_{n=0}^{N^-(x)-1} L_n\sum_{k=n+1}^{N^-(x)-1}\frac{1}{L_k} > b_1^{-2} (N^-(x)-1)^2,\label{Nminus:inequalities}
\eea		
where $b_1$ is a constant $O(1)$, so that 
\beq 
\lceil b_1\, x^{-\half}\rceil >N^-(x).\label{Nminus:upper}
\eeq
Lemma \ref{lemma:etaprime} then follows by setting $N=N^-(x)$ in \eqref{etaprime:lower} and using  \eqref{Nminus:upper}. For future use we note that because $\eta_k(x)<\eta_k(0)<\eta_0(0)$ we have
\beq 
\lfloor c\, x^{-1/\dha}\rfloor <N^-(x).\label{Nminus:lower}
\eeq
\end{proof}

\subsection{Relationship between Hausdorff and spectral dimensions}

Our first result is 

\begin{theorem}  \label{ds leq dh} For any  graph $M\in \C M$ such that the Hausdorff dimension $\dha$ exists and is less than 4 then, if the spectral dimension exists, $d_s\le \dha$.
  \end{theorem}
\begin{proof}
The proof is by contradiction. First define the set of numbers
\beq 
\C X=\{x_k: 1=x_kF_{k-1}(x_k), k=1,2,3\ldots\}.
\eeq
Applying Cauchy-Schwarz inequality to \eqref{GN:defn} we have
\beq 
G_N(x)> \frac{F_N(x)^2}{|B_N|}
\eeq
so  applying  Lemma \ref{lemma:etaprime} gives
\bea
\abs{\eta_0'(x_k)}	&>&\frac {c}{x_k^2 |B_k|}.\label{etaprime:Xk}
\eea
We assume that $\dha$ exists for $M$ so that
\bea\label{eqthm8} \abs{\eta_0'(x_k)}&>&\frac {c}{x_k^2 k^{\dha}}\psi(k),\eea
where $ \psi(k)$ denotes a generic logarithmically varying function of $k$. %
Using  \eqref{Nminus:upper} and \eqref{Nminus:lower} (the latter if necessary to bound the logarithmic part) the right hand side is bounded below  by $c x_k^{-2+\dha/2}\psi(x_k)$. Now assume that $d_s$ exists in which case $\abs{\eta_0'(x_k)}< c' x^{-2+d_s/2}_k$;   however if $d_s>\dha$ there exist an infinite number of values $x\in\C X$ arbitrarily close to zero which contradict this. Therefore $d_s\le \dha$.
 
Note that for $\dha=4$ we get $\abs{\eta_0'(x_k)}>\psi(x_k)$. If $\psi(x)$ is logarithmically diverging as $x\to 0$ then again we can conclude that $d_s\le \dha$; the case when $ \eta_0'(0)$ is finite and $\dha=4$ is more subtle and we will not pursue it here.
\end{proof}

We can obtain more specific information about the spectral dimension by being more specific about the properties of the ensemble $\R M=\{\C M,\mu\}$. Define
\bea 
B_N^{(1)}&=&\sum_{k=0}^NL_k\sum_{n=k+1}^\infty \frac{1}{L_n},\\
\overline B_N^{(2)}&=&\sum_{k=0}^NL_k\sum_{m=k}^N \frac{1}{L_m}\sum_{n=k+1}^N \frac{1}{L_n}.	
\eea
Then we have the following
\begin{lemma}\label{lemma:dS:upper}
Given a graph $M\in\C M$ such that $\dha$ exists and
\bea B_N^{(1)}\sim N^{2+\gamma},\qquad 
\overline B_N^{(2)}\sim N^{4-\dha+\delta},	\label{lemma:dS:upper:conditions}
\eea
then the spectral dimension if it exists must satisfy
\bea d_s
&\le&\dha-\frac{2\delta-(4-\dha)\gamma}{2+\gamma}.\label{FullSpectralDim}
\eea
\end{lemma}
\begin{proof}
The proof uses Lemma \ref{lemma:etaprime} to show that
\beq  \abs{\eta_0'(x)} > \underline{c}\, x^{-2+\alpha/2}\abs{\log x}^{\underline c'}\eeq
where   
$\alpha$ is given by the right hand side of \eqref{FullSpectralDim}. Firstly by combining \eqref{eta:lower2} and the definition of $G_N(x)$ we have
\beq G_{N^-(x)-1}(x)\ge e^{-2} \overline B_{N^-(x)-1}^{(2)},\eeq
while from the definition of $N^-(x)$ and the fact that $\eta_n(x)$ is a decreasing function of $x$ we get   \eqref{FNminus:defn} 
\bea \frac{1}{x}<F_{N^-(x)}(x)	&<& B_{N^-(x)}^{(1)}.\eea
Lemma \ref{lemma:dS:upper} follows by combining these two results with the conditions \eqref{lemma:dS:upper:conditions}
and Lemma \ref{lemma:etaprime}. 
Again, the special case $\dha=d_s=4$ is  more subtle because $\eta_0'(0)$ might be finite and we will not pursue it here.

\end{proof}

 Now define
\bea 
 B_N^{(2)}=\sum_{k=0}^NL_k\sum_{m=k}^\infty \frac{1}{L_m}\sum_{n=k+1}^\infty \frac{1}{L_n}\,.	\eea
Then
\begin{lemma}\label{lemma:dS:lower}
Given a graph $M\in\C M$ such that $\dha$ exists and there exists a $N_0>0$ such that for $N>N_0$
\bea \eta_N(0)\sim N^{2-\dha+\rho},\qquad
 B_N^{(2)}\sim N^{4-\dha+\delta'},	\label{lemma:dS:lower:conditions}
\eea
or an ensemble $\R M=\{\C M,\mu\}$ such that 
\bea \expect{\eta_N(0)}{\mu}\sim N^{2-\dha+\rho},\qquad
 \expect{B_N^{(2)}}{\mu}\sim N^{4-\dha+\delta'},	
\eea
then the spectral dimension is bounded by
\bea d_s
&\geq&\dha- \frac{(4-\dha)\rho-(2-\dha)\delta'}{2+\delta'-\rho}\label{FullSpectralDim2}
\eea
provided that $\rho\leq \delta'+1$.
\end{lemma}
\begin{proof} Note that since $\eta_N(x)$ is a finite convex decreasing  function in $x=[0,1)$ 
\beq \abs{\eta_{N}'(x) } < \frac{\eta_{N}(0)}{x},\eeq
and $G_N(x)<G_N(0)=B_N^{(2)}$,
 combining this with \eqref{etaprime:upper} gives
\bea 
\abs{\eta_0'(x)} &<& (1-x)^{-N}\left(\frac{\eta_{N}(0)}{x}+B_N^{(2)} \right). \label{etaprime:upper:1} 
\eea
Choosing $N=x^{-\frac{1}{2+\delta'-\rho}} $ gives
\beq \abs{\eta_0'(x)}<   {\overline c}\, x^{-\frac{4-\dha+\delta'}{2+\delta'-\rho}}\abs{\log x}^{\overline c'}\eeq
for $x<x_0$ and provided that $\rho\leq \delta'+1$. The result for $d_s$ follows. In the case of the ensemble average we simply take the expectation value in
\eqref{etaprime:upper:1} before proceeding as before.  
\end{proof}

There are a number of constraints on and relations between the quantities
$\rho$, $\dha$, $\gamma$, $\delta$ and $\delta'$ which are summarized by

\begin{lemma}\label{constraints}
For any graph $M\in\C M$ such that $\dha$ and $\rho$ exist
\bea \rho\ge0,\quad\delta'\ge0,\quad \delta'\ge 2\gamma, \eea
and for any graph $M\in\C M$ such that $\dha$ and $\rho$ exist
\bea \gamma=\rho,\quad \delta'=2\rho,\quad \delta=2\rho.\eea
\end{lemma}
The proofs are elementary manipulations and outlined in Appendix \ref{Simple}.

The main result of this section is
\begin{theorem}\label{Theorem:dS-dH} For  any graph $M\in \C M$  such that $\dha<4$ and $\rho$ exist the spectral dimension is given by
\beq \label{ds_transient_rho}
d_s=\frac{2\dha}{2+\rho}.
\eeq
\end{theorem}
\begin{proof} The theorem follows from the upper and lower bounds in Lemmas \ref{lemma:dS:upper}
 and \ref{lemma:dS:lower}
 and using the relations between $\delta$, $\delta'$, $\gamma$ and $\rho$ in Lemma \ref{constraints}.

\end{proof}

It is an immediate corollary of Theorem \ref{Theorem:dS-dH}  that $\rho=0$ is a necessary and sufficient  condition for 
$\ds=\dha$.

\section{Scale dependent spectral dimension in the transient case} 
\label{scale dependence in transient case}

In this section we extend our results from section \ref{Scale dependent spectral dimension in the recurrent case} regarding the scale dependent spectral dimension in the recurrent case to the transient case. In particular, we are interested in the situation, motivated by the numerical simulations of four-dimensional CDT, where there is a scale dependent spectral dimension varying from $d_s^0=2$ at short distances to $d_s^\infty=4$ at long distances  (see chapter \ref{motivation}).

\subsection{Lessons from the recurrent case and simulations}
Unlike in two dimensions where the measure of the multigraph ensemble is obtained analytically, the situation in higher dimensions is more complicated and only numerical results are available. However to proceed we need some information about the measure of the multigraph ensemble and it becomes necessary to introduce some assumptions. 
The insights we gained from the two-dimensional model and the numerical results from computer simulations can be our guide for these ansatz. 

Both the recurrent and transient cases, indicate that the spectral dimension is determined by only two relevant characteristics; i) the \textit{volume} growth, i.e. the growth in the number of time-like edges with distance from the root, and the ii) the \textit{resistance} growth, i.e. the behaviour of the graph resistance. This statement is justified from the spirit of proofs in \cite{Durhuus:2009sm, Giasemidis:2012rf} and becomes manifest in expression \eqref{ds_transient_rho}, which shows that the spectral dimension depends only on the Hausdorff dimension, defined by the volume growth, and the anomalous exponent $\rho$ of the resistance growth. So, our assumptions should be related to these two quantities. Next, we should determine their functional form.
 
Recall from section \ref{GW trees and simply generated trees} that in two-dimensional CDT, the ensemble average of the number of time-like edges at distance $N$ from the root and the volume of a ball of radius $N$ are given by 
\bea 
\label{connectivity_avg}
\left \langle L_n \right \rangle _{\mu} &=&  n \, {f^U}''(1) + 1, \qquad  n\geq 1, \\
\label{volume_avg}
\langle \left | B_n| \right \rangle _{\mu} &\equiv& \left \langle \sum_{k=0}^{n-1} L_k \right \rangle _{\mu} = \half n (n-1) {f^U}''(1) + n, \qquad n\geq 1,
\eea
In addition,  computer simulations in four-dimensional CDT \cite{Ambjorn:2004qm, Ambjorn:2005db, Ambjorn:2007jv} show that for a triangulation $\tr$ with maximum distance $t$ from the root the number of four-simplices is 
\footnote{We use ``$\simeq$" to denote equality up to a multiplicative constant.}
\beq
\avg{|B_t|}_{Z} = \avg{N_4(t)}_Z  \simeq t^4.
\eeq
The set of causal triangulations is characterised by its bulk variables $N_i(t)\equiv N_i(T(t))$, $i=0,1,2,3,4$ which denote the number of $i$-simplices of this section of the triangulation. We can further distinguish these variables, for example, there are two different four-simplices and we have $N_4(t)=N_4^{(4,1)}(t)+N_4^{(3,2)}(t)$. Further, there are three different types of three-simplices and two different types of triangles and links (i.e.\ space-like and time-like). These ten different bulk variables are related by seven topological constraints, leading to only three independent variables \cite{Ambjorn:2001cv}. From these topological relations one has for example for the number of time-like links $N_1^{\mathrm{TL}}(t)$ (up to boundary terms) that
\beq
N_1^{\mathrm{TL}} =2 N_0(t) +\frac{1}{2}N_4^{(3,2)}(t)  - 3 \chi(T(t)).
\eeq
and using $N_0(t) \leq N_4(t)/5$ and $N_4^{(3,2)}(t) \leq N_4(t)$ we get
\beq \label{timelike links_simulations}
\avg{N_1^{\mathrm{TL}}(t) }_Z \leq \text{const.} \avg{N_4(t)}_Z\simeq t^4.
\eeq
Therefore, we  make the following ansatz for the behaviour of $\avg{L_N}_\mu$ in analogy to the two-dimensional case
\beq \label{avg:LN}
\avg{L_N}_\mu\simeq\nu N^3 +N,
\eeq
which implies
\beq \label{avg:BN}
cN\avg{L_N}_\mu<\avg{|B_N|}_\mu<c' N\avg{L_N}_\mu,
\eeq
where $c<c'$ are positive constants, and therefore is consistent with Monte-Carlo results \eqref{timelike links_simulations}. \eqref{avg:LN} is a natural generalisation of \eqref{connectivity_avg}, where $\nu$ takes the role of ${f^U}''(1)$.  The $N^2$ sub-leading term is absent in \eqref{avg:LN}, as to survive the continuum limit it would have to couple to $\sqrt{\nu}$, which in turn would imply its appearance in the Euclidean Einstein-Regge action.  
If $\nu$ is small in \eqref{avg:LN} then loosely speaking at very large distances the Hausdorff dimension is 4 while at short distances the linear term dominates and the volume growth appears to be two-dimensional.  

The next assumptions concern the size of the fluctuations in $L_N$. First, it is proven in the two-dimensional causal triangulations (UICT) \cite{Durhuus:2009sm} that the upward fluctuations in $L_N$ are controlled by $L_N \leq \avg{L_N}_{\mu} \log (f''(1) N)$, for $N > N_0$, for almost all graphs. Similarly we adopt
\beq
L_N \leq \avg{L_N}_{\mu} \psi \left (\sqrt{\nu} N^{1-\epsilon/2} \right).
\eeq
The last assumption should control the downward fluctuations and should be also related to the resistance. For this reason we assume
\beq
R(N) \leq \frac{N}{\avg{L_N}_{\mu}} \psi_{+} \left (\sqrt{\nu} N^{1-\epsilon/2} \right ). 
\eeq 
Let us summarise the ansatz of the multigraph ensemble 

\begin{assume}\label{assumptions4d} 
The multigraph ensemble $\R{M}=\{\C{M},\mu\}$ satisfies
\newcounter{saveenum}
\begin{enumerate}
\item \beq \label{assumption_i}\avg{L_N}_\mu\simeq\nu N^{3-\epsilon} +N,\eeq
\setcounter{saveenum}{\value{enumi}}
\end{enumerate}
with $\epsilon>0$ being arbitrarily small and for $\mu$-almost all multigraphs there exists a $N_0>0$ such that for $N>N_0$
\begin{enumerate}\setcounter{enumi}{\value{saveenum}}
\item \beq \label{assumption_ii}\eta_N (0) \leq \frac{N}{\avg{L_N}_\mu}\psi_+(\sqrt{\nu} N^{1-\epsilon/2}), 
\eeq
\item \beq \label{assumption_iii}L_N \leq \avg{L_N}_{\mu} \psi (\sqrt{\nu} N^{1-\epsilon/2}),  
\eeq
\end{enumerate}
where $\psi (x)$ and $\psi_+(x)$ are functions which diverge and vary slowly at $x=0$ and $x=\infty$. 
\end{assume}

The introduction of the arbitrarily small constant  $\epsilon>0$ is for  technical reasons and for all practical purposes one can think of it as being zero. 
Notice that under these assumptions the multigraph ensemble $\R{M}=\{\mathcal{M},\mu\}$ is almost surely transient.

\subsection{Continuum limit in the transient case}
Having introduced the assumptions which determine the ensemble measure, we are ready to apply the continuum limit formalism in this ensemble of transient multigraphs. 

\begin{theorem}\label{VariableSpectralDimensionTransient}
A multigraph ensemble $\R{M}=\{\C{M},\mu\}$ which satisfies Assumption \ref{assumptions4d} has $d_s^0=2$ at short distances while at long distances $d_s^\infty=4-\epsilon$ for $\epsilon>0$ arbitrarily small.
\end{theorem}
\begin{proof}
To prove the theorem we need the following lemma which we will prove in the following two  subsections
\begin{lemma} \label{lm:6.1}
For a multigraph ensemble $\R{M}=\{\C{M},\mu\}$ satisfying Assumption \ref{assumptions4d} 
\beq \label{Qprime_bounds}
 c_-  \frac{1}{\nu b' x^{\epsilon/2} +x} < \avg{\abs{\eta_0'(x)}} _{\mu}< c_+ \frac{1}{\nu b' x^{\epsilon/2}+x}  \psi_+^2\left( \sqrt{\frac{\nu b'}{x^{1-\epsilon/2}}} \right),
\eeq
where $c_-, c_+,b'$ are positive constants.
\end{lemma}
Defining the scaling limit as 
\beq
\abs{ \tilde Q'(\xi;G)} =\lim_{a\to 0} \left (\frac{a}{G} \right )\avg{\left |Q'\left(x=a\xi ; \nu=\frac{a^{1-\epsilon/2}}{b'G}\right)\right |}_{\mu} 
 \eeq
 and using Lemma \ref{lm:6.1} gives
\bea
c_-\frac{1}{\xi^{\epsilon /2}+G\xi } < \abs{\tilde Q'(\xi;G)} < c_+ \frac{1}{\xi^{\epsilon/2} +G\xi }  \psi_+^2\left( \sqrt{\frac{1}{G\xi^{1-\epsilon/2}}} \right).
\eea
For  short walks or equivalently $ \xi\gg G^{-1}$ we see that $\abs{\tilde Q'(\xi;G)}\sim\xi^{-1}$ giving $d_s^0=2$, while  for long walks or $ \xi\ll G^{-1}$, the $\abs{\tilde Q'(\xi;G)}\sim\xi^{-\epsilon/2}$ leads to $d_s^\infty=4-\epsilon$ which completes the proof of the main theorem.
\end{proof}

\subsubsection{Lower bound}
We now prove the lower bound of Lemma \ref{lm:6.1}. We begin with Lemma \ref{lemma:etaprime}  and proceed  as in the proof of Theorem \ref{ds leq dh} by applying the Cauchy-Schwarz inequality to get
\bea \label{dimred lower bound1}
\abs{\eta '_0(x)} &>& c \frac{\left (\sum _{n=0} ^{N^{-}-1}L_n \eta_{n+1}(x) \right )^2}{\sum _{n=0}^{N^{-}-1}L_n} \\
                           &=& c \frac{\left (\sum _{n=0} ^{N^{-}}L_n \eta_{n+1}(x) - L_{N^-}\eta _{N^{-}+1}(x) \right )^2}{\sum _{n=0}^{N^{-}-1}L_n}.
\eea
We recall \eqref{FNminus:defn} and \eqref{Nminus:upper} to bound the sums in the numerator and denominator respectively. In addition we use the fact that $\eta_{N^{-}+1}(x) < \eta_{N^{-}}(x)<\eta_{N^{-}}(0)$ and assumptions \eqref{assumption_ii} and \eqref{assumption_iii} to get
\bea
 \abs{\eta '_0(x)} &>& c \frac{\left (\frac{1}{x} - N^{-}(x) \psi \left ( \sqrt{\nu} (N^{-})^{1-\epsilon/2}\right ) \psi_+\left ( \sqrt{\nu} (N^{-})^{1-\epsilon/2}\right ) \right )^2}{\sum _{n=0}^{N^{*}}L_n} \\
                           &>& c \frac{\left (1-x N^{*}(x) \psi \left ( \sqrt{\nu} (N^{-})^{1-\epsilon/2}\right ) \psi_+\left (\sqrt{\nu} (N^{-})^{1-\epsilon/2}\right)  \right )^2}{x^2\sum _{n=0}^{N^{*}}L_n} ,
\eea
where $N^* = \lceil b_1x^{-\half} \rceil$.  Since $\psi(x)$, $\psi_+(x)$ are slowly varying functions the second term in the numerator is sub-leading as $x\to0$.
Taking the expectation value and applying Jensen's inequality we find for $x<x_0<1$
\beq
 \avg{\abs{\eta '_0(x)}} _{\mu} > c \frac{1}{x^2\avg{\sum _{n=0}^{N^{*}}L_n}_{\mu}} 
\eeq
which, together with $\avg{|B_N|}_\mu<c' N\avg{L_N}_\mu$, implies the lower bound of Lemma \ref{lm:6.1} with $b'=b_1^{2-\epsilon}$.

\subsubsection{Upper bound}

To prove the upper bound we first note that from \eqref{etaprime:upper:1} 
\beq
\abs{\eta_0'(x)} < (1-x)^{-N} \left(   \frac{  \eta_{N}(0)  }{x}+B_N^{(2)} \right)
\eeq
for any $N$. Taking expectation values we now get
 \bea
\avg{\abs{\eta_0'(x)}}_\mu &<&(1-x)^{-N}  \left(   \frac{\avg{\eta_{N}(0)}_{\mu}}{x}+  \avg{B_N^{(2)}}_{\mu} \right)  \\
                                     &<& (1-x)^{-N}  \left (\frac{N}{x \avg{L_N}_{\mu}} + \frac{N^3}{\avg{L_N}_{\mu}} \right ) \psi ^2_{+}(\sqrt{\nu}N^{1-\epsilon/2})
\eea
where we used the fact that $\avg{B_N^{(2)}}_{\mu} < \textrm{const} + c_3 \frac{N^3}{\avg{L_n}_{\mu}} \psi ^2_{+}(\sqrt{\nu}N^{1-\epsilon/2})$; to prove this  proceed similarly to the upper bound of \eqref{B1:upper} using \eqref{avg:BN} together with \eqref{assumption_ii}. Choosing $N=\lceil b_1 x^{-\half}\rceil$ gives the upper bound of Lemma \ref{lm:6.1}.

\section{Conclusion and outlook}
\label{multigraphs_outlook}
In this chapter we discussed multigraph ensembles motivated by their close relationship to various causal quantum gravity models. In particular this approach is well suited to studying the spectral dimension and exploring its possible scale dependence in causal quantum gravity. We studied two ensembles, the recurrent and transient, which correspond to reduced models of two- and higher-dimensional CDTs respectively.  

We explained that the measure on the recurrent multigraph is induced by the generalised uniform measure on infinite causal triangulations or, equivalently, a critical Galton- Watson process conditioned on non-extinction. This multigraph ensemble has Hausdorff dimension $\dha=2$. We show that by scaling the variance of the Galton-Watson process to zero at the same time as one scales the walk length to infinity (cf. \eqref{tildeQ:definition}) one obtains a continuum limit with a scale dependent spectral dimension which is $d_s^\infty=2$ at large scales and $d_s^0=1$ at small scales. Here $1/\sqrt{\lambda}$ is related to the rescaled second moment ${f^U}''(1)$ of the branching process and $\lambda$ determines the scale separating the short and the long walk limit. Regarding the physical interpretation of this model two comments are in order:
\begin{enumerate}
\item In pure two-dimensional CDT (UICT) there is no dependence on  Newton's constant due to the Gauss-Bonnet theorem. Hence, there is no length scale such as the Planck length in the model. This is  reflected in the fact that the uniform measure on infinite causal triangulations corresponds to a critical Galton-Watson process with off-spring distribution $p_k=2^{-k-1}$ which has $f''(1)=2$ fixed (section \ref{GW trees}). 
On the other hand as discussed in Section \ref{Scale dependent spectral dimension in the recurrent case}  the model with arbitrary ${f^U}''(1)$ can be thought of as describing CDT with a weight in the action coupling to the absolute value of the curvature \cite{Durhuus:2009sm} and $\sqrt{\lambda}$ acquires a physical description as the renormalised two-dimensional analogue of the gravitational constant $G^{(2)}$. 
\item Another point of interest is the dynamics of the model. As we explained in the beginning, the  model of random combs with scale dependent spectral dimension is a purely kinematic model proposed to show the existence of the scaling limit in a simplified context. On the other hand, the multigraph ensemble introduced in Section \ref{Scale dependent spectral dimension in the recurrent case} is directly related to CDT. It was shown in \cite{Sisko:2012an} that the rescaled  length process $l(t)=2a L_{[t/a]}/f''(1)$ of the multigraph  is described by the usual CDT Hamiltonian in the continuum limit
\beq
\hat{H}=-2\frac{\partial}{\partial l}-l \frac{\partial^2}{\partial l^2}+2 \mu l,
\eeq
where $\mu$ is the cosmological constant.
\end{enumerate}

Given that our aim is to introduce a model for dynamical dimensional reduction in four-dimensional CDT, we are led to consider multigraphs which have $\ds \geq2$, i.e. multigraph ensembles with transient walks, restricting ourselves to the physically interesting regime with $2\leq \dha \leq 4$. Before applying the continuum limit formalism we study the long distance properties. The main results are that for any multigraph $M$ such that $\dha$ and $\ds$ exist one has
\beq
d_s\leq\dha.
\eeq
If in addition the resistance exponent $\rho$  exists, then
\beq
d_s=\frac{2\dha}{2+\rho}.
\eeq
This implies that  $\rho=0$, which is a purely geometrical condition on the distribution of edges,  is a necessary and sufficient condition for $\ds=\dha$. It is interesting to notice in this context how multigraphs with $\rho=0$ attain the upper bound in \eqref{ds_inequalities}.

To perform the continuum limit in transient ensembles we have to adopt some assumptions because of the absence of analytical results in four-dimensional CDT. However those assumptions are not ad-hoc but are guided by the two-dimensional recurrent model. We then propose a model of a multigraph ensemble with scale dependent spectral dimension in the transient regime. In particular, we assume that the measure $\mu$  satisfies
\beq 
\avg{L_N}_\mu\simeq\nu N^3 +N,
\eeq
in addition to two more technical properties stated in assumptions \eqref{assumption_ii} and \eqref{assumption_iii}. It is then shown that this multigraph ensemble has a scale dependent spectral dimension with $d_s^0=2$ at short scales while at long scales $d_s^\infty=4$. One should notice that the assumption \eqref{assumption_i} implies the effective reduction of the Hausdorff dimension from $4$ to $2$ under scaling 
\footnote{Even-though a scale-dependent Hausdorff dimension has not been observed in computer simulations.},
which determines the reduction of the spectral dimension in the ensemble. This implies that the resistance exponent $\rho$ effectively remains zero 
\footnote{However one could in principle modify the assumptions which determine the ensemble and assign a scale dependence on $\rho$ keeping $\dha$ equal to the topological dimension at all scales (see sub-section \ref{Relation to the multigraph model} for further comments).}.

The result of Theorem \ref{VariableSpectralDimensionTransient} realises the goal we set at the beginning. It captures analytically the phenomenon of scale dependent spectral dimension in a graph model which inherits some characteristics of the full CDT.  However one can argue more about the physical implications of this result. For example, one might ask the following questions. What are the common degrees of freedom between multigraphs and CDTs and what are their physical interpretation? Can we extract the profile of the spectral dimension and compare it with computer simulations in both four-dimensional and three-dimensional CDT? Is the multigraph approximation adequate for analytically exploring the relation between CDT and other approaches to quantum gravity?

We attempt to answer these questions in the next chapter where we discuss and analyse the physical implications of the formalism developed and the models introduced so far.

\end{chapter}

\clearemptydoublepage
\begin{chapter}{Physical Implications}
\label{physics}

Having dealt with a considerable amount of technical work in the previous chapters, we are now ready to discuss the implications of our results and extract more physical information from our formalism. In order to better understand the physics which is underlined in our methods we have to address a crucial point; the validity of the multigraph approximation. In other words, how adequate is the radial approximation?
The reason that multigraphs serve as realistic models is twofold.  Firstly, as commented on in section \ref{scale dependence in transient case}, it is analytically proven in \cite{Durhuus:2009sm} that $L_N$  is bounded above by logarithmic fluctuations around the average for almost all graphs in the ensemble, i.e. $L_N \leq c N \log N$ for large $N$, where $c >1$. In other words, the number of space-like edges at finite height $N$, $\abs{S_N}$, remains finite since $L_N = \abs{S_N} +\abs{S_{N+1}}$. Thus, omitting the space-like edges in the reduced model does not affect the random walk at large times and therefore the value of the spectral dimension of the causal triangulation. Secondly, this intuition has been turned into a rigorous argument. As we have repeatedly mentioned,  the recurrent multigraphs bound above the spectral dimension of the UICT and numerical simulations indicate that this bound is tight, suggesting that both two-dimensional CDT and recurrent multigraphs share the same value of spectral dimension. These arguments suggest that the multigraph approximation does not affect the spectral dimension of CDT.

In the previous chapter we pursued this intuition further in the four-dimensional model and showed that it can account for the behaviour of the spectral dimension observed in numerical simulations. In this chapter we present further evidence and agreement with the Monte-Carlo results of four-dimensional CDT in the physical phase where vertices of arbitrary high degree are not observed. In essence we argue that a reduced model based on an ensemble of multigraphs obtained from radial reduction of the CDTs carries all the information needed about spectral dimension; it does not of course carry information about everything else as many degrees of freedom have been integrated out.

Next, we discuss the three-dimensional model and compare our results with numerical results presented in \cite{Benedetti:2009ge, Kommu:2011wd}. Lastly, we review the status of the dynamical dimensional reduction which has also been observed in other approaches to quantum gravity. We compare our findings and point out possible links between those approaches and CDT-like models.

\section{Further insights into the four-dimensional model}
\label{4dim}
The methodology to find the scale-dependent spectral dimension in the previous chapter can be summarised as follows. First, we sufficiently bound the ensemble average generating function $\avg{Q_M(x)}_{\mu}$, second, apply an appropriate scaling, take the continuum limit and finally extract the value of the spectral dimension at different scales. This process presupposes the application of a Tauberian theorem at the final step. However we could equivalently follow another slightly different method for determining the scale dependent spectral dimension which is described as follows. We may write expression \eqref{Qprime_bounds} in the compact form 
\beq \label{Qprime_compact}
\avg{|Q'_M(x)|}_{\mu (\nu)} \sim \frac{1}{\nu x^{\epsilon/2} + x}.
\eeq
We can now extract the average return probability as a function of large walk length. In particular, from the definition of the return generating function \eqref{Q_and_P_def_x} we get
\bea \label{1-x_times_Qprime}
(1-x)\avg{\abs{Q_M'(x)}}_{\meas} = \sum _{t=0}^{\infty} \frac{t}{2} \avg{p_M(t)}_{\meas}  (1-x)^{t/2}.
\eea
The left hand side of \eqref{1-x_times_Qprime} reads from \eqref{Qprime_compact}
\bea
(1-x)\avg{\abs{Q_M'(x)}}_{\meas} \sim L \left ( \frac{1}{1-\sqrt{1-x}}\right ) \, \, x^{-\epsilon/2} \qquad \text{as} \qquad x \to 0,
\eea 
where $L \left ( \frac{1}{1-\sqrt{1-x}}\right ) = \frac{1-x}{\nu + x^{1-\epsilon/2}}$ and $L(x)$ is a slowly varying function at infinity. Using a Tauberian theorem \cite[chapter XIII]{Feller:1965b} one gets (setting $\epsilon$ to zero in the expressions below to simplify the discussion)
\bea      \label{pav}
\avg{p_M(t)}_{\meas}  \sim 
\frac{2}{ t^2}  \left( \frac{\nu+1}{(1-1/t)^{2}} -1\right)^{-1} 
\eea
as $t \rar \infty $. Scaling $t(a) = \lfloor \sigma/a \rfloor$ and $\nu(a) = a/G$ as before one obtains the return probability density of the continuous diffusion time $\sigma$ through
\bea \label{pav-scale}
\tilde{P}(\sigma) \equiv \avg{P(\s)}_{\mu(\nu)} =\lim _{a \rar 0} \left ( \frac{a}{G} \right )^{-1} \avg{p_M(t)}_{\mu(\nu )} \sim   \frac{2 G^2}{\sigma^2}  \frac{1}{ 1 + 2G / \sigma}.
\eea
This is precisely expression \eqref{intro-P} which was conjectured in \cite{Ambjorn:2005db} as the behaviour of the continuum return probability density for diffusion on four-dimensional CDT.  It yields the scale dependent continuum spectral dimension (defined in \eqref{ds_def})
\bea \label{Ds_contin4}
D_s(\sigma) \equiv -2\frac{d \log \tilde P(\sigma)}{d \log \sigma} = 4\left(1- \frac{1}{2+ \sigma/G} \right).
\eea
The functional form of \eqref{Ds_contin4} is consistent with the numerical results in \eqref{DsJan}.

One can notice either from Lemma \ref{lm:6.1} or from expressions \eqref{pav-scale}, \eqref{Ds_contin4}  that the scale separating the regime of $d_s^{\infty}=2$ and $d_s^{\infty}=4$ is set by the rescaled $G=a/\nu$. Viewing the multigraph ensemble as a model of four-dimensional causal quantum gravity one can interpret $G$ as the normalised Newton's constant. While $G$ sets a scale on the duration of the walk, it is $\sqrt{G}$ that corresponds to the extent of the walk distance on the graph which can be identified with the Planck length $l_P$, being dimensionally consistent.

\section{Three-dimensional model}
\label{3dim}
The reduced multigraph approximation of the four-dimensional CDT seemed to be rather successful. In view of the numerical results of three-dimensional CDT \cite{Benedetti:2009ge, Kommu:2011wd} we would also like to apply our methods to the three-dimensional model.  Due to the absence of an analytical model for higher-dimensional CDT we must adopt similar assumptions which are justified with the same arguments as before and adjusted in the case of three dimensions.  
These assumptions can be summarised as follows 
\bea \label{connectivity_ansatz_3}
\left \langle L_N \right \rangle _{\mu (\nu_3)} &\simeq& \nu_3 N^{2} +N, \\
R(N) &\leq & \frac{N}{ \avg{L_N}_{\mu(\nu_3)}}  \psi _{+}(\sqrt{\nu_3N}), \\
\label{connectivity_upper_bound_3}
L_N &\leq& \avg{L_N}_{\mu(\nu_3)} \psi (\sqrt{\nu_3 N} ).
\eea
Here $\nu_3$ is the three-dimensional inverse bare Newton's constant, which, as we will shortly see, sets a length scale on the scale dependent spectral dimension.   

Following the arguments of section \ref{scale dependence in transient case} but under the assumptions \eqref{connectivity_ansatz_3}-\eqref{connectivity_upper_bound_3} we find
\beq \label{Qprime_compact3}
\avg{|Q'(x)|}_{\mu(\nu_3)} \sim \frac{1}{\sqrt{x} (\nu_3+\sqrt{x})}.
\eeq
Special attention has to be paid to defining the correct scaling limit of the discrete random graph model in which both the walk length as well as the characteristic length scale diverge in a ``double scaling limit'', as was analysed in previous chapters. In particular we apply the scaling $x= a\xi$ and $\nu_3 = \sqrt{a}/G_{3}$ with scaling exponent $\Delta_{\mu (\nu_3)} = 1$
\beq \label{Qprime_contin3}
\abs{ \tilde Q'(\xi; G_{3}) }\equiv \lim _{a \rar 0 } \left(a/G_{3}\right) \avg{\abs{Q_M' (x=a \xi)}}_{\mu\left (\nu_3=\sqrt{a}/G_{3}\right)}
\eeq
which leads to 
\beq \label{Qprime_cont_3dim}
 \abs{ \tilde Q'(\xi; G_{3}) } \sim
\begin{cases}
\xi ^{-1}, &\xi >> G^{-2}_{3},\\
\xi ^{-1 /2}, &\xi <<  G^{-2}_{3},
\end{cases} 
\eeq
where $G_{3}$ is the renormalised three-dimensional Newton's constant. This result implies $d_s^{\infty}=3$ at large distances while $d_s^{0} = 2$ at short scales, which agrees with the numerical results in \cite{Benedetti:2009ge, Kommu:2011wd}. Although $G_{3}^2$ sets a scale on the {\it duration} of the walk, it is its square root which corresponds to the {\it length extent} on the graph and is identified with the Planck length in three dimensions, which is consistent by dimensional analysis 
\footnote{Recall the relationship between Planck length and Newton's constant in $d$ topological dimensions
\beq
{l_{P(d)}}^{d-2} = \frac{l_P^2}{G} G_{d} \Rightarrow l_{P(3)}= \frac{\hbar}{c^3} G_{3}
\eeq}.

As we explained in the previous section, we can equivalently apply to \eqref{Qprime_compact3} a Tauberian theorem and find the average return probability density
\beq \label{return probability_3}
\avg{p(t)}_{\mu(\nu_3)} \sim \frac{2}{\sqrt{\pi}} \frac{1}{t^{3/2}} \frac{\left (1-1/t \right )^2}{\nu_3 + \left (1-(1-1/t)^2 \right )^{1/2}}
\eeq
at large $t$. Next, we apply the scaling  as in \eqref{Qprime_contin3} and determine the return probability density $\tilde P(\sigma, G_3)$ for continuous diffusion time $\sigma$.  
In this case one has
\beq \label{heat_trace_2+1}
\tilde P(\sigma, G_3)  \equiv  \avg{P(\s)}_{\mu(\nu_3)} = \lim_{a\to 0} a^{-1} \avg{p(t=\lfloor \sigma/a \rfloor)}_{\mu(\nu_3)} \simeq \frac{1}{\sigma^{3/2}\left (1+\sqrt{2} G_{3}/\sqrt{\sigma}\right )}
\eeq 
 which implies that the scale dependent spectral dimension is given by
 \beq \label{Ds_2+1}
 D_s(\sigma) \equiv -2\frac{d \log \tilde P(\sigma,G_3)}{d \log \sigma} = 3 - \frac{1}{1+ \sqrt{\sigma/(2 G_{3}^2)}}
 = \begin{cases}
 2, \qquad &\sigma \to 0, \\
 3, \qquad &\sigma \to \infty,
 \end{cases}
 \eeq
which confirms the dynamical dimensional reduction observed in numerical simulations of three-dimensional CDT\cite{Benedetti:2009ge, Kommu:2011wd}. 

However we should note that, while having the correct limits, this result is slightly different from the rational or exponential fits suggested in \cite{Benedetti:2009ge, Kommu:2011wd} to explain the numerical data. 
In particular, the authors in \cite{Benedetti:2009ge}  achieve a better fit using a function of the form
\beq \label{darios_fit}
D_s(\sigma)  = a + b e^{-c \sigma}
\eeq   
instead of using the rational profile $D_s(\sigma) = a + \frac{b}{c+\sigma}$ which was used to fit the numerical data  in four-dimensional CDT \cite{Ambjorn:2005db} and successfully confirmed by the reduced multigraph approximation. One  observes that none of these fits is in agreement with the interpolation function \eqref{Ds_2+1}.  Of course the rational fit is ``much closer" in a sense to \eqref{Ds_2+1} but the crucial difference is the appearance of  $\sqrt{\sigma}$ in the denominator of \eqref{Ds_2+1}
 \footnote{One could argue that if we scaled time differently, e.g. $t=\lfloor\sigma ^{\kappa}/a \rfloor$, we could adjust $\kappa$ such that we get the rational function \eqref{frac_fit}. In this case the spectral dimension yields
 \beq
D_s(\sigma) = \kappa\left( 3- \frac{1}{1+\frac{\sigma ^{\kappa/2}}{\sqrt{2}G_{3}} }\right ),
\eeq  
 which implies that for $\kappa = 2$ the spectral dimension takes the values $6$ and $4$ in the IR and UV limits respectively. These values are non-physical. Therefore we conclude that the only physical solution is for the unique choice $\kappa = 1$.}.

For this reason we fit the spectral dimension \eqref{Ds_2+1} to the numerical data 
\footnote{We are thankful to Dario Benedetti, Joe Hanson and Rajesh Kommu for providing their data, for their correspondence and their useful clarifications in data analysis.}.
We start with the numerical results presented in \cite{Benedetti:2009ge}. The data consists of eight ensembles, each having a fixed number of simplices. Here we focus on the ensembles with $N=70,000$ and $N=200,000$ simplices as in \cite{Benedetti:2009ge}. In addition, we are interested in the region where the spectral dimension is unaffected from discreteness and finite size/curvature effects and is also increasing with diffusion time, which means $20\leq \sigma \leq300$, similarly to \cite{Benedetti:2009ge}. The latter condition is imposed because our fitting function is a strictly increasing function of the diffusion time. 

At this point we should emphasise that all the following fits are monotonically increasing functions from 2 to 3. However the numerical results in \cite{Benedetti:2009ge} led the authors to argue that the fit might not have a large scale limit free of quantum effects. This means that the spectral dimension might reach the topological value 3 when quantum effects are still important, then increase above 3 for some diffusion time $\sigma$ and finally decrease to 3 again at diffusion time large enough for quantum corrections to be neglected and small enough for curvature effects of the classical geometry to be sub-leading (the ``bump" effect). One can argue that this kind of behaviour is an artefact of the data analysis. Others have assumed this effect is real and studied the similarities with the three-dimensional foliation-defining scalar Ho\v rava-Lifshitz gravity where the scalar, which diffuses, is a physical propagating degree of freedom  \cite{Sotiriou:2011mu}. In any case the spectral dimension \eqref{Ds_2+1} derived from the multigraph model is clearly a monotonic function and cannot imitate such behaviour whether it is physical or an artefact. 

Therefore, we silently adopt the assumption that the quantum effects are negligible at about $\s=300$ and apply the exponential, rational fits and the fit from the multigraph model (the ``multigraph'' fit from now on), each having three free parameters, i.e. 
\bea
\label{expo_fit}
D_s^{\text{expo}}(\s) &=& a+b e^{-c \s}, \\
\label{frac_fit}
D_s^{\text{frac}}(\s) &=& a+\frac{b}{c+\s},
\eea
\beq
\label{multi_fit}
D_s^{\text{multi}}(\s) = a+\frac{b}{c+\sqrt{\s}},
\eeq
leading to the following values \vspace{2mm} \newline 
\begin{tabular}{|l|c|c|c|c|c|}
\hline
$N=70k$& $a$                                & $b$                                     & $c $                                    & SSE      & $R^2$ \\ \hline
\scriptsize{exponential} & 2.981 ($\pm 0.001$)  & -0.857 ($\pm 0.002$)     & 0.0139 ($\pm0.0001$)   & 0.00176 & 0.9998 \\ \hline 
\scriptsize{fractional}   & 3.145 ($\pm 0.007$)  & -54.99 ($\pm 1.94$)       & 45.11 ($\pm 2.2$)        & 0.02914 & 0.9963 \\ \hline 
\scriptsize{multigraph} & 3.385 ($\pm 0.024$)  & -7.823 ($\pm 0.577$)     & 2.585 ($\pm 0.425$)      & 0.05978 & 0.9924 \\ \hline  \hline
$N=200k$& $a$                              & $b$                                     & $c $                                    & SSE      & $R^2$ \\ \hline
\scriptsize{exponential} & 3.052 ($\pm 0.001$)  & -1.015 ($\pm 0.003$)     & 0.0167 ($\pm0.0001$)   & 0.00149 & 0.9998 \\ \hline 
\scriptsize{fractional}   & 3.205 ($\pm 0.006$)  & -45.04 ($\pm 1.46 $)      & 28.64 ($\pm 1.61$)          & 0.0372 & 0.996 \\ \hline 
\scriptsize{multigraph}& 3.401 ($\pm 0.019$)  & -5.833 ($\pm 0.374$)     & 0.6078 ($\pm 0.289$)      & 0.07974 & 0.9915 \\ \hline
\end{tabular}

\vspace{2mm} 
In the last two columns we also provide the sum of squared errors of prediction (SSE) and the $R^2$ of each fit respectively. The better the fit to the numerical data the closer SSE tends to zero and $R^2$ to one. Here some comments are in order. First, we reproduce the values of the parameters for the exponential fit reported in \cite{Benedetti:2009ge}. Second, we observe that the best fit is indeed given by the exponential function as argued by Benedetti and Henson, while the second best is the rational fit. Even though the fit from the multigraph model seems to be the least appropriate, we elaborate on it and study its residuals and errors.

Figure \ref{70k_multi_ds} shows the ``multigraph'' fit to the  data points for the ensemble with size $N=70,000$ simplices and figure \ref{70k_multi_res} illustrates the  corresponding residuals, i.e. the difference between the numerical and analytical value. We notice that the ``multigraph'' fit is qualitatively close to the numerical data. To make quantitative statements we must compare the residuals to the data-point errors. These errors were not provided for every single data-point but for a selection of them and are of order of $\pm0.02$.  Comparing with the residuals we conclude that the fit from the multigraph approximation is good enough except for the small and large values of the diffusion time. The discrepancy for small values of $\s$ might be due to the discreteness effects. Although we cut off the rapid oscillating data points which are dominant for $\s < 20$, the discreteness effects might still play a (non-observable) role altering slightly the value of the spectral dimension for $\s <25$. The discrepancy for values at $\s \simeq 300$ is at the edge of the error bars and is due to the finite size effects. Looking closely at the data points for $270\leq \s \leq 300$, we see that the spectral dimension peaks at $\s = 289$ and then starts decreasing slowly because of the finite size effects. Clearly, for $\s \geq 270$ the spectral dimension starts getting contributions from the finite-size effects.
\begin{figure}
\begin{center}
\includegraphics[scale=0.6]{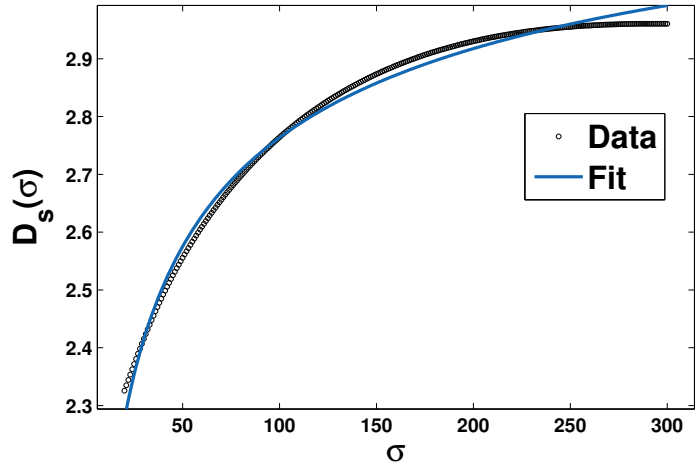}
\caption{The spectral dimension as a function of diffusion time - $N=70k$ simplices.}
\label{70k_multi_ds}
\end{center}
\end{figure}

\begin{figure}
\begin{subfigure}{0.4 \textwidth}
\includegraphics[scale=0.48]{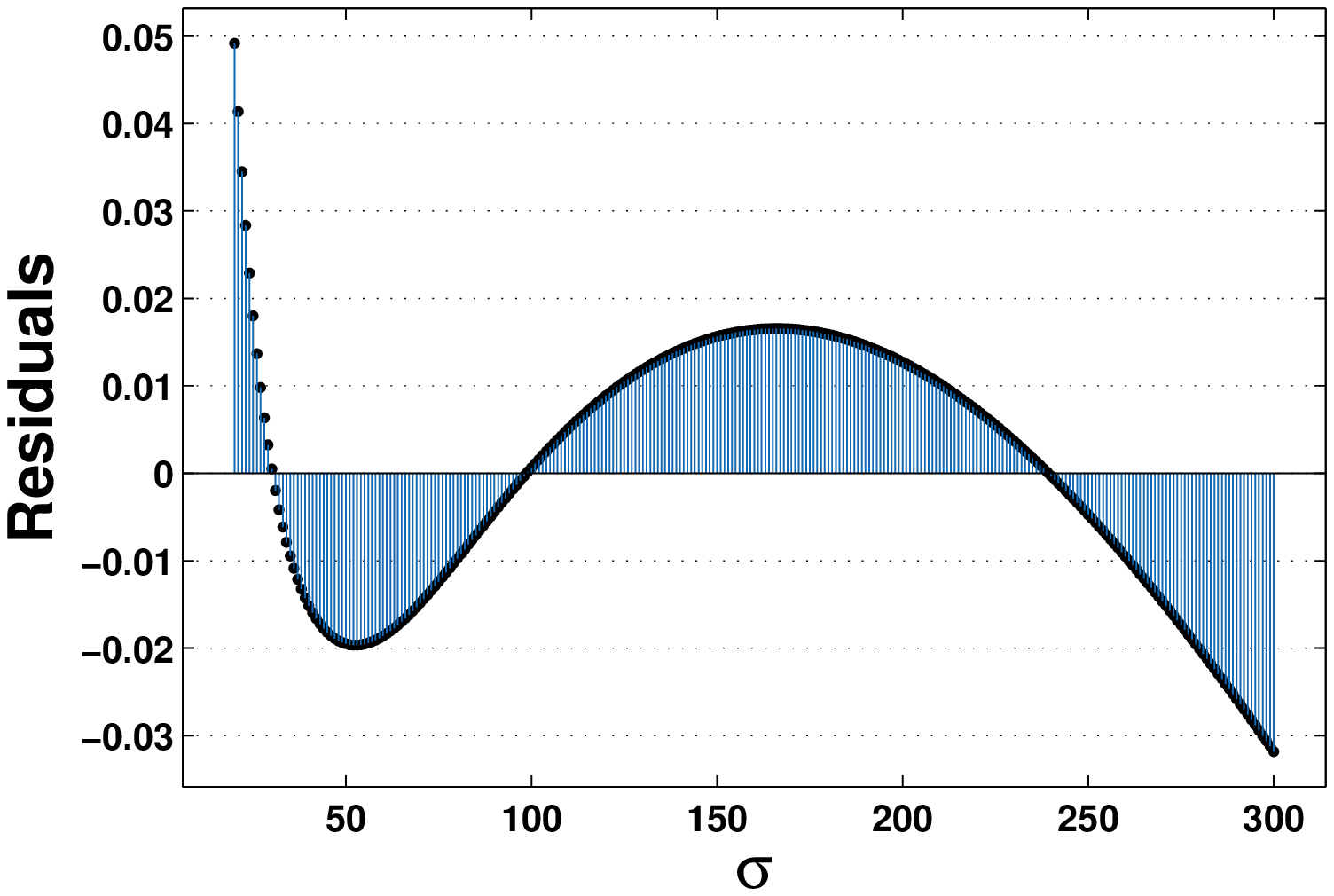}
\caption{$N=70k$ simplices}
\label{70k_multi_res}
\end{subfigure}
\qquad \qquad
\begin{subfigure}{0.4 \textwidth}
\includegraphics[scale=0.48]{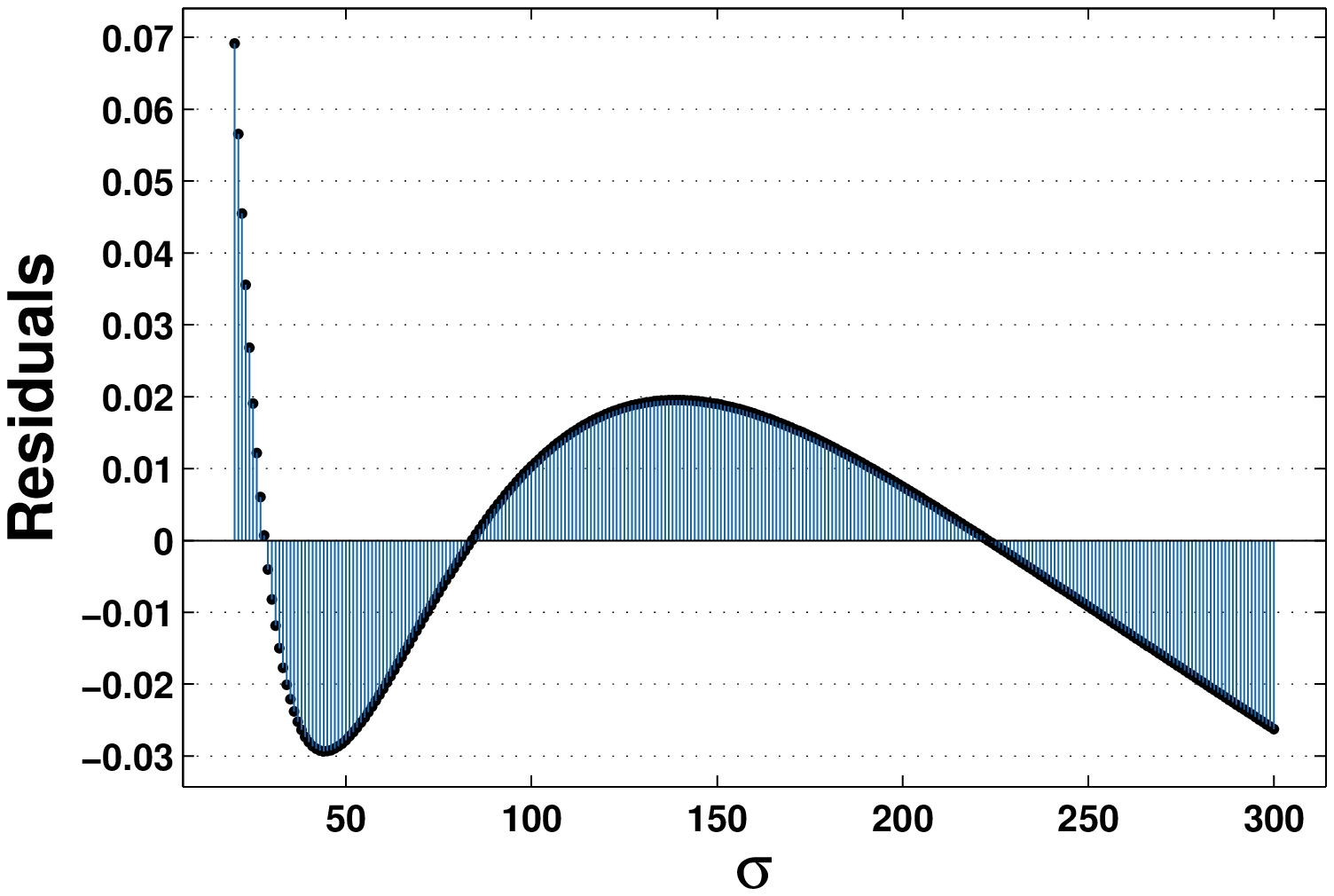}
\caption{$N=200k$ simplices}
\label{200k_multi_res}
\end{subfigure}
\caption{Residuals}
\end{figure}

Figure \ref{200k_multi_ds} also presents the spectral dimension from the ensemble with geometries of size $N=200,000$. Here the errors are estimated to be of order $\pm0.035$. From the residual plot \ref{200k_multi_res} we observe that for small values of $\s <25$ there is significant difference between the data and the ``multigraph'' fit. The discrepancy falls off quickly for $\s \geq 25$ which leads us to assume that it is due to discreteness effects as we argued in the last paragraph. For larger values of $\s$ the residuals are within the error-bars. This is consistent with the previous discussion because in this case the finite size effects have not kicked in yet. Actually, a close look at the data points reveals that the spectral dimension increases beyond $\s =300$ and peaks at $\s=403$, when the finite-size effects take over. 
\begin{figure}
\begin{center}
\includegraphics[scale=0.6]{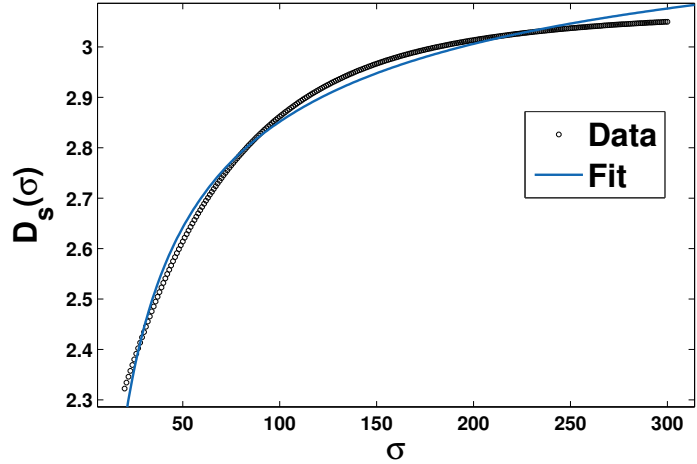}
\caption{The spectral dimension as a function of diffusion time - $N\!=\!200k$ simplices.}
\label{200k_multi_ds}
\end{center}
\end{figure}

Another independent simulation was performed in \cite{Kommu:2011wd}, where the author applies both fits to the numerical data
\bea\label{rational_Kommu}
D_s(\sigma) &=& 3.03 - \frac{10.51}{17.87+\sigma}, \\
\label{exp_Kommu}
D_s(\sigma) &=& 3.19 - 0.97 e^{-0.013 \sigma}.
\eea 
 
Our goal is to apply the ``multigraph'' fit to these data-points too. After our request, we received two ensembles of geometries, the first one had 501 members of size $N=82,000$ 3-simplices approximately and the second ensemble had 482 members of size $N=110,000$ 3-simplices approximately. The data was in form of the return probability of each geometry for $0\leq\sigma\leq600$. So we had to convert this information to the ensemble average spectral dimension. Our data analysis consists of the following steps. First we determined the ensemble average return probability, $\bar P_N(\s)$. Second, we extracted the spectral dimension by using two methods. 

In the first method we applied a discretised version of the definition \eqref{ds_def} and extracted the plots in figure \ref{Ds_my_analysis}. 
In the first ensemble, we observe that  the oscillating points vanish around $\s=170$, then the spectral dimension remains approximately constant in the region $170<\s\leq 300$ and decreases for $\s >300$ due to finite size effects. In the second ensemble, we notice that the oscillating points due to the discreteness  vanish only when the spectral dimension starts decreasing due to finite volume. This  is very similar to the plot \cite[4 (b)]{Anderson:2011bj} which corresponds to the spectral dimension in 2+1 dimensional CDT too. In these plots there is no regime where the spectral dimension is free from both lattice and compactness effects. As a result we are not able to apply the ``multigraph'' fit to these data-points. 
\begin{figure}[h]
\begin{subfigure}{0.5 \textwidth}
\includegraphics[scale=0.45]{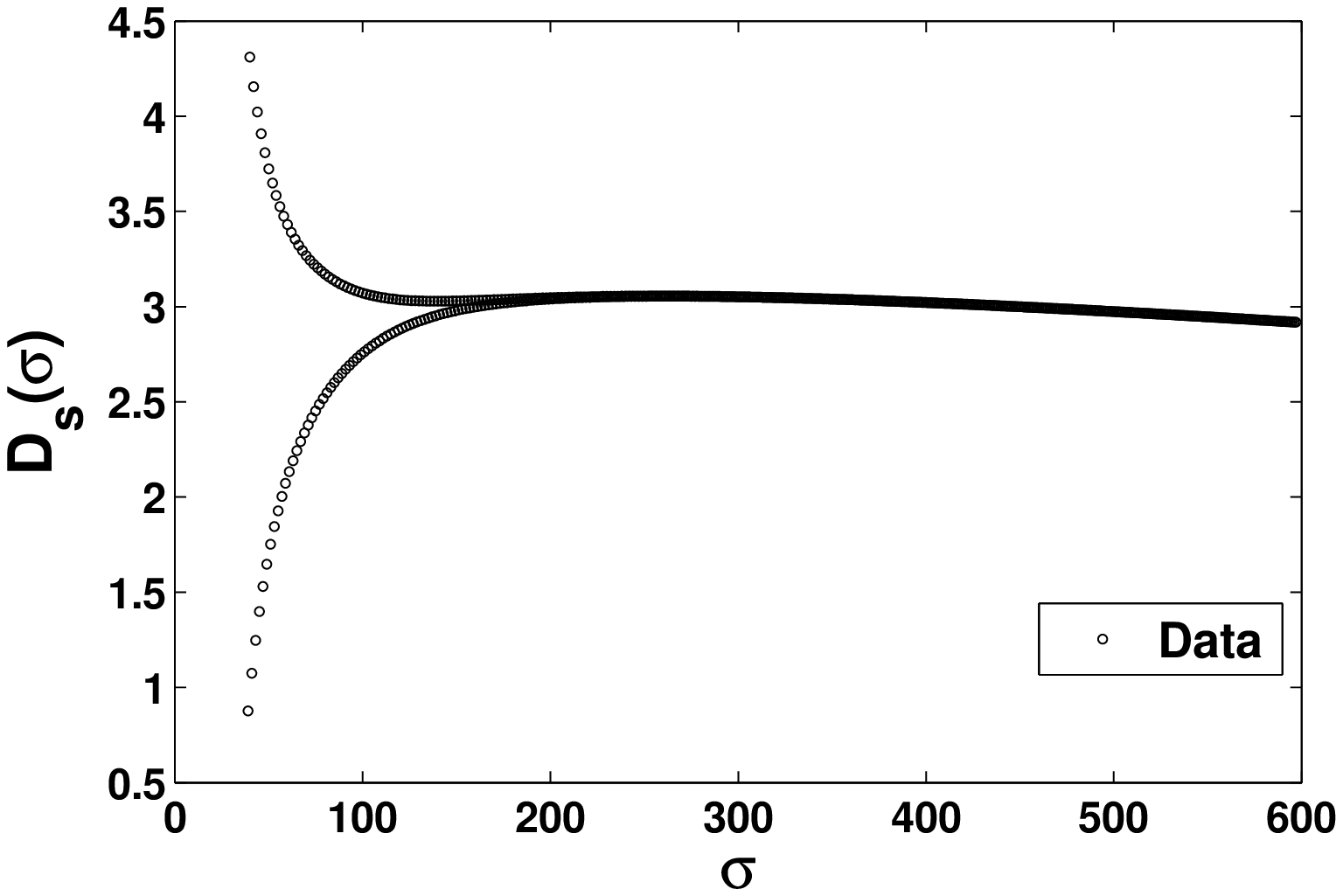}
\caption{Ensemble 1: 3-simplices of size 82k}
\end{subfigure}
\begin{subfigure}{0.5 \textwidth}
\includegraphics[scale=0.45]{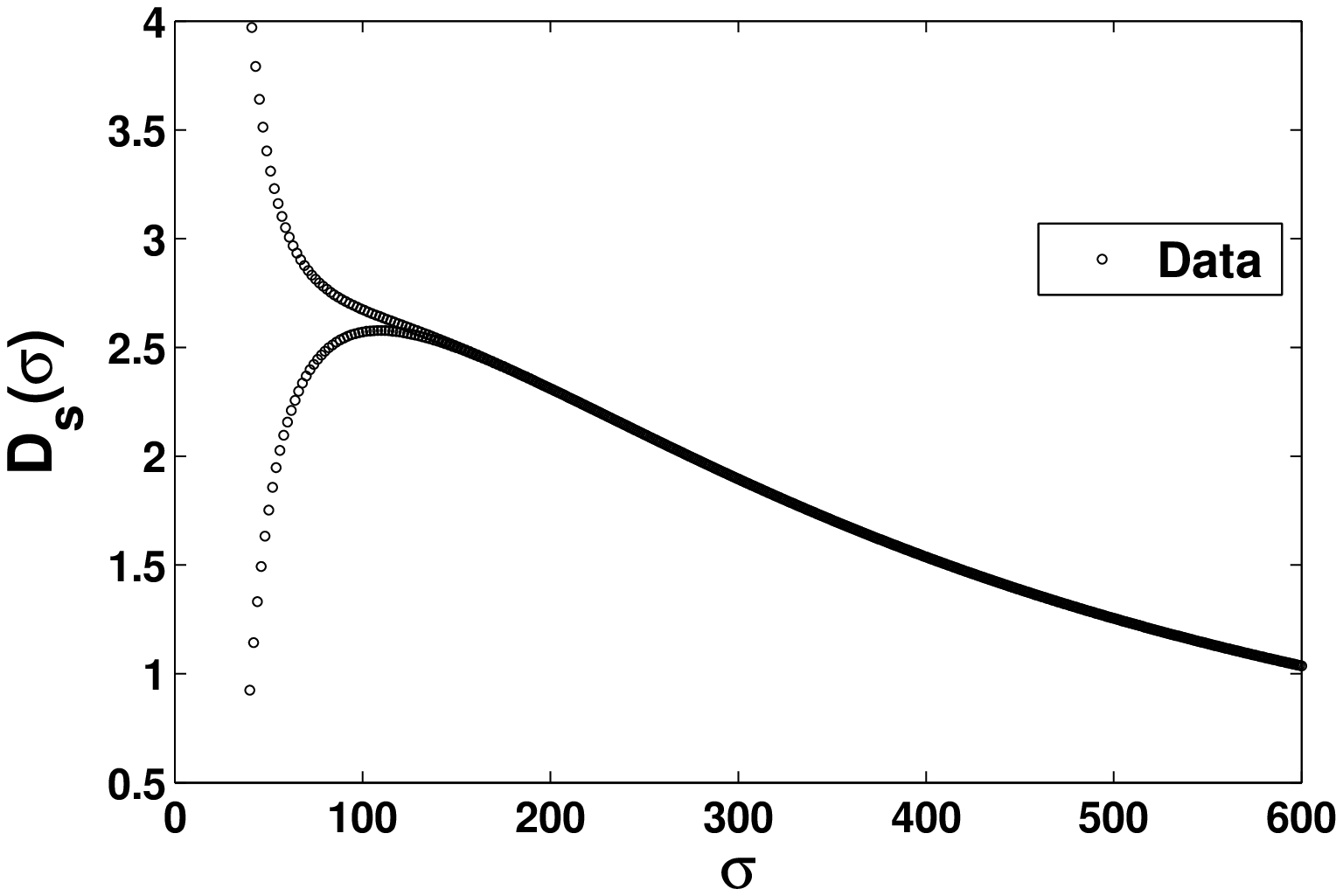}
\caption{Ensemble 2: 3-simplices of size 110k}
\end{subfigure}
\caption{Applying a discretised version of the definition \eqref{ds_def} to find the spectral dimension. $40\leq\s\leq600$.}
\label{Ds_my_analysis}
\end{figure}

In principle, we want to make sure that our results are not due to our (limited knowledge of) data analysis methods. For this reason we apply a second method for determining the spectral dimension used by the author in \cite{Kommu:2011wd}. That is, we plot $(-2 \log \bar P(\s), \log \s)$ and compute the spectral dimension through the slope of successive points leading to figure \ref{Ds_Kommus_analysis}. 
\begin{figure}[h]
\begin{subfigure}{0.4 \textwidth}
\includegraphics[scale=0.45]{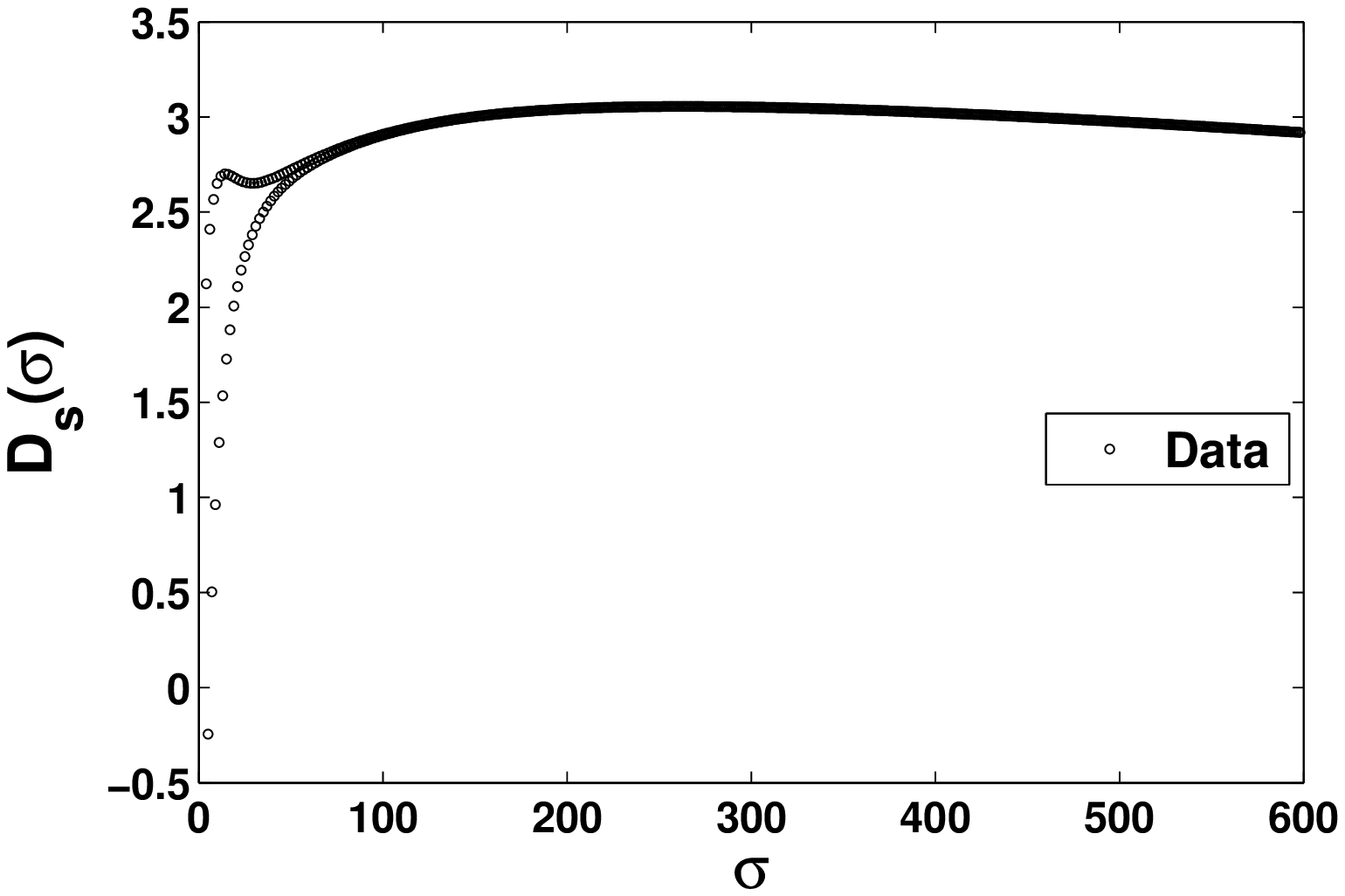}
\caption{For 5 successive points}
\end{subfigure}
\qquad \qquad
\begin{subfigure}{0.4 \textwidth}
\includegraphics[scale=0.45]{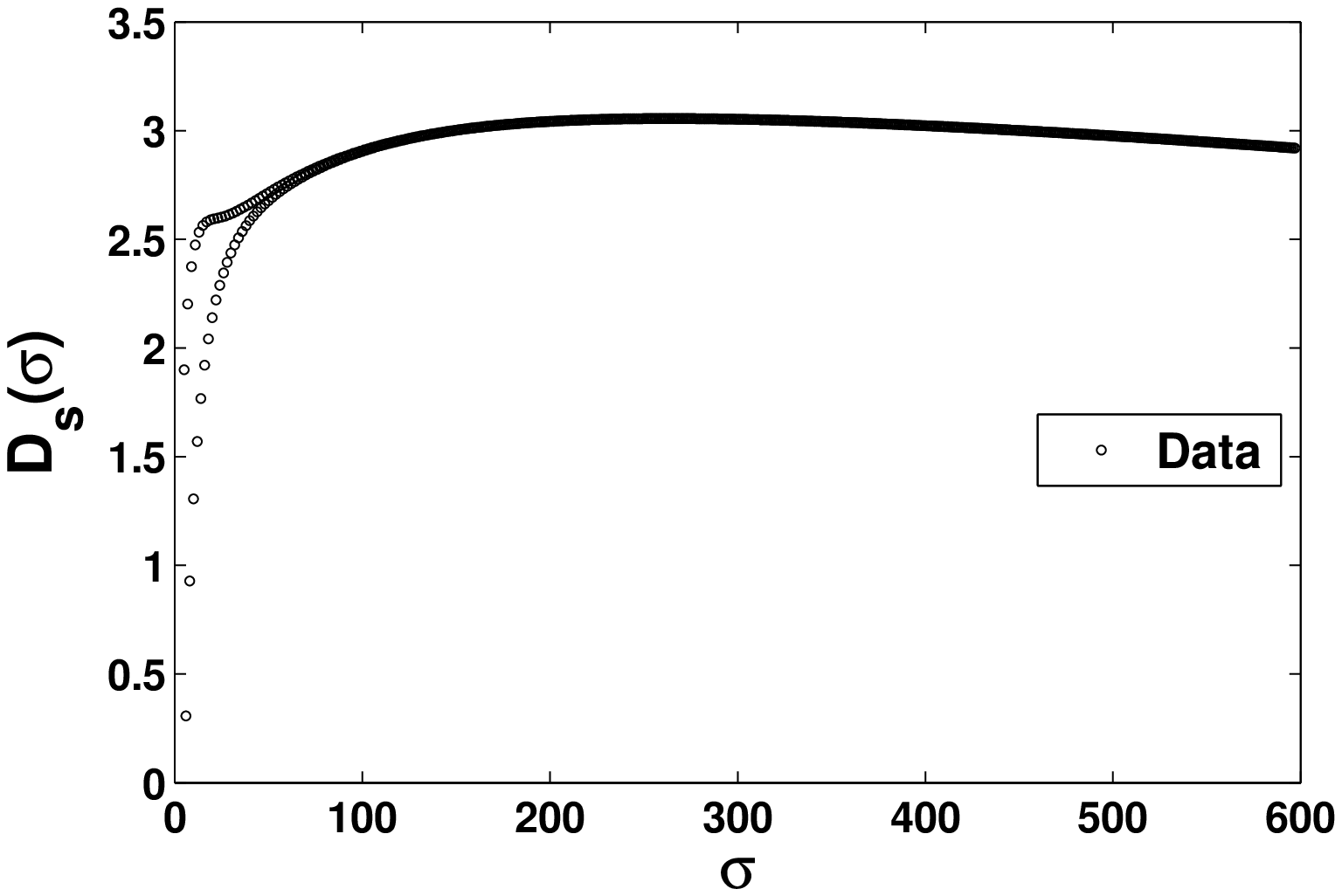}
\caption{For 7 successive points}
\label{82k_Kommus_analysis_step_7}
\end{subfigure}
\caption{Spectral dimension of ensemble 1 through the slope of $(-2 \log \bar P_N(\s), \log \s)$. $1\leq\s\leq600$. By increasing the step of successive points the oscillating pattern is smoothed out.}
\label{Ds_Kommus_analysis}
\end{figure}
Using 5 or 7 successive points for finding the slope, has the side effect of smoothing the data. As a consequence the large oscillations in plot \ref{Ds_my_analysis} have smoothed out in plot \ref{Ds_Kommus_analysis}. However they have not vanished completely. Now we observe that the envelope of the oscillating points oscillates for $\s < 40$ but starts increasing for $\sigma \geq 40$. For the sake of the analysis, we interpret the oscillation of the envelope as the region where discreteness effects are still present and the increase of the envelope at $40\leq\s\leq 250$ as the region which is free from both discreteness and finite-volume effects (figure \ref{82k_multi_fit1}). Despite this vague interpretation of the data-points we attempt to apply the three fits.
\begin{figure}[h]
\begin{center}
\includegraphics[scale=0.6]{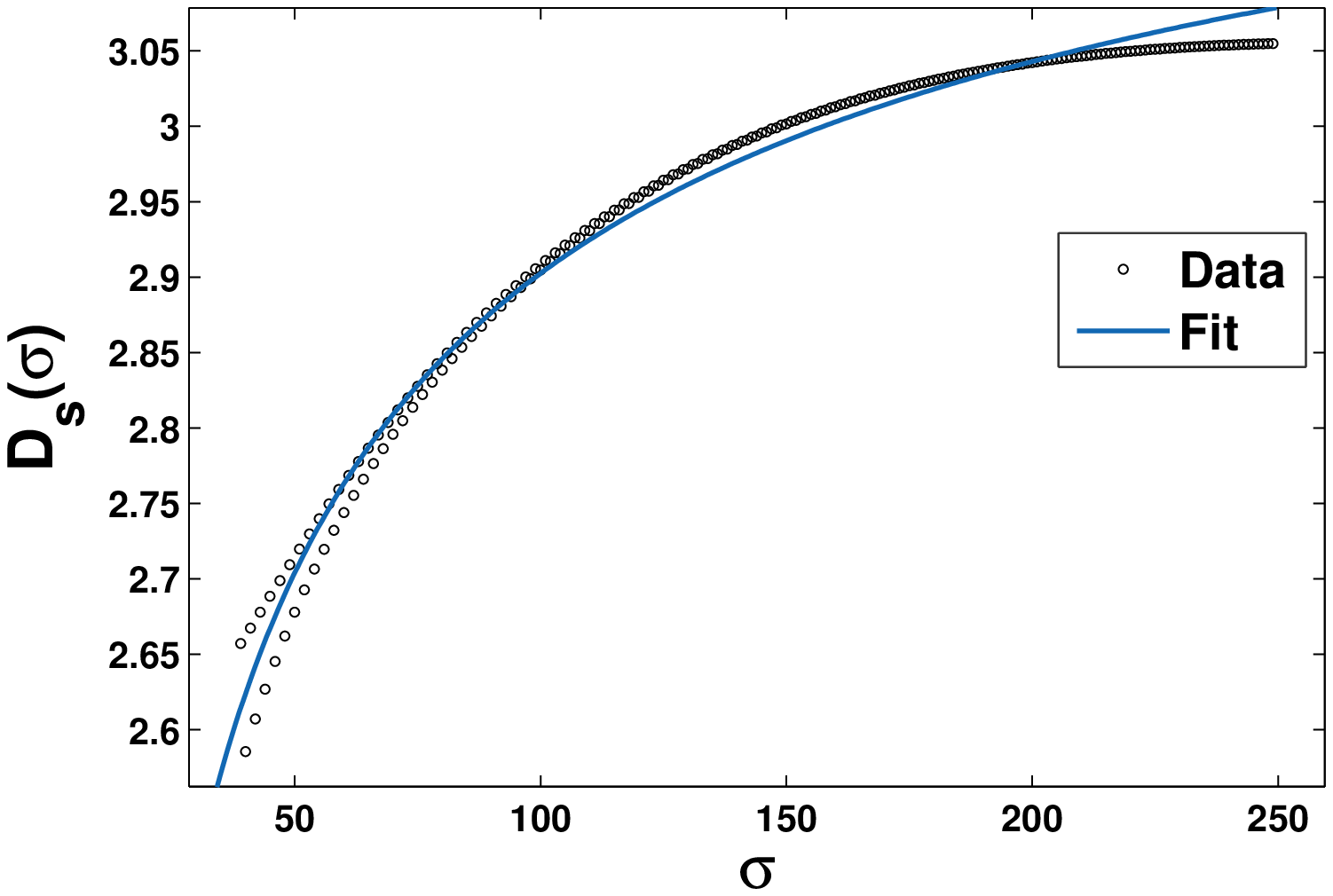}
\end{center}
\caption{Zooming in on the region $40\leq\s\leq250$ of the plot \ref{82k_Kommus_analysis_step_7}. We also apply the ``multigraph'' fit.}
\label{82k_multi_fit1}
\end{figure}
The results are the following. \vspace{2mm} \newline
\begin{tabular}{|l|c|c|c|c|c|}
\hline
$N=70k$& $a$                                & $b$                                     & $c $                                    & SSE      & $R^2$ \\ \hline
\scriptsize{exponential} & 3.073 ($\pm 0.004$)  & -0.861 ($\pm 0.017$)     & 0.01653 ($\pm0.0004$)   & 0.01313 & 0.9955 \\ \hline 
\scriptsize{fractional}   & 3.2 ($\pm 0.009$)  & -35.06 ($\pm 2.23$)       & 19.76 ($\pm 3.42$)        & 0.02216 & 0.9924 \\ \hline 
\scriptsize{multigraph} & 3.325 ($\pm 0.02$)  & -3.676 ($\pm 0.348$)     & -1.193 ($\pm 0.387$)      & 0.02725 & 0.9905 \\ \hline  
\end{tabular}
\vspace{2mm} \newline
By observing the parameters, a few remarks follow. First, we expect to find an approximating match between these values and those in \eqref{rational_Kommu}-\eqref{exp_Kommu}. This would signal some sort of agreement with the results of \cite{Kommu:2011wd}. We observe that this is the case for the exponential fit, however the disagreement of the parameters in the fractional fit is more significant 
\footnote{We thank R. Kommu for letting us know that the ensembles we examine are different from those analysed in \cite{Kommu:2011wd}.}.
Secondly, the best fit of \eqref{multi_fit} gives a negative value of the parameter $c$, which is supposed to be positive. If we restrict it to be positive the best fit gives $c=3\cdot 10^{-5}$, which is negligible compared to the second term in the denominator of \eqref{multi_fit} for $40\leq \s \leq250 $. Therefore, for both values of $c$ the functional  form of the ``multigraph'' fit \eqref{multi_fit} is effectively altered. 

To conclude, we return to figure \ref{Ds_my_analysis}. Both plots indicate that the ensembles we have analysed have limited volume because finite-size effects become dominant before discreteness effects vanish. This is the reason why we do not observe a genuine region where the spectral dimension increases without an oscillating pattern. This interpretation does not contradict the results in \cite{Kommu:2011wd}, since we have not been given the same ensembles. For this reason, we believe that the fitting procedure is not conclusive and one should repeat the above analysis for larger volumes of geometries.


\section{Lessons from other approaches beyond CDT}
\label{other_approaches}
The dynamical reduction of the spectral dimension is not an exclusive feature of CDT. It has been observed in other approaches to quantum gravity too (summarised nicely in \cite{Carlip:2009km, Carlip:2009kf}). In this section we review the basic features of other models stating the similarities and differences with CDT. 

\subsection{Asymptotic safety scenario}
\label{ASG}
The asymptotic safety scenario was conjectured by S. Weinberg \cite{Weinberg:1979ud} and refers to the possibility that gravity has a non-Gaussian UV fixed point where the theory is asymptotically safe but not free. In addition, the theory must possess a finite number of relevant operators otherwise it looses its predictive power. The basic idea for asymptotically safe gravity is that the couplings and the effective metric run with the energy scale $k$. Then one searches for the evolution of the parameters following exact renormalisation group (ERG) techniques \cite{Rosten:2010vm}. This line of research has accumulated good evidence for the existence of a non-Gaussian fixed point (NGFP) for a number of different truncations (see \cite{Niedermaier:2006ns, Reuter:2006as, Benedetti:2009rx} and references therein). It has also been argued recently that the UV-stable surface is finite in four dimensions for generic $f(R)$-action \cite{Falls:2013bv, Benedetti:2013jk}. These pieces of evidence indicate that the asymptotic safety scenario is a viable candidate for quantum gravity.

Asymptotically safe gravity can be considered as a complimentary picture to lattice regularisation methods, for example CDT, in the Wilsonian spirit. In this picture a continuum quantum field theory is associated with the lattice description of a system at a second-order phase transition. In this sense, the appearance of an UV second order phase transition in the CDT phase diagram, as explained in chapter \ref{motivation}, provides further evidence for the existence of the non-Gaussian fixed point (see \cite[section 1]{Ambjorn:2012jv}, \cite[section 3.3]{Ambjorn:2013tki} for further discussion). 

Soon after the first numerical evidence of a scale dependent spectral dimension \cite{Ambjorn:2005db}, another calculation by Launcher and Reuter determined the spectral dimension using ERG techniques in the Einstein-Hilbert truncation \cite{Lauscher:2005qz}. They reported a dynamical reduction of the spectral dimension which follows the relation
\beq
D_s(\s ; d) =
\begin{cases}
d, &\s \to \infty,\\
d/2 &\s \to 0.
\end{cases}
\eeq 
where $d$ is the topological dimension of space-time. Their argument goes as follows. Since the effective metric depends on the scale $k$, the Laplace-Beltrami operator in the diffusion equation \eqref{diffusion_eq} acquires an energy dependence too, i.e.   
\beq \label{F(k2)}
\Delta (k) = F(k^2) \Delta (k_0)
\eeq
where $F(k^2)$ is an interpolation function which relates the Laplace-Beltrami operator at a fixed reference scale $k_0$ and the scale $k$. This relation modifies the return probability density to
\footnote{Although one should, in principle, compute the average return probability density $\avg{P(\s)}_Z$, as explained in chapter \ref{motivation}, in asymptotically safe theories of gravity, the following approximation is used
\beq
\avg{\C{O}(\gamma_{\mu \nu})} \approx \C{O} (\avg{g_{\mu \nu}}_k)
\eeq
where $\gamma_{\mu \nu}$ is the microscopic metric close to the UV fixed point and $\avg{g_{\mu \nu}}_k$ is the average (smooth) metric which solves the effective field equations at scale $k$. The approximation is good given that the operator $\C{O}$ involves momenta of order $k$ \cite{Lauscher:2005qz}.} 
\beq \label{heat_trace_F}
P(\s) = \int \frac{d^dk}{(2\pi)^d} e ^{-k^2 F(k^2)\s}.
\eeq
It turns out that the interpolation function has a power law behaviour $F(k^2) \simeq k^{\delta(k)}$ along the RG trajectory. Substituting such scaling behaviour into \eqref{ds_def}, we are led to a spectral dimension of the form
\beq \label{Ds_ERG}
D_s(\s) = \frac{2d}{2+\delta}.
\eeq 
In the classical regime, $F(k^2) = 1$ and $\delta = 0$. Along the RG trajectory the running of $\delta$ exponent is given by
\beq
\delta(k) = 2 + \beta_{\lambda_k} (g_N(k), \lambda_k) / \lambda_k
\eeq
where we have defined the \textit{dimensionless} cosmological, $\lambda_k = \La k^{-2}$,  and Newton's constants, $g_N(k) \equiv G_N(k) k^{d-2}$,  and $\beta_{\lambda_k} (g_k, \lambda_k)$ is the $\lambda_k$'s beta function. Therefore, we get $\delta = 2$ at the non-Gaussian fixed point where the beta functions vanish. 

This result gives rise to the following remarks:
\begin{itemize}
\item[1.] The above result can be generalised to the full (untrucated) effective action in four dimensions, assuming that the theory possess a non-Gaussian fixed point.  In this  case, it can be shown by a general argument that $\delta = 2$ at the fixed point \cite{Lauscher:2005qz}.
\item[2.] In \cite{Lauscher:2001ya}, the authors gave another intuitive argument which also unveils the two-dimensional nature at short distances  as seen by a test-graviton propagator. Their argument was the first report of dynamical dimensional reduction from four in the IR to two in the UV regime, even though they did not go through the determination of the spectral dimension. It is instructive to briefly review this argument here \cite{Lauscher:2005qz, Niedermaier:2006ns}.    
Starting with a scale-dependent action, we note that the inverse Newton's constant $1/G_N(k)$ can be treated as a wave-function renormalisation of the metric, which according to renormalisation group (RG) arguments is related to the anomalous dimension of the coupling under the RG flow $\eta_N = \frac{\partial \ln G_N(k)}{\partial \ln k}$. Next, we introduce the dimensionless Newton's constant, $g_N(k) = G_N(k) k^{d-2}$. The beta function $\beta(g_N)$ of the the dimensionless Newton's constant, which describes the running of the coupling under the (RG) flow, takes the form
\beq
\beta(g_N) \equiv k \frac{\partial g_N}{\partial k} = (d-2+\eta_N(g_N, \ldots) )g_N.
\eeq
When $\beta(g_N^*) = 0$, the coupling becomes independent of the energy scale and the theory is said to be at a fixed point $g_N=g_N^*$. This is true when $\eta _N(g^*_N, \ldots) = 2-d$. Additionally the anomalous dimension of a field contributes to the propagator as $(1/p^2)^{1-\eta_N/2}$. Therefore, for $\eta_N = 2-d$, the propagator contributes a factor $p^{-d}$, which makes a typical integral $\int \frac{d^dp}{p^d}$ to diverge only logarithmically. However, such a logarithmic  behaviour is characteristic of a two-dimensional theory.

One should notice that the effective dimension extracted from this argument is always 2 independent of the topological dimension $d$. The argument is consistent with previous analysis only at $d=4$. The difference is in the fact that the graviton propagator is modified due to RG flow of the Newton's constant, which scales as $G_N \simeq k^{2-d}$ at the vicinity of the NGFP, whereas the reduction of the spectral dimension is due to the RG running of the cosmological constant which scales as $\La_k \simeq k^2$ near the NGFP \cite{Lauscher:2005qz}. An intriguing observation is that in the multigraph approximation we achieved the reduction of the spectral dimension due to the scaling/renormalisation of Newton's constant, as imprinted in assumption \eqref{assumption_i} or expression \eqref{pav-scale}, and not by the scaling of the cosmological constant in contrast to the asymptotic safety scenario. 

\item[3.] Expression \eqref{Ds_ERG} implies that in $d=3$ dimensions the spectral dimension flows from $3$ to $3/2$, which is inconsistent, at first sight, with the interpretation of the results coming from the computer simulations \cite{Benedetti:2009ge, Kommu:2011wd}. However, Reuter and Saueressig, working in the Einstein-Hilbert truncation, considered three scaling regimes for the interpolation function, $F(k^2)$; the classical, where $\delta = 0$, the semi-classical, where $\delta = d$ and the non-Gaussian fixed point regime, where $\delta =2$ \cite{Reuter:2011ah}. In such a scenario, the spectral dimension reduces from the value of $d$ in the IR, to a  value less than $d/2$ at semi-classical scales and then increases to the UV value $d/2$. In order to resolve the mismatch with three-dimensional CDT data, they argued that the Monte-Carlo data should not be extrapolated to the deep UV, since simulations do not probe the Planck scale. Further, they showed that the dimensional flow of their model is a good fit to the numerical data in the classical regime, but they further argued that simulations are inadequate to  probe the semi-classical and quantum regimes, where the spectral dimension is expected to fall to the value (less than) $d/2$, according to their description.  
Following similar arguments, the authors in \cite{Rechenberger:2012pm} derived qualitatively similar results in the $R^2$ truncation analysis.
\end{itemize}

\subsubsection{Relation to the multigraph model}
\label{Relation to the multigraph model}
One may observe an intriguing similarity between the expressions of the spectral dimension \eqref{ds_transient_rho} and \eqref{Ds_ERG}. Here, we comment on potential similarities and differences that arise between asymptotic safety and CDT-like models. At first sight, there are three apparent differences. The first expression relates the \textit{graph} spectral dimension to the Hausdorff dimension and the anomalous exponent of resistance, $\rho$, whereas the second one provides a relation among the \textit{scale dependent} spectral dimension, topological dimension and the running exponent $\delta(k)$. In \cite{Reuter:2011ah}, it was argued that the Hausdorff dimension equals the topological dimension at any scale in Einstein-Hilbert truncation.  In addition, giving scale dependence to $\rho$, under appropriate ansatz for the multigraph ensemble, both expressions can potentially describe the reduction of the spectral dimension. Therefore, it suffices to investigate the relation between the anomalous exponent $\rho$ and the running exponent $\delta(k)$.

To begin with, we recall some arguments from \cite{Reuter:2011ah}. The \textit{walk} dimension, $d_w$, is defined via the average square displacement of a random walk, i.e. $\avg{r^2}\sim \s^{2/d_w}$. A diffusion process is regular if $d_w=2$, whereas for $d_w \neq 2$ it is anomalous. For standard fractals, it is known that
\beq \label{ds-dw-dh}
\frac{D_s}{2} = \frac{\dha}{D_w}
\eeq
where $D_s$ and $D_w$ are changing with scale and $d_H$ is fixed. A random walk is characterised as recurrent if $D_w > \dha$ and non-recurrent or transient if $D_w<\dha$. From \eqref{Ds_ERG}, \eqref{ds-dw-dh} and the fact that $\dha=d$ we conclude that $D_w = 2 +\delta$. It becomes evident that $\delta$ controls the recurrent/transient character of the diffusion. 
In particular when $\delta = 0$ (IR), the random walk is transient in $d=4$. In the semi-classical regime where $\delta=4$, the random walk is recurrent, whereas in the NGFP regime, $\delta=2$, the diffusion is marginally recurrent.  
 
Let us recall the role of the exponent $\rho$ through the definition of resistance \eqref{lemma:dS:lower:conditions}, $R(N) \sim N^{2-\dha+\rho}$, for large $N$. As we have already explained in section \ref{resistance}, the random walk is transient (recurrent) when the resistance is finite (infinite). Thus, we note that $\rho$ also determines the recurrent character of the random walk with $\rho_c = \dha-2$ to be the critical value, similarly to $\delta$.

Despite the similarities there is a crucial difference between the two descriptions. On the one hand, asymptotic safety has $\dha=d$ at all scales and the walk dimension flows, implying a regular diffusion at large scales and an anomalous one along the RG trajectory. On the other hand, the random walk on the radially reduced CDT is always regular and assumption \eqref{assumption_i} implies that the Hausdorff dimension varies under scaling. We may bridge the two pictures by studying features of anomalous random walks on graphs and defining the continuum limit. We leave this study to future work.

\subsection{Ho\v{r}ava-Lifshitz gravity}
\label{HLG}
As we have repeatedly mentioned, Einstein's theory of general relativity is perturbatively non-renormalisable and therefore should be treated as an effective theory which breaks down at some energy scale. In an attempt to remedy this, physicists considered theories beyond general relativity, adding higher order curvature terms \cite{Stelle:1977hd}. In this case the theory restores perturbative renormalisability, but it also becomes non-unitary, due to the higher order time derivatives. However, if one added higher order spatial derivatives without adding higher order time derivatives one could achieve both renormalisability and unitarity. Obviously, such a construction breaks Lorentz invariance. 

In \cite{Horava:2009uw}, P. Ho\v{r}ava introduced a gravity model, where time and space have an anisotropic scaling, which is characterised by the (scale dependent) dynamical critical exponent $z$,
\beq
[x] = [k]^{-z}, \qquad [t] = [k]^{-1}.
\eeq
At large scales $z=1$ in order that the theory recovers Lorentz invariance, whereas at short distance $z=3$ so that the theory is power-counting renormalisable in $d=4$ dimensions. Following the ADM decomposition of the metric in standard general relativity, the dynamical variables become the lapse scalar, $N$, the shift vector, $N_i$, and the spatial metric, $g_{ij}$. 
This decomposition between space and time implies a preferred foliation of space-time. The symmetries of the theory must also respect this foliation structure and therefore the theory is not invariant under standard diffeomorphisms, but a more restricted set, the \textit{foliation preserving diffeomorphisms}. 

As a result of the new structure of space-time, the spatial component of the Laplace-Beltrami operator must be modified. To explain the difference we write the diffusion equation \eqref{diffusion_eq} as follows 
\beq \label{diffusion_eq_horava1}
\frac{\partial K_g({\bf y},\tau, {\bf y_0}, \tau_0,\s)}{\partial \s} = \left (\frac{\partial ^2}{\partial \tau^2} + \partial _i \partial _i \right ) K_g({\bf y},\tau, {\bf y_0}, \tau_0,\s)
\eeq 
where $\tau$ is Euclidean time. The anisotropic scaling does not affect the time component but modifies the spatial dimensions \cite{Horava:2009if}
\beq \label{diffusion_eq_horava2}
\frac{\partial K_g({\bf y},\tau, {\bf y_0}, \tau_0,\s)}{\partial \s} = \left (\frac{\partial ^2}{\partial \tau^2} + (-1)^{z+1}(\partial _i \partial _i )^z\right ) K_g({\bf y},\tau, {\bf y_0}, \tau_0,\s).
\eeq 

The modified diffusion equation \eqref{diffusion_eq_horava2} can be solved and implies that the return probability density takes the form
\beq
P(\s) \simeq \frac{1}{\s ^{(1+D/z)/2}},
\eeq
which implies under the definition of the spectral dimension \eqref{ds_def}
\beq
D_s(\s) = 1+ \frac{D}{z},
\eeq
where $D$ is the number of spatial dimensions. One observes that, in $d=3+1$, the spectral dimension varies from $4$ in the IR, where $z=1$, to $2$ in the UV, where $z=3$, being consistent with the result coming from CDT simulations. 

The aforementioned dynamical reduction in Ho\v{r}ava-Lifshitz theories is not the only common feature with CDT. The foliation structure of space-time also resembles the time-sliced structure of the triangulations imposed by CDT. Thus, it was soon realised that Ho\v{r}ava-Lifshitz gravity might be the continuum counterpart of CDT and studying the relation between the two theories is an active research program \cite{Ambjorn:2010hu, Anderson:2011bj, Sotiriou:2011mu, Ambjorn:2013joa}. In particular, it was argued in \cite{Ambjorn:2010hu} that the phase diagram of four-dimensional CDT matches qualitatively the phase diagram of an effective Lifshitz theory, providing extra evidence for the link between these two theories. In addition, three-dimensional  projectable Ho\v{r}ava-Lifshitz gravity was quantised using the CDT formalism, and computer simulations indicate the existence of an extended-geometry phase \cite{Anderson:2011bj}. Finally, very lately, it was proved that two-dimensional CDT, which is analytically solvable, shares the same continuum Hamiltonian obtained by quantising two-dimensional (projectable) Ho\v{r}ava-Lifshitz gravity \cite{Ambjorn:2013joa}.

\subsubsection{From spectral dimension to dispersion relation}
Three-dimensional gravity provides an interesting and fruitful playground to study the relation among different approaches to quantum gravity, as one can see from the above discussion and refs. \cite{Anderson:2011bj, Sotiriou:2011mu}.    
Since we now have an analytical model which fits successfully the CDT data, we further explore the relationship between CDT and Ho\v rava-Lifshitz gravity. In particular,  one can extract information about the dispersion relation from the spectral dimension and vice versa as shown in \cite{Sotiriou:2011aa}. We apply these methods to the spectral dimension \eqref{Ds_2+1}. 

In order to make contact with Ho\v rava-Lifshitz gravity, we start with a dispersion relation $\omega = \Omega (k)$ which is in general Lorentz-violating and define $\Omega (k)^2 = f(k^2)$. After some mathematical manipulations described in \cite{Sotiriou:2011aa}, one can relate the heat trace $P (\sigma)$ to the Laplace transform of the function $k(\Omega)$ in the variable $\Omega^2$, i.e.
\beq
\frac{D}{\sqrt{\s}}P(\s) \simeq  \int _{0}^{\infty} k(\Omega)^D e^{-\s \Omega^2} d\Omega^2.
\eeq
Therefore the inverse Laplace transform determines $k(\Omega)$
\beq \label{k(Omega)}
k(\Omega)^D \simeq \frac{1}{2\pi i} \int _{C} \frac{D}{\sqrt{\s}}P(\s) e^{\Omega^2\s}d\s.
\eeq

In our case, we substitute the heat trace \eqref{heat_trace_2+1} for $d=2+1$ gravity and get
\beq \label{inverse_laplace_result}
 k(\Omega)^2 \simeq \left (-1 + \frac{2}{\sqrt{\pi}} w  + e^{w^2}\text{Erfc}(w) \right ) /G_3^2
 \eeq
 where $w = \sqrt{2} G_3\Omega $ and $\text{Erfc}$ is the complementary error function, defined by 
 \beq \label{Erfc}
 \text{Erfc}(z) = 1- \text{Erf}(z) = \frac{2}{\sqrt{\pi}} \int_{z}^{\infty} e^{-r^2}dt.
 \eeq
One observes that it is rather complicated to invert \eqref{inverse_laplace_result} and find $\Omega(k)$. For this reason we focus our study in three regions where  \eqref{inverse_laplace_result} is simplified and can be inverted. 

First, we study the limit $w>>1$ or $\omega = \Omega(k) >>1/G_{3}$. In this limit the $\text{Erfc}(w)$ function has the following expansion
\beq \label{Erfc_w_infinity_expansion}
\text{Erfc}(w) \simeq e^{-w^2} \left ( \frac{1}{\sqrt{\pi}w} - \frac{1}{2\sqrt{\pi}w^3}+o(w^{-4})\right )
\eeq
and \eqref{inverse_laplace_result} admits the solution
\beq \label{w_infinity_expansion}
w(k) = \frac{\sqrt{\pi}G_3^2}{2} \left (k^2 +  c_1\right ),
\eeq
where all $c_i$'s are irrelevant constants.

Second, we study the region $w \simeq 1$ or equivalently $\Omega \simeq 1/G_{3}$. In this region the $\text{Erfc}(w)$ function has the expansion
\beq \label{erfc_w_unit_expansion}
\text{Erfc}(w)\simeq \text{Erfc}(1) - \frac{2(w-1)}{e \sqrt{\pi}} + O \left ((w-1)^2\right),
\eeq
where $\text{Erfc}(1) \simeq 0.157$ and  \eqref{inverse_laplace_result} has a solution of the form 
\beq \label{w_unit_expansion}
w(k) \simeq c_2 G_3^2 k^2+ c_3.
\eeq

Last, we consider the limit $w<<1$. The expansion of $\text{Erfc}(w)$ function in this limit has the form
\beq \label{erfc_w_zero_expansion}
\text{Erfc}(w) \simeq 1-\frac{2 w}{\sqrt{\pi}} + \frac{2 w^3}{3\sqrt{\pi}} + O (w^4)
\eeq
which implies
\beq \label{w_zero_expansion}
w(k) \simeq c_4 G_3 k.
\eeq

In 2+1 dimensions the inverse of the three-dimensional Newton's constant $1/G_{3}$ is equal to the three-dimensional Planck mass (in natural units). Thus, the three limits we considered  describe the dispersion relation,  $\Omega(k)$, for a test particle at energies much larger than, of order of and much less than the Planck mass respectively. In other words they correspond to trans-Planckian, Planckian and IR regimes respectively. Summarising the above results in one expression
\bea \label{dis_rel_summary}
\omega = \Omega (k) \simeq
\begin{cases}
k^2/G_3 + \text{const}/G_3, \qquad &k >> \frac{1}{G_{3}} = M_{P(3)},\\
k^2/G_3 + \text{const}/G_3, \qquad & k \simeq \frac{1}{G_{3}} = M_{P(3)}, \\
k, \qquad \qquad \qquad  &k << \frac{1}{G_{3}} = M_{P(3)}. 
\end{cases}
\eea 
We see that the dispersion relation appears to have common quadratic behaviour in both the trans-Planckian and Planckian limits while in the IR it increases linearly, as expected. 

We can now relate our results to the dispersion relation which originates from the foliation-defining scalar of $2+1$ Ho\v rava-Lifshitz gravity and  is given by \cite{Sotiriou:2011mu}
\beq \label{disp_rel_3dHL}
\omega(k) ^2 = A k^2\frac{1+B k^2+C k^4}{1+D k^2}.
\eeq 
It is readily seen that it has the same asymptotic behaviour as the dispersion relation \eqref{dis_rel_summary} when $k\to 0$ and $k\to \infty$. However, the fact that these two dispersion relations agree in the UV and IR limits is not surprising. It originates from the fact they both share the same UV and IR values of the spectral dimension. In principle, since we are equipped with an analytical CDT-like model which explains satisfactorily the numerical results, we can further study the relation between Ho\v rava-Lifshitz gravity and CDT-like models.

To conclude this section we summarise our findings. On one hand we argued that the multigraph approximation might provide a satisfactory fit to numerical data ignoring the bump effect. On the other hand $2+1$ Ho\v rava-Lifshitz gravity considered in \cite{Sotiriou:2011mu} also explains, under an appropriate adjustment of its free parameters, the $2+1$ CDT data including the ``bump". However, the four free parameters of the model \eqref{disp_rel_3dHL} in \cite{Sotiriou:2011mu} provide enough freedom to fit the numerical data without the ``bump" too. In this case we may adjust these parameters to fit the ``multigraph'' profile, offering a potential analytical link between the two models.

\subsection{Multi-fractional spacetimes}

The idea of the \textit{multi-fractional} nature of space-time, i.e. fractional geometry with multiple characteristic scales, has been extensively studied by G. Calcagni (see \cite{Calcagni:2012hb} and references therein for an introduction). Here we present only a few elements of the theory in order to introduce some terminology and notation and stress the similarities with the multigraph model. 

A fractional space-time is defined by an embedding into a Minkowski (or Euclidean) space-time $M^d$ and is equipped with a Lebesgue-Stieltjes (factorisable) measure,
\beq \label{LS_measure}
d \varrho _{\alpha}(x) = d^dx \upsilon _{\alpha}(x),
\eeq 
an appropriately modified Laplace-Beltrami operator and calculus  \cite{Calcagni:2012hb}. For example, the measure can take the form
\beq
\upsilon_{\alpha} (x) = \prod _{\mu}   \upsilon _{\alpha} (x^{\mu}) = \prod _{\mu} \frac{|x^{\mu}|^{a_{\mu}-1}}{\Gamma(a_{\mu})}
\eeq
where $\mu = 0, \ldots, d-1$ and the parameters $0<\alpha _\mu \leq1$ are the \textit{fractional charges}. It is customary to consider the ``isotropic'' case where $\alpha _{\mu} = \alpha$ for any $\mu$. Under the measure \eqref{LS_measure}, common integration is modified to
\beq
\int _{A}d^dx \to \int _A d\varrho_{\alpha} (x).
\eeq

Multifractional features arise when we consider multiple copies of (isotropic) fractional geometries with different measure, i.e. multiple fractional (isotropic) charges $\alpha _n, \ n=1, \ldots, N$. Now the diffusion equation is modified to incorporate the multifractional structure. We consider the simplest case where the diffusion equation becomes
\beq \label{multifrac_dif_eq}
\left (\frac{\partial}{\partial \s} - \sum _{n=1}^N \zeta_n(\ell) \C{K}_{\alpha_n}^E \right ) P(x, x', \s) = 0
\eeq
where $\C{K}_{\alpha_n}^E$ is the Euclidean fractional Laplace-Beltrami operator, describing a diffusion process which takes place on a space-time with $N-1$ characteristic scales $\ell_1< \ldots < \ell_{N-1}$. It has been shown that the effective fractional charge is given by \cite{Calcagni:2012qn, Calcagni:2012rm}
\beq
\alpha _{N-1}^{\text{eff}}(\ell) = \frac{1+\sum _{n=1}^{N-1}\zeta_n(\ell)\alpha _n}{1+\sum _{n=1}^{N-1}\zeta_n(\ell)}, \qquad \text{where} \qquad \zeta_{n} (\ell)= \left (\frac{\ell_n}{\ell-\ell_{n-1}} \right )^2.
\eeq
The particular example \eqref{multifrac_dif_eq} implies a simple profile for the spectral dimension flow
\beq
D_s(\ell) = d \cdot \alpha _{N-1}^{\text{eff}}(\ell).
\eeq

\subsubsection{Relation to the continuum random multigraph and comb}
For $N=2$, there is only one characteristic scale $\ell_1$ which separates the two regimes. The profile of the spectral dimension becomes for $d=4$ 
\beq \label{Ds_N=1_mfrac}
D_s(\ell) = 4 \cdot \frac{1+ \zeta_1(\ell) \alpha_1}{1+\zeta_1(\ell)}, \qquad \text{where} \qquad \zeta_1=\left( \frac{\ell_1}{\ell}\right)^2.
\eeq
We observe that for scales much larger than the characteristic scale $\ell_1$, we have $D_s(\ell >> \ell_1) = 4$. On the other hand, the short distance limit, i.e. $\ell<<\ell_1$ yields $D_s(\ell<<\ell_1) =4 \alpha_1$. Choosing $\alpha_1=1/2$, the dynamical dimensional reduction \eqref{Ds_N=1_mfrac} is consistent with both CDT computer simulations and the multigraph model. Additionally, we can make explicit contact with the spectral dimension profile \eqref{Ds_contin4} by writing expression \eqref{Ds_N=1_mfrac} as
\beq 
D_s(\ell) = 4 \left (1-\frac{1}{2+2 (\ell_1/\ell)^{-2}} \right )
\eeq
and rescaling the fictitious diffusion time  $\s = \ell^2 \bar \s$, where $\bar \s$ is a dimensionless parameter \cite{Calcagni:2012qn, Calcagni:2012rm}. As a result, the two profiles of the spectral dimension are identical and the two regimes are separated by one scale, which, in the multigraph model, obtains a physical interpretation as the renormalised Newton's constant.  
 
Next, we consider the case $N=3$, which implies two characteristic scales $\ell_1 << \ell_2$ and three plateaux in the profile of the spectral dimension. In  essence, the effective fractional charge becomes
\beq
\alpha_2^{\text{eff}}(\ell) = \frac{1+\zeta_1(\ell)\alpha_1+\zeta_2(\ell)\alpha_2}{1+\zeta_1(\ell)+\zeta_2(\ell)}.
\eeq
The spectral dimension is scale dependent with profile $D_s(\ell) = d \cdot \alpha_2^{\text{eff}}(\ell)$. Thus, we observe that there are three regimes, the long distance (or IR), where $\ell>>\ell_2>>\ell_1$ which implies $D_s = d$, the intermediate regime, where $\ell_2 >> \ell >>\ell_1$, which gives $D_s = d \alpha_2$ and the short scale regime (or UV), characterised by $\ell_2>>\ell_1>>\ell$, which results in $D_s = d \alpha_1$. We recall that the fractional charges $\alpha_i, \ i=1, \ 2$ can take any value in the interval $(0,1]$. In \cite{Calcagni:2012qn, Calcagni:2012rm}, the author chooses $\alpha _1= 1/2$ and $\alpha_2=1/3$ to reproduce the results from the asymptotic safety analysis \cite{Reuter:2011ah, Rechenberger:2012pm, Calcagni:2013vsa} (see also the third comment in section \ref{ASG}). 

The above choice of $\alpha_1, \alpha_2$ leads to a profile where the intermediate plateau has lower value than the UV and the IR values. Had we chosen $\alpha_1 < \alpha _2$, we would have obtained a monotonically increasing profile with one intermediate plateau, similar to the continuum random comb with double characteristic scales (section \ref{multiple_scales}).

\subsection{Loop quantum gravity, spin foams and others}
\label{LQG}
The evidence for a scale dependent spectral dimension does not stop here and extends to other proposals for quantum gravity which have fewer common features with CDT. For example, loop quantum gravity (LQG) also reports qualitatively similar results. In particular, the area spectrum of LQG, $A_j \simeq l_P^2\sqrt{j(j+1)}$, induces a scaling behaviour for the three-dimensional spatial metric, which modifies the Laplace-Beltrami operator and the heat trace similar to \eqref{F(k2)} and \eqref{heat_trace_F} respectively, with
\beq
F_{LQG} (k^2) = \sqrt{\frac{k^4(k^2_0+E_P^2)}{k^4_0(k^2+E_P^2)}}+1.
\eeq
The computation follows the same arguments as those exposed in the asymptotic safety scenario. In the context of LQG the spectral dimension of the spatial sector flows from the value of 3 in the IR to the value of 1.5 at intermediate scales and then increases to 2 in the deep UV limit \cite{Modesto:2008jz}. This behaviour resembles the intermediate semi-classical plateau reported in asymptotic safety. However this model suggests that the effective dimension at short distance is three-dimensional instead of 2 as observed in other approaches.  

Additionally, different four-dimensional spin-foam models have different area spectrums and therefore induce different scaling behaviour for the metric. For example, three models were considered in \cite{Modesto:2008jz, Modesto:2009kq} with area spectrums
\beq \label{spectrum_sf}
A_j = 
\begin{cases}
l_P^2 j, \\
l_P^2 (2j+1),\\
l_P^2 \sqrt{j(j+1)},
\end{cases}
\eeq
which modify the following scaling for the Laplace-Beltrami operator 
\footnote{The $+1$ term in expressions \eqref{Fsf} is added by hand by the author \cite{Modesto:2008jz, Modesto:2009kq}. He argues that this modification does not change the UV behaviour, but rather improves the IR limit allowing him to trade the $k\to k_0$ limit for the $k\to0$ limit in the calculations.}
\beq \label{Fsf}
F_{sf} (k^2) = 
\begin{cases}
k^2/k_0^2+1, \\
\frac{k^2(k_0^2+2E_P^2)}{k^2_0(k^2+2E_P^2)}+1, \\
\sqrt{\frac{k^4(k^2_0+E_P^2)}{k^4_0(k^2+E_P^2)}}+1.
\end{cases}
\eeq 
Under the first scaling, the spectral dimension  varies from four at large scales to two at Planck scales, in agreement with previous observations from other approaches. The second scaling implies that the spectral dimension is 4 at large distances, runs to 2 at scales of order of the Planck scale, and increases to 4 again in the deep UV limit, i.e. energies much larger than the Planck energy. Finally, the third scaling gives a scale dependent spectral dimension which reduces from 4 in the IR to 2 in the UV, but increases to $8/3$ in the deep UV limit.  

Further indications for a dynamical reduction of the spectral dimension have been found in the context of $\kappa$-Minkowski \textit{non-commutative space-time} \cite{Benedetti:2008gu} and quantum space-times with \textit{minimal length} \cite{Modesto:2009qc}.  In these approaches quantum effects modify the classical geometry, which requires new structure and calculus to describe it. Therefore the spectral dimension of the quantum space-time is not determined through an ensemble average but as a diffusion on the modified quantum geometry which is described with new calculus (which is reflected on the Laplace-Beltrami operator) and/or initial conditions \cite{Calcagni:2012rm}. 

The idea that space-time  is non-commutative at Planck scales is not new. In \cite{Benedetti:2008gu}, Benedetti studied the diffusion process on a non-commutative Minkowski space-time,   
and found a scale dependent spectral dimension which flows from 4 in the IR to 3 in the UV, which is in agreement with the LQG calculation but differs from other approaches. 

Finally, the authors in \cite{Modesto:2009qc} considered diffusion on a quantum space-time, which emerges as the average quantum geometric  fluctuations, with a minimal length scale and flat background metric.    
Due to the presence of the minimal length in the geometry, the initial condition \eqref{initial_condition} of the diffusion has to be modified into a Gaussian profile with width of minimal length. They reported a change of the effective dimension from 4 at distances much larger than the minimal length to 2 at scales of order of the minimal length.

\section{Conclusion and outlook}

The multigraph ensemble describes radially reduced four-dimensional CDT and provides some physical insight into the degrees of freedom which determine  the spectral dimension in the physical phase. Firstly the fact that the  behaviour of the return probability density \eqref{intro-P} implied by simulations \cite{Ambjorn:2005db} can be reproduced strongly suggests that the detailed structure of the spatial hypersurfaces is not important; it is the behaviour of the number of time-like edges $L_n$ which is crucial. To determine the spectral dimension on the multigraphs it is sufficient to know the volume growth and the resistance behaviour  reflected in the assumptions \eqref{assumption_i}-\eqref{assumption_iii}. These are all motivated from robust results in lower-dimensional studies \cite{Giasemidis:2012rf, Durhuus:2009sm}  but it is a non-trivial result that the continuum limit exists and that one can perform it exactly to obtain the return probability density \eqref{pav-scale} and show that there is a scale dependent spectral dimension varying from four at large scales to two at small scales.

These results show that the intuition about random walks on sliced graphs described earlier leads to a consistent picture in which the computer observations of a scale dependent spectral dimension can be related to the relatively simple question of the distribution of time-like edges in the CDT providing evidence that they could be a real continuum physical phenomenon rather  than a consequence of finite size effects. 

Having obtained a successful description for the numerical simulations of four-dimensional CDT, we proceeded by applying this formalism to three dimensions. The reason is that three-dimensional gravity might be simpler, but not trivial, and could become a useful field in the study of toy models of quantum gravity. In this case our results do not indicate an immediate agreement with the proposed fits to Monte Carlo output. After analysing the numerical data, we argued that the ``multigraph'' fit is consistent with computer simulations.  

Next, we discussed how the phenomenon of dynamical reduction is observed in other approaches to quantum gravity, most notably the asymptotic-safety scenario and Ho\v{r}ava-Lifshitz gravity. This common feature is the starting point for further exploration of the potential relations among these models, which appear to have deeper connections with CDT. Although there is agreement in four dimensions, the three-dimensional quantum gravity serves as a toy model to investigate the differences between those approaches. The three-dimensional model also gave us the opportunity to explore the relation to  Ho\v{r}ava-Lifshitz gravity. In addition, we stated that either asymptotic safety or Ho\v{r}ava-Lifshitz might serve as the continuum counterpart of CDT. This depends on whether the second-order transition point of CDT is isotropic or anisotropic respectively, as already discussed in chapter \ref{motivation}. We reported further indications of a running effective dimension coming from fundamentally different proposals for quantum geometry, where  a new description for the space-time is needed. 

Despite their differences, the approaches beyond CDT considered here follow the same strategy in computing  the spectral dimension. The first principles that characterise every model modify accordingly the Laplace-Beltrami operator (see for example equations \eqref{F(k2)}, \eqref{diffusion_eq_horava2} and \eqref{multifrac_dif_eq}) or the initial conditions which lead to a modified diffusion equation. The test-particle diffuses on an (averaged) effective geometry, which represents the quantum space-time. Among these models, some report plateaux of constant value of the spectral dimension and values less than 2 at some (intermediate) energy scale. 

These scenarios are in contrast to the diffusion process as considered in CDT and in multigraph approximation at two levels; first CDT does not single out only one representative averaged geometry, but considers diffusion on all possible configurations, following the ``sum over all histories'' prescription of quantum mechanics. In other words, the diffusion process in CDT and the multigraph model is a regular random walk on fractal geometry, in contrast to other approaches which consider anomalous diffusion process on smooth flat averaged geometries. Second, the modification of the diffusion equation might result in pathologies. In particular, the authors in \cite{Calcagni:2013vsa} showed explicitly that the solution to these modified diffusion equations is not positive semi-definite and the heat trace looses its probabilistic interpretation
\footnote{To remedy this, the authors introduced a new set of diffusion equations which restore the positivity of the  probability density without changing the profiles of the spectral dimension. The new modified diffusion equations include either non-linear diffusion time or non-trivial source terms and are constructed subject to the physical principles of each approach.}. 
It is important to underline that the multigraph approximation and the results coming from it do not suffer from this problem because we consider regular random walks on discrete geometries and then define the continuum limit. For these two reasons, we believe that the random walk on continuum multigraph ensembles captures the complete physical content of diffusion on the quantum geometry, because it is well-defined, from the mathematical point of view, and probes the entire quantum regime and its dynamics. 

Finally, the interpretation of the CDT-data differs among diverse approaches; the interpolation function in CDT is taken to be monotonically increasing from 2 to 4, without any intermediate plateaux. This interpretation has also been followed in our methods so far.

\end{chapter}

\clearemptydoublepage
\begin{chapter}{Summary and Outlook}
\label{conclusions}

Quantisation of gravity has led to significant advances in theoretical physics. The early attempts at quantisation unveiled that the theory is perturbatively ill-defined in four dimensions. To resolve the problem physicists have been considering several methods. One way to approach the problem in a unified manner is by embodying gravity in a extended theory with enlarged symmetry, e.g. superstring theory. Such treatment modifies the number of space-time dimensions. Other approaches attempt to define quantum gravity non-perturbatively exploiting conventional tools of quantum field theory. In these approaches dimensionality of space-time is not constant but varies dynamically. Both cases indicate that the number of space-time dimensions seems to be a crucial parameter in quantising the theory. 

A non-perturbative way to define the gravitational path-integral is through a lattice regularisation. This line of research has led the CDT approach to quantum gravity to be considered as a reliable model.  One of the observables that can be defined is the effective dimensionality of the fractal geometries that arise at the quantum regime. The spectral dimension is such an effective measure. Computer simulations of four-dimensional CDT unveiled an intriguing result, the spectral dimension of space-time varies dynamically from 4 in the classical limit to 2 in the quantum regime. However the interpretation of the numerical data is under constant debate. Despite the derivation of similar results from other approaches to quantum gravity, an analytical model was needed to systematically understand the mechanism of dimensional reduction within the CDT line of research. The main objective of this thesis is to present such an analytic approach.

We started, in chapter \ref{motivation}, by reviewing the CDT approach to quantum gravity and introducing the phenomenon of scale dependent spectral dimension. To accomplish our goal we use the (unconventional for quantum gravity research) toolbox of random infinite graphs. In chapter \ref{graphs} we introduced the relevant definitions and argued why the theory of random infinite graphs is essential in discretised models of quantum gravity. The generating function technique for the determination of the spectral dimension is the principal method in our research. Further, we introduced the Hausdorff dimension and the notion of graph resistance, which both encode further characteristics of graphs. We also discussed the relationship between tree graphs and branching processes and elaborated on the theoretical aspects of the latter. We concluded this mathematical introduction by underlining the importance of random trees with a unique infinite spine, i.e. the critical generic Galton-Watson tree conditioned on non-extinction.

In chapter \ref{combs} we exploited the simplicity of random combs to investigate the continuum limit of the generating function of return probabilities. To accommodate the varying spectral dimension we introduced a characteristic length scale in the ensemble measure. Intuitively, random walks shorter than this scale experience a different graph structure and thus a different spectral dimension than longer random walks. We formally defined the notion of short and long walk lengths and then verified analytically that there are indeed ensembles of random geometries for which $i)$ the continuum limit of the generating function can be rigorously defined and $ii)$ a scale dependent spectral dimension emerges through the continuum formalism. We examined three comb ensembles and all of them exhibit the reduction of the spectral dimension from a value greater than one at large distances to one at short scales. 
The simplicity of random combs enabled us to study the mathematical subtleties in examining the continuum limit of discrete random geometries and defining a running spectral dimension of continuum graphs. However, they lack physical content and are not sufficient to describe the dynamical dimensional reduction observed in the computer simulations. 

It becomes evident that we should go beyond random combs to more realistic graphs which have a dynamical/local growth law. One possibility would to be to study the generic random tree (GRT) because of its bijection to uniform infinite causal triangulations (UICT). 
However the fractal structure of trees is essentially different from that of causal triangulations, which is imprinted in the different values of spectral dimension. This signals that GRT does not encode the right degrees of freedom to describe the spectral dimension of CDT. Another useful mapping is the multigraph approximation, which we studied extensively in chapter \ref{multigraphs}. We argued that these radial reduced ensembles play a significant role in the proof of the spectral dimension of the UICT and there is accumulated evidence that both ensembles share the same value for the spectral dimension. We started with the random recurrent multigraph, because its measure is subject to analytical control. The offspring probability of the generalised uniform Galton-Watson process induces a measure on the GRT, generalised UICT and random recurrent multigraph too. We applied the continuum formalism to the recurrent multigraph and proved that the spectral dimension varies from 2 at large scales to 1 at small distances. We interpreted this result as the dynamical dimensional reduction of two-dimensional CDT which has an absolute value of the curvature term in the action.

Another advantage of the multigraph approximation is that it is valid for higher-dimensional CDT as well. In this case the random walk becomes non-recurrent. Before applying the continuum formalism to non-recurrent multigraphs we explored their properties because little was known about them. The main result of this study is given by Theorem \ref{Theorem:dS-dH}, which tells us that the spectral dimension is only related to the volume growth, through the Hausdorff dimension, and the resistance growth, through the anomalous exponent of graph resistance, $\rho$. In addition, the spectral dimension of transient multigraphs equals the Hausdorff dimension if and only if $\rho=0$. Higher-dimensional CDT is not subject to analytical results and the measure of the corresponding multigraph ensembles cannot be determined in contrast to two-dimensions. However, in order to apply the continuum formalism we need to specify a few characteristics of the ensemble. We argued why these characteristics must be related to the volume and resistance growth. The lessons and experience from the UICT guided us to adopt the exact form of our assumptions. A key point in our argument is that the fluctuations of spatial hyper-surfaces are bounded from above similarly to the UICT, which is analytically proven, and as observed numerically in computer simulations of higher-dimensional CDT. Having gained information about the ensemble of radially reduced four-dimensional CDT we applied the continuum limit and found that the spectral dimension varies from 4 at long distances to 2 at short scales. 

In chapter \ref{physics}, we focused on the physical implications of our methods. By applying  a Tauberian theorem we were able to determine the ensemble average of  return probability of discrete random walks. Scaling the latter and taking the continuum limit we found the ensemble average return probability density of continuous diffusion, expression \eqref{pav-scale}, which has the same functional form with the one conjectured purely from numerical data. This is the main result of this chapter and one of the main conclusions of the thesis. In other words, our intuition about the multigraph approximation  and its validity in higher dimensions has been substantiated. The next step was to study three-dimensional CDT. We adjusted our assumptions to the three-dimensional model and repeated our formalism. The dynamical dimensional reduction from 3 in the classical regime to 2 in the UV is consistent with numerical results. We found the functional form of the reduction of the spectral dimension, which does not agree with the best fits on the numerical data reported in the literature. We applied the functional form derived from the multigraph approximation to the data-points. The fit is qualitatively good and the residuals are within the error-bars of the data-points. This result increases our confidence in the validity of the multigraph approximation and of our assumptions. Put differently, the agreement with the numerical results, assures us that the multigraph approximation carries those degrees of freedom which are responsible for the dynamical reduction of the spectral dimension observed in computer simulations of CDT. On the other hand, the approximation cannot replace the dynamics of the full CDT, since many (spatial) degrees of freedom have been integrated out. However it is a good approximation to better understand the mechanism behind the reduction of the spectral dimension as observed in computer simulations.

An intriguing fact about the dynamical dimensional reduction is that it has been verified by other approaches to quantum gravity too. The physical reason for this similarity is that the reduction of the effective dimension might account for the regularisation of the theory in the UV limit. So, it is reasonable in some sense that non-perturbative approaches with no extra symmetries or degrees of freedom present such a mechanism. We introduced the underlying elements of some approaches. We also investigated the potential similarities and differences between these approaches and the CDT and/or the multigraph model. We commented that our formalism is well-defined and probes the full quantum regime and its dynamics in contrast to most continuum approaches which consider a modified diffusion process, which might lead to an ill-defined probabilistic interpretation, on a single averaged flat geometry. 
\vspace{4mm}

Our methods and results give rise to new research projects in two main directions. One direction is the field of mathematical physics. Within this field there is great interest and activity on the fractal properties of random (infinite) graphs and surfaces. For example, mathematicians use probabilistic techniques to determine the spectral dimension of trees, e.g. the invasion percolation tree \cite{Omer:2008ip} or the incipient cluster on trees \cite{Barlow:2006rw}. The generating function method, defined via \eqref{ds_via_Q}, might be applied to these problems that would be of interest to mathematicians too. Beyond the spectral dimension one may also explore the properties of the causal triangulations which are in bijection with these new types of random trees.

Our results might be useful and find applications within the community of quantum gravity too. In particular, while our study only requires the averages of functions of $L_n$ it indicates that further light may be shed on the mechanisms of dynamical dimensional reduction in four-dimensional CDT by investigating the distribution and correlations of the $L_n$ in the numerical simulations. Understanding these distributions would help towards an analytical solution of the full four-dimensional model. The analytical control on the numerical data might also open up new ways to bridge CDT with continuum non-perturbative approaches. Although we have made a step towards this direction there is much more to be done. As we have already commented, one can methodically study the relation between asymptotic safety and models of discretised causal quantum gravity by studying anomalous random walks on ``flat'' multigraphs. Another, more ambitious, research program would be to study the cosmological consequences of the dynamical reduction of the effective space-time dimension. Inspired by the fact that this phenomenon is shared by a plethora of approaches to quantum gravity, it is interesting to investigate the consequences of this mechanism on quantum cosmology and particularly on inflationary models.

\vspace{4mm}

Quantum gravity is an inconclusive and fascinating field of study, which enhances physics with new tools from mathematics and gives rise to diverse and/or similar theoretical phenomena. This latter apparent paradox is at the core of scientific progress. To this respect, we conclude this thesis by quoting P. Bergmann \cite{Rovelli:2004qg}: 
``\textit{In view of the great difficulties of this program, I consider it a very positive thing that so many different approaches are being brought to bear on the problem. To be sure, the approaches, we hope, will converge to one goal.\textit}"

\end{chapter}
\clearemptydoublepage

\appendix
\begin{chapter}{}
\label{Appendix_combs}
\vspace{-1cm}
\section{Generating functions of basic combs}
The first part of the appendix consists of supplementary material to chapter \ref{combs}. For further details we refer the reader to \cite{Durhuus:2005fq}.

We begin with the proof of \eqref{mG_bounds}. We first aim to write the modified two-point function, $G_C^{(0)}(x;n)$, as a product of first return generating functions of random walks restricted to not reach vertex $s_n$, similarly to expression \eqref{G_decomp_P}. 
We denote the set of walks that contribute to $G_C^{(0)}(x;n)$ by $\Omega ^{(0)}$. 
We now decompose $\Omega ^{(0)}$ into a sequence of $n-1$ random walks $\Omega ^{(0)}_k, \ k=1,\ldots, n-1$, which go at most as far as vertex $s_{n-1}$, and a final step from vertex $s_{n-1}$ to $s_n$. $\Omega ^{(0)}_k, \ k=1,\ldots, n-1$, is a random walk from $s_{k-1}$ to $s_{k}$ which is identical to the part of $\Omega ^{(0)}$ which leaves vertex $s_{k-1}$ for the last time going to $s_{k}$, returns to $s_k$ multiple times until it leaves $s_k$ for the last time. Adding a last step to the random walk $\Omega ^{(0)}_k, \ k=1, \ldots, n-1$, back to the vertex $s_{k-1}$, we reconstruct a random walk which returns to the vertex $s_{k-1}$ for the first time without visiting vertex $s_{n}$. This is equivalent to the first return random walk, $\Omega^{<n-k-1}$, which starts from the root of the truncated comb $C_{k-1}$ and does not reach vertex $n-k-1$ of $C_{k-1}$. We denote the corresponding generating function $P_{C_{k-1}}^{(\Omega^{<n-k-1})}(x)$. The extra step contributes a factor of $\sqrt{1-x}/\sigma(k)$, hence we divide out by the same amount. Finally, when the random walk leaves vertex $s_{n-1}$ for the last time has only one possibility, to step to vertex $s_{n}$, contributing $\sqrt{1-x}/\sigma(n-1)$ to $G^{(0)}(x;n)$. Therefore we write
\bea \label{G0_decomp_P}
G^{(0)}(x;n) &=&  \frac{\sqrt{1-x}}{\sigma(n-1)}\prod _{k=0}^{n-2} \frac{P_{C_k}^{(\Omega^{<n-k})}(x)/\sigma(k)}{(1-x)^{1/2}/\sigma(k+1)} \nn\\
                   &=& (1-x)^{-(n-2)/2}\prod_{k=0}^{n-2}P_{C_k}^{(\Omega^{<n-k})}(x).
\eea

Next we use a slight modification of Lemma \ref{MonoLem2} which states that the generating function $P_{C_k}^{(\Omega^{<n-k})}(x)$ is a decreasing function of the length of the teeth, $\ell_j, j\geq1$,~i.e.
\beq \label{P_restr_bounds}
P_{*_k}^{(\Omega^{<n-k})}(x) \leq P_{C_k}^{(\Omega^{<n-k})}(x) \leq P_{\infty_k}^{(\Omega^{<n-k})}(x).
\eeq
Considering the decomposition \eqref{G0_decomp_P} and the bounds \eqref{P_restr_bounds} we end up with \eqref{mG_bounds}. 

Further generating functions of the half-line ($C=\infty$) are presented in Appendix \ref{Appendix_multigraphs} derived in the spirit of multigraphs. Applying \eqref{PhalflineLf} to \eqref{G0_decomp_P}, we also derive the modified generating function of the half-line 
\beq \label{GinfL}
G_\infty^{(0)}(x;n)=(1-x)^{n/2}\frac{2\sqrt{x}}{(1+\sqrt{x})^{n}-(1-\sqrt{x})^{n}} 
\eeq
which is a strictly decreasing function of $n$.    

We also need to determine the first return generating function for the comb with teeth of length $\ell$ equally spaced at intervals of $n$, denoted by $P_{\ell, *n} (x)$. The result is obtained by decomposing the walks contributing to $P_{\ell, *n} (x)$ into two sets; $\Omega_1 \equiv \Omega^{<n} $ which consists of walks which do not reach vertex $s_n$ and $\Omega_2$  which consists those walks that move beyond $s_n$,
\beq
P_{\ell, *n} (x) = P_\infty^{(\Omega^{<n})}(x) + \frac{\left(G_\infty^{(0)}(x;n)\right)^2}{3 - P_{\ell}(x) - P_{\ell, *n} (x)  -  P_\infty^{(\Omega^{<n})}(x)}.
\eeq
Solving with respect to $P_{\ell, *n} (x)\leq1$, it yields
\bea
P_{\ell, *n} (x)= \frac{3-P_{\ell}(x)}{2} - \half \left[\left(3-P_{\ell}(x)-2P_\infty^{(\Omega^{<n})}(x)\right)^2-4G_\infty^{(0)}(x;n)^2\right]^\half\label{Pelln}
\eea
and $P_{*n}(x)$ is obtained by setting $\ell=\infty$ in this formula. Also $P_{\ell, *n} (x)$ is a strictly decreasing function of $\ell$ and increasing function of $n$, viewed as continuous positive semi-definite real variables.

Finally, the continuum limit of the following generating functions under the scaling $x=a\xi$ and $\Lambda = a^{-\Delta}\lambda^{\Delta}$ is essential for obtaining $\tilde Q(\xi ; \lambda)$ and is given by
\begin{align} 
\label{G0_cont}
\lim_{a\to 0} \; \; & a^{-\half}G_\infty^{(0)}\left(x=a\xi; n=a^{-\half}\rho_1 \right)=\xi^\half \mathrm{cosech}\left (\rho_1\xi^\half \right), \\
\label{1-P_cont}
\lim_{a\to 0} \; \; & a^{-\half}\left(1-P_{(\ell=a^{-\half}\rho_2), *(n=a^{-\half}\rho_3)} (x=a\xi)\right)=-\half\xi^\half\tanh(\rho_2\xi^\half) \nn\\ 
&\qquad \quad  +\half\xi^\half\left[4+4\tanh\rho_2\xi^\half\coth\rho_3\xi^\half+\tanh^2\rho_2\xi^\half\right]^\half,
\end{align}
where $\ell, n$ are functions of $x, \La$ and $\rho_i, \ i=1,2,3$, are functions of $\xi, \lambda$, depending on the choice of $\tilde H, \tilde D$ and $\tilde k$.  
\end{chapter}

\begin{chapter}{}
\label{Appendix_multigraphs}

\vspace{-1cm}
\section{Basic solvable examples}
\label{Solvable} 

We  give two exactly solvable examples to illustrate some of the features derived in chapter \ref{multigraphs}.

\subsection{Recurrent case: spectral dimension of the half line}
First we present some results for the half line, known in the literature \cite{Durhuus:2005fq}, from the multigraph point of view, that is, we consider 
it as a special multigraph with 
$\{L_k=1,k=0,1,...\}$. As we see in Lemma \ref{monotonicity} the half line plays an important role in providing certain upper bounds for the return probability on any multigraph.

Since the random walk on the half line has to leave the root with probability one and otherwise can move to either neighbour  with probability $1/2$  the generating function for the first return probability \eqref{P_recurr_mgraph} satisfies
\beq \label{Phalfline}
P_\infty(x) =\frac{1-x}{2-P_\infty(x)}.
\eeq
which agrees with \eqref{P_half_line}, as expected. From this we get that $P_{\infty}(x) = 1- \sqrt{x}$ and $Q_{\infty}(x) = x^{-1/2}$,  
which diverges as $x\!\to\! 0$ 
and the spectral dimension is $d_s=1$. From the multigraph point of view, one also has the trivial result that $|B(N)|=\sum^N_{k=0} L_k=N+1$ and thus that $d_s=\dha$.

If  instead of the half line we consider a line segment of length $\ell$, then \eqref{P_recurr_mgraph} becomes
\beq \label{PhalflineL}
P_\ell(x) =\frac{1-x}{2-P_{\ell-1}(x)}
\eeq
for $\ell\geq 1$ and $P_0(x)=1$. This relation can be iterated to give 
\cite{Durhuus:2005fq}
\beq \label{Psegment}
P_\ell(x) =1- \sqrt{x} \frac{(1+\sqrt{x})^\ell - (1-\sqrt{x})^\ell}{(1+\sqrt{x})^\ell + (1-\sqrt{x})^\ell}.
\eeq
Similarly the contribution to the first return probability on the full half line from walks that do not extend beyond $N$ is 
\beq \label{PhalflineLf}
P^{(\Omega^{<N})}_{\infty}(x) =1- \sqrt{x} \frac{(1+\sqrt{x})^N +(1-\sqrt{x})^N}{(1+\sqrt{x})^N - (1-\sqrt{x})^N}.
\eeq

\subsection{Non-recurrent case: spectral dimension of a multigraph with $L_k\sim k^2$}
As an explicit example of a non-recurrent graph we consider a (fixed) multigraph $M=\{L_k,k=0,1,...\}$ with $L_k =(k +1)(k+2)$. This particular multigraph is rather special as we will see in the following. Note that the probability for a random walker at vertex $k+1$ of $M$ to go forward is 
\bea
p_{k+1} = \frac{L_{k+1}}{L_{k} +L_{k+1}}= \frac{k+3}{2(k+2)}
\eea
and the probability of returning from $k+2$ to $k+1$,
\bea
q_{k+1} =1- p_{k+2} = \frac{L_{k+1}}{L_{k+1} +L_{k+2}}= \frac{k+2}{2(k+3)}.
\eea
We observe that 
 \bea
p_k q_k = p_k (1-p_{k+1}) =\frac{k+2}{2(k+1)}\frac{k+1}{2(k+2)} = \frac{1}{4} 
 \eea
as was the case for the half line. Hence we can relate the first return generating function for a random walker on the multigraph $M_k$  to that on half line by just compensating for the last step of the random walk which on the half line would occur with probability $1/2$ while on $M_k$ it occurs with probability $q_k$ leading to
 \beq
P_{M_k}(x) = 2 q_k P_{\infty}(x).
 \eeq
It follows that
\bea
\label{etak2}
\eta _{M_k}(x) \equiv \frac{Q_{M_k(x)}}{L_k}= \frac{1}{L_k}\frac{1}{1-P_{M_k}(x) } =\frac{1}{(k+1)(1+(k+1)\sqrt{x})}.
\eea

We note that $\eta_k(0)=1/(k+1)$ is finite which shows that the random walk is non-recurrent. The first derivative is
\beq
-\eta _{M_0}'(x) = \frac{1}{2(1+\sqrt{x})^2 \sqrt{x}}\sim x^{-1/2} \quad \textrm{as}\quad x\to 0
\eeq
and hence the spectral dimension is $d_s=3$ while  $|B(N)|=\sum^N_{k=0} L_k\simeq N^3$  and thus $d_s=\dha$. It is straightforward to check that the resistance exponent $\rho=0$.

\section{Simple results for Lemma \ref{constraints} }\label{Simple}
Here we outline the proofs of Lemma \ref{constraints} assuming unless otherwise stated that $\dha$ exists and  $N>N_0$.
To prove that $\rho\ge 0$ note that 
\bea  \eta_{N}(0)	=\sum_{n=N}^{\infty} \frac{1}{L_n}
				>\sum_{n=N}^{2N} \frac{1}{L_n}
				>\frac{(2N-N)^2}{\sum_{n=N}^{2N} L_n}
				\sim {\rm const}\,N^{2-\dha},   \label{etalwrbnd} \label{A:result1}\eea
where we have used Jensen's inequality. To prove that $\delta'\ge 0$ note that 
\bea  B_N^{(2)}&=&\left(\sum_{k=0}^{N_0}+\sum_{k>N_0}^{N}\right) L_k\sum_{m=k}^\infty \frac{1}{L_m}\sum_{n=k+1}^\infty \frac{1}{L_n}\nonumber\\	
 &=& \textrm{const} + \sum_{k>N_0}^{N} L_k\sum_{m=k}^\infty \frac{1}{L_m}\sum_{n=k+1}^\infty \frac{1}{L_n} \eea
 and then use \eqref{etalwrbnd}; similar arguments show that $\gamma\ge0$ and $\delta\ge0$.
 Using Cauchy-Schwarz inequality
\beq
\left(\sum_{k=0}^{N-1} L_k\eta_{k+1}(0) \right)^2 < N^{\dha}\sum_{k=0}^{N-1} L_k\eta_{k+1}(0) ^2\eeq
which gives $\delta'\ge 2\gamma$.

To establish the relation between the other exponents and $\rho$, assuming it exists, we need the inequalities
%
that for any graph $M\in \C{M}$
\bea  
\! \! \! \! \! \! \! \! \! \! &(i)&\! \! \! \! \! \! \! \! \quad \eta_{N+1}(0) \left |B_N\right |< B_N^{(1)} < B_{N_1}^{(1)}+ \! \! \! \sum_{r=1}^{\lceil \log_2\frac{N}{N_0}\rceil }  \! \!  \eta_{\lceil N/2^r \rceil }(0)\left (\left |B_{\lceil N/2^{r-1}\rceil}\right | - \left | B_{\lceil N/2^r \rceil}\right |\right ) \label{B1:upper}   \\
\! \! \! \! \! \! \! \! \! \! &(ii)&\! \! \! \! \! \! \! \! \quad \frac{\left(B_N^{(1)}\right)^2}{\left | B_N \right |}<B_N^{(2)} < B_{N_1}^{(2)}+ \! \! \! \sum_{r=1}^{\lceil \log_2\frac{N}{N_0}\rceil }   \eta_{\lceil N/2^r \rceil}(0)(B_{\lceil N/2^{r-1}\rceil }^{(1)} -B_{\lceil  N/2^r \rceil}^{(1)})  \label{B2:upper}
 \eea
for $N_0\le N_1\le 2 N_0 < N$.
The proofs exploit the fact that $\eta_k(0)$ is a decreasing sequence. For example to prove $(i)$  we have the lower bound
\bea 
B_N^{(1)}>\eta_{N+1}(0)\sum_{k=0}^{N}L_k=\eta_{N+1}(0) | B_N |.
\eea
An upper bound is given by
%
\bea B_N^{(1)}	&=&\sum_{k=0}^{\lceil N/2 \rceil} L_k\eta_{k+1}(0)+\sum_{k=\lceil N/2 \rceil+1}^N  L_k\eta_{k+1}(0)\nn \\
			&<&B_{\lceil N/2 \rceil}^{(1)}+\eta_{\lceil N/2 \rceil}(0)\left (\left | B_N \right | - \left |B_{\lceil N/2\rceil}\right | \right )
\eea
and iterating to get \eqref{B1:upper}.  The proof of (ii) proceeds analogously.
Assuming that  $\rho$ exists then it follows from \eqref{B1:upper} and \eqref{B2:upper} respectively that
\beq \gamma=\rho,\quad \delta'=2\rho.\eeq
By definition $\delta\le \delta'$ and noting
 that 
\bea   \overline B_N^{(2)} 		&>& \sum_{k=0}^{\lfloor N/2 \rfloor }L_k\left(\sum_{n>\lfloor N/2\rfloor}^N \frac{1}{L_n}\right)^2\\
						&=&c \left(\frac{N}{2}\right)^\dha\left (\eta_{\lfloor N/2\rfloor}(0)-\eta_N(0)\right)^2\\
						&=& c' N^{4-\dha+2\rho}\eea
we also have  $\delta\ge2\rho$ so conclude that $\delta=2\rho$.

\end{chapter}
\clearemptydoublepage
\addcontentsline{toc}{chapter}{Bibliography}
\bibliography{random_graphs,spectral_dimension}        
\bibliographystyle{utphys}                                                     
\clearemptydoublepage

\newpage\thispagestyle{empty}\beq\phantom{page blanche}\nn\eeq
\end{document}